\documentclass[10pt,journal]{IEEEtran}
\IEEEoverridecommandlockouts
\usepackage{cite}
\usepackage{amsmath,amssymb,amsfonts}
\allowdisplaybreaks[4]
\usepackage[linesnumbered,ruled,vlined]{algorithm2e}
\usepackage{textcomp}
\usepackage{orcidlink}
\usepackage{xcolor}
\usepackage{balance}
\usepackage{microtype}
\usepackage{graphicx}
\usepackage{booktabs} 
\usepackage{hyperref}
\usepackage{makecell}
\usepackage{multirow}
\usepackage{colortbl}
\usepackage{pdfpages}
\usepackage{tikz}
\usepackage{enumitem}
\usepackage{adjustbox}
\usepackage{stfloats}
\usepackage{float}
\usepackage[caption=false]{subfig} 

\newtheorem{theorem}{Theorem}
\newtheorem{proposition}{Proposition}
\newtheorem{lemma}{Lemma}

\newtheorem{definition}{Definition}
\newtheorem{assumption}{Assumption}
\newtheorem{remark}{Remark}

\hypersetup{colorlinks=true,linkcolor=black,citecolor=black,urlcolor=black}

\newenvironment{proof}{{\bfseries\itshape Proof}.\ }{\hfill $\blacksquare$\par} 
\newcommand{\refappendix}[1]{\hyperref[#1]{Appendix~\ref*{#1}}}

\def\BibTeX{{\rm B\kern-.05em{\sc i\kern-.025em b}\kern-.08em
    T\kern-.1667em\lower.7ex\hbox{E}\kern-.125emX}}
    
\begin{document}

\title{A Game-Theoretic Framework for Privacy-Aware Client Sampling in Federated Learning}

\author{Wenhao~Yuan\orcidlink{0009-0001-6625-7496},~\IEEEmembership{Student Member, IEEE}, Xuehe~Wang\orcidlink{0000-0002-6910-468X},~\IEEEmembership{Senior Member, IEEE}

\IEEEcompsocitemizethanks{\IEEEcompsocthanksitem Wenhao Yuan and Xuehe Wang are with the School of Artificial Intelligence, Sun Yat-sen University, Zhuhai 519082, China. Email Address: \href{mailto:yuanwh7@mail2.sysu.edu.cn}{yuanwh7@mail2.sysu.edu.cn}, \href{mailto:wangxuehe@mail.sysu.edu.cn}{wangxuehe@mail.sysu.edu.cn}.
}
\thanks{This work is supported by the National Natural Science Foundation of China (NSFC) under Grant No. 62206320 and Guangdong Basic and Applied Basic Research Foundation under Grant No. 2024A1515010118. (Corresponding author: Xuehe Wang.)}
}


\maketitle

\begin{abstract}
In federated learning (FL) systems, the central server typically samples a subset of participating clients at each global iteration for model training. To mitigate privacy leakage, clients may insert noise into local parameters before uploading them for global aggregation, leading to FL model performance degradation. This paper aims to design a \underline{P}rivacy-aware \underline{C}lient \underline{S}ampling framework in \underline{FED}erated learning, named FedPCS, to tackle the heterogeneous client sampling issues and improve model performance. First, we obtain a pioneering upper bound for the accuracy loss of the FL model with privacy-aware client sampling probabilities. Based on this, we model the interactions between the central server and participating clients as a two-stage Stackelberg game. In Stage \uppercase\expandafter{\romannumeral1}, the central server designs the optimal time-dependent reward for cost minimization by considering the trade-off between the accuracy loss of the FL model and the rewards allocated. In Stage \uppercase\expandafter{\romannumeral2}, each client determines the correction factor that dynamically adjusts its privacy budget based on the reward allocated to maximize its utility. To surmount the obstacle of approximating other clients’ private information, we introduce the mean-field estimator to estimate the average privacy budget. We analytically demonstrate the existence and convergence of the fixed point for the mean-field estimator and derive the Stackelberg Nash Equilibrium to obtain the optimal strategy profile. Through rigorously theoretical convergence analysis, we guarantee the robustness of our proposed FedPCS. Moreover, considering the conventional sampling strategy in privacy-preserving federated learning, we prove that the random sampling approach's price of anarchy (PoA) can be arbitrarily large. To remedy such efficiency loss, we show that the proposed privacy-aware client sampling strategy successfully reduces PoA, which is upper bounded by a reachable constant. To address the challenge of varying privacy requirements throughout different training phases in FL, we extend our model and analysis and derive the adaptive optimal sampling ratio for the central server. Experimental results on different datasets demonstrate the superiority of FedPCS compared with the existing state-of-the-art FL strategies under IID and Non-IID datasets.
\end{abstract}


\section{Introduction} 
With the burgeoning development of machine learning (ML) techniques, the Internet of Things (IoT), and their networking applications, the exponential growth of data generated at edge devices presents vast potential for Artificial Intelligence (AI) and facilitates the applications in various fields \cite{hu2022incentive, wang2024novel}, such as facial recognition \cite{zhang2024validating}, disease diagnosis \cite{chen2024disentangle} and embodied intelligence \cite{long2023human}. However, the performance of the conventional centralized machine learning (CML) paradigm is restricted by the limited wireless communication bandwidth for data transmission \cite{luo2024adaptive}. Moreover, the heightened awareness of privacy protection and the enactment of stringent data privacy regulations such as GDPR and CCPA \cite{pardau2018california} poses a formidable challenge for distributed devices in accessing and sharing data. 

In tackling the challenges and safeguarding the confidential information of participants in ML systems, \emph{Federated Learning} (FL) \cite{mcmahan2017communication}, served as an emergent and preeminent distributed machine learning (DML) paradigm for secure model training, has been proposed, which enables geographically dispersed clients to collaboratively refine a shared model orchestrated by a central server while keeping private data localized \cite{luo2022tackling}, achieving multi-fold performance benefits including privacy preservation and communication efficiency \cite{he2023three}. 

To reduce the huge communication burden and coordinate geographically distributed mobile devices, a subset of participating clients are typically selected for FL model training with the uniform client sampling \cite{wang2024delta, luo2024adaptive}. Recent research has demonstrated the effectiveness of the partial client participation approach in FL with a solid theoretical convergence foundation \cite{bian2024accelerating}. In addition, given the diverse nature of heterogeneous data (typically amassed from distributed devices operating under varying mobile scenarios), there exists an inherent risk associated with the encompassed confidential information, such as personal health information and financial records, thereby elevating the potential for substantial privacy breaches \cite{kang2023incentive}. To address the significant privacy leakage associated with the gradient propagation scheme in FL systems \cite{zhu2019deep}, $\rho$-zero-concentrated differential privacy ($\rho$-$z$CDP) mechanism \cite{bun2016concentrated} has been widely adopted in recent literature \cite{hu2023shield, shi2024efficient}, which mitigates private information leakage by introducing Gaussian distributed-based noise into local or global model parameters.

\begin{figure*}[t]
\centerline{\includegraphics[width=0.85\textwidth, trim=0 0 0 0,clip]{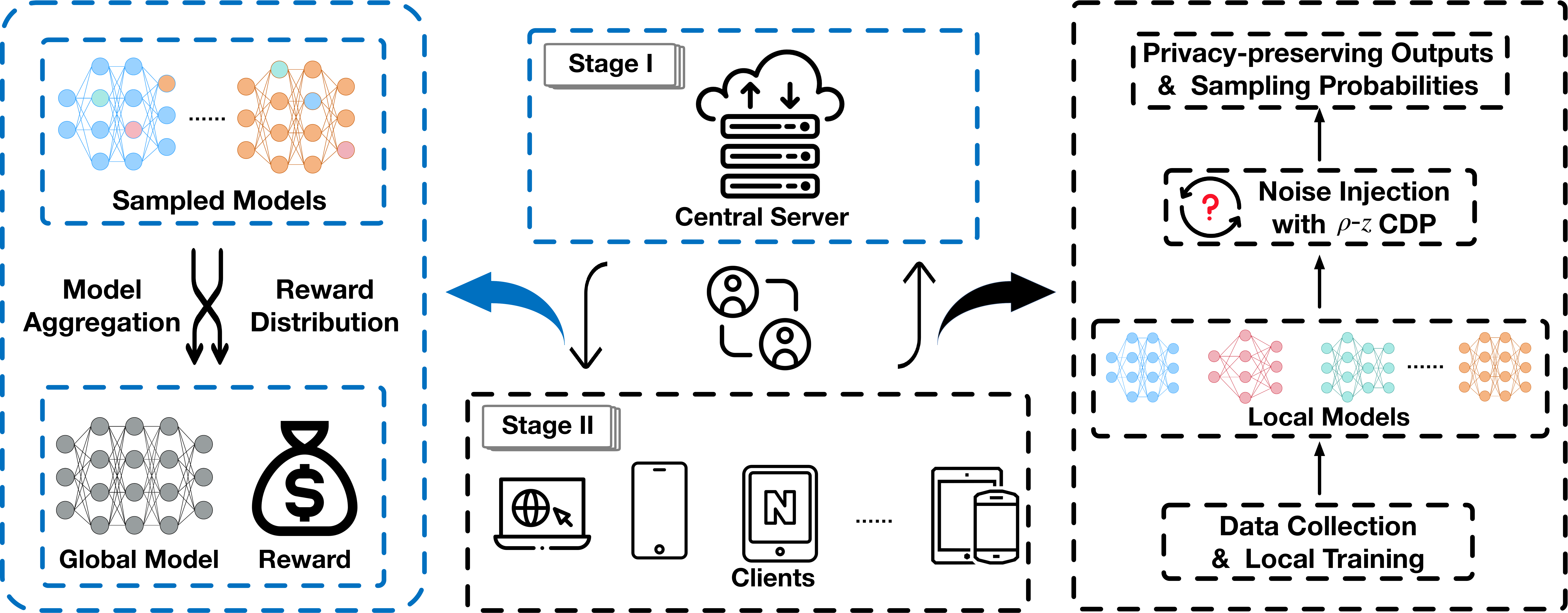}}
\caption{The framework of FedPCS: A $\rho$-$z$CDP technique-based incentive mechanism architecture in FL, where the blue dashed outline indicates the model aggregation and reward allocation operated by the central server, and the black dashed outline indicates the local training, model disturbance, and upload~process.}
\label{framework}
\vspace{-14pt}
\end{figure*}

However, uniformly random sampling of participating clients results in significant model bias during model aggregation. The randomly selected unrepresentative subsets may exacerbate variance introduced by client sampling, thereby slowing down the convergence rate and degrading model performance, which is particularly problematic in the general presence of system heterogeneity (e.g., artificial noise for privacy preservation) and statistical heterogeneity (i.e., Non-IID datasets) \cite{wang2024delta, liao2024optimal}. Based on the $\rho$-$z$CDP mechanism, the noise determined by heterogeneous and selfish clients further compromises uniform sampling due to the divergent noise magnitudes (according to the personalized privacy budgets). Consequently, there is a substantial probability of including clients with high noise magnitudes in each iteration, resulting in suboptimal model performance. Moreover, a prevailing and idealized assumption in most existing literature is that heterogeneous clients will participate in the FL systems unconditionally once invited by the central server. This assumption neglects significant factors such as computational consumption, communication resource burden, and the risk of severe information leakage during FL model training \cite{wu2023incentive}. Accordingly, without a well-designed economic incentivization, egocentric clients may be reluctant to contribute their local computation results (or computing facilities), given the potential privacy risks and significant training costs~\cite{hu2022incentive}.

In a nutshell, the limitations of existing research stem from \emph{the failure to jointly consider the impact of inherent system heterogeneity on client sampling and the design of effective incentive mechanisms within the context of privacy-preserving FL}. Specifically, clients should strategically modulate their privacy budgets to enhance their sampling probability and maximize long-term utilities. This insight drives us to investigate the following key questions.

\textbf{Key Question 1:} \emph{How to design an effective client sampling strategy for the central server to tackle system heterogeneity and ensure optimal model performance?}

\textbf{Key Question 2:} \emph{How to incentivize heterogeneous clients with various privacy budgets to participate in FL training under dynamic privacy budget constraints?}

\textbf{Key Question 3:} \emph{How should the client decide the optimal strategy to dynamically adjust its privacy budget over time in the context of unknown strategies employed by other clients?}

Several challenges arise when addressing the above questions: (1) It is difficult to analytically quantify the convergence performance of privacy-preserving FL models based on the client sampling probabilities. (2) Modeling the interaction between the central server and clients is complicated and challenging within an uncertain environment. (3) The closed-form of the optimal strategy profile of both sides is obstructive to obtain due to the incomplete information and the dynamic constraints of the decision variable.

In light of the discussions above, we proposed a novel framework FedPCS for the incentive mechanism design in FL with the privacy-aware client sampling probabilities, as shown in Fig. \ref{framework}. The multi-fold contributions of our work are summarized as follows:
\begin{itemize} 
\item \emph{A Game-Theoretic Formulation of the Privacy-aware Client Sampling:} To the best of our knowledge, we are the first to jointly consider the time-depend client sampling and reward design in privacy-preserving FL models. We propose a novel framework, named FedPCS, which constructs the interaction between the central server and clients as a two-stage Stackelberg game. Through dynamically adjusting the subset of sampled clients over time, we theoretically obtain the central server's optimal reward and client sampling strategy in Stage~\uppercase\expandafter{\romannumeral1} and the optimal correction factor of each client that dynamically adjusts its privacy budget in Stage~\uppercase\expandafter{\romannumeral2}. Subsequently, we rigorously demonstrate that the optimal strategy profile constitutes a Stackelberg Nash Equilibrium.
\item \emph{Mean-Field Estimator for Decentralized Strategy Profile Design:} Given the incomplete information resulting from the lack of local information exchange among clients during local model training, which complicates the determination of optimal decentralized client strategies, we introduce a mean-field estimator to evaluate the average privacy budget. Leveraging this estimator, we successfully derive the approximate optimal decentralized correction factor in closed-form. Moreover, through theoretical analysis, the existence of the fixed point for the mean-field estimator is proved.

\item \emph{Efficiency Analysis of Clients' Strategies via PoA:} To analyze the efficiency improvement of clients' strategies under our proposed privacy-aware sampling compared to random sampling, we utilize the PoA to measure the ratio between the social welfare achieved by the socially optimal strategy and that attained by these two methods. Our theoretical analysis illustrates that the random sampling policy (commonly employed in FL systems) is significantly affected by egocentric clients, resulting in PoA approaching infinity as the lower boundary of privacy budgets closes to 0. On the contrary, our proposed privacy-aware sampling scheme incentivizes egocentric clients to determine their privacy budgets prudently and successfully bounding PoA within a reachable context.

\item \emph{Adaptive Optimal Client Sampling Ratio under Dynamic Privacy Constraints:} Besides conducting FL training with a constant client sampling ratio, we further extend our model and analysis to more realistic scenarios, in which the privacy preservation demands of the FL systems vary in different stages of model training. In this extended scenario, the central server adaptively determines the optimal sampling ratio and reward allocated according to the dynamic privacy constraints, which makes the strategy analysis more complicated. Nonetheless, we successfully derive the closed-form optimal sampling ratio and reward. In addition, we conduct experiments to demonstrate the superiority of FedPCS under dynamic constraints.

\item \emph{Convergence Analysis and Performance Assessment:} We conduct rigorous analyses of the convergence upper bound of FedPCS with artificial noise and partial participant involvement, thus establishing a robust theoretical foundation. Further, we comprehensively assess the performance of our proposed framework on different real-world datasets (Fashion-MNIST, CIFAR-10, SVHN, CIFAR-100, CINIC-10, and Tiny-ImageNet). Experimental results demonstrate that FedPCS outperforms the state-of-the-art benchmarks in terms of utilities and FL model performance.
\end{itemize}

The rest of this paper is organized as follows. We review the related work in Section~\ref{related_work} and introduce the system model in Section~\ref{system_model}. The optimal strategy profile and Stackelberg Nash Equilibrium are discussed in Sections~\ref{Stackelberg_Nash_Equilibrium_Analysis}. The rigorous theoretical convergence analysis for our proposed framework FedPCS is presented in Section~\ref{Convergence_Analysis}. In Section~\ref{poa_analysis}, we analytically compare the two strategies via PoA analysis. Subsequently, we extend the system model and analysis to the dynamic privacy constraint scenario in Section~\ref{adaptive_sampling_ratio}. Experimental results are conducted in Section~\ref{Experiments} and this paper is concluded in Section~\ref{Conclusion}.

\section{Related Work} \label{related_work}
In this section, under the privacy-preserving scenarios, we conduct a comprehensive review of recent literature focusing specifically on two aspects, i.e., client sampling strategy for FL and incentive mechanism design in FL.

\subsection{Client Sampling Strategy for FL}
Client sampling technique is pivotal in addressing the statistical and system heterogeneity challenges in cross-device FL \cite{luo2022tackling}. Jiang \emph{et al.} \cite{jiang2024dordis} propose an ‘add-then-remove’ scheme, named Dordis, which enforces a required noise level precisely and ensures the privacy budget prudently utilization considering unpredictable sampled clients drop out. Chen \emph{et al.} \cite{chen2024personalized} propose an adaptive client sampling strategy to accelerate the model convergence rate and improve model performance with privacy guarantees and data heterogeneity trade-offs. Zhang \emph{et al.} \cite{zhang2023fed} design a heterogeneity-aware client sampling mechanism, named Fed-CBS, that effectively reduces the class-imbalance problem of the grouped datasets. Wu \emph{et al.} \cite{wu2023anchor} propose FedAMD which separates clients as anchor and miner groups to conduct gradient calculation, thus accelerating the training process and improving model performance. He \emph{et al.} \cite{he2023gluefl} propose GlueFL that alleviates the effect of staleness in client sampling and minimizes downstream bandwidth in cross-device FL via sticky sampling and mask shifting techniques.

Our research differs from these earlier works \cite{jiang2024dordis, chen2024personalized, zhang2023fed, wu2023anchor, he2023gluefl} in two aspects. First, to precisely assess the model performance, we rigorously analyze the convergence properties of model accuracy rather than commonly utilized empirical metrics in \cite{xu2023personalized, wu2023incentive}. Moreover, we construct the interaction between the central server and clients as a more realistic model considering both the profit conflicts among participants and dynamic constraints in FL systems, which makes it more challenging to tackle the trade-off between the profit and training costs during the distributed on-device model training process.

\subsection{Incentive Mechanism Design in FL}
In the domain of FL systems, the design of incentive mechanisms, predicated on both complete and incomplete information frameworks, has been subject to extensive investigation. The literature primarily focuses on motivating participants, who are financially self-interested, to contribute high-quality data and local training results \cite{hu2022incentive}. Mao \emph{et al.} \cite{mao2024game} formulate clients’ privacy-preserving behaviors in cross-silo FL as a multi-stage privacy preservation game to trade-off between the convergence performance and privacy loss. Huang \emph{et al.} \cite{huang2024collaboration} consider the collaboration behaviors among clients and the central server with reward allocation under the heterogeneity of participating clients’ multidimensional attributes. Wang \emph{et al.} \cite{wang2023trade} propose a dynamic privacy pricing game that considers the multi-dimensional information asymmetry among clients to lower the scale of locally added DP noise for differentiated economic compensations, thus enhancing FL model utility. Huang \emph{et al.} \cite{huang2024imfl} propose a data quality-aware incentive mechanism to encourage clients’ participation in the context of artificial intelligence-generated content (AIGC) empowered FL.

Different from the earlier works \cite{mao2024game, huang2024collaboration, wang2023trade, huang2024imfl} with fixed client subset for incentive mechanism design, we focus on dynamic sampling on heterogeneous clients with time-varying personalized attributes. Without an efficient dynamic client sampling strategy, the performance of FL systems may significantly suffer from the high heterogeneity of clients with random noise. 

\section{System Model and Problem Formulation} \label{system_model}
In this section, we start by summarizing the preliminaries regarding the standard FL model and $\rho$-$z$CDP mechanism. Then, we introduce our proposed privacy-aware client sampling strategy in FL systems. Finally, we construct the optimization problem of the central server and clients respectively. The workflow of our proposed FL framework FedPCS is presented in Fig.~\ref{framework} and the key notations that are used throughout the paper are provided in Table \ref{tab1}.

\subsection{Standard Federated Learning Model}
Under the standard FL framework, we hypothesize that $N$ geographically distributed clients participate in FL, and each client $i \!\in\! \{1, 2, \ldots, N\}$ uses their private and locally generated data set $\mathcal{D}_{i}$ with datasize $|\mathcal{D}_{i}|$ to perform model training. At each global iteration $t \!\in\! \{0,1, \ldots, T\}$, client $i$ performs mini-batch stochastic gradient descent (SGD) training to update its local model parameter~parallelly:
\begin{align} \label{sgd}
\boldsymbol{w_{i}(t+1)}=\boldsymbol{w(t)}-\eta \nabla F_{i}(\boldsymbol{w(t)}), 
\end{align} 
where $\eta$ represents the learning rate and $\nabla F_{i}(\boldsymbol{w(t)})$ is $i$-th client's gradient at iteration $t$. Once $N$ clients upload the local parameters, the central server averages the local models to generate the updated global parameter by $\boldsymbol{w(t)} = \sum_{i=1}^{N} \theta_{i} \boldsymbol{w_{i}(t)}$ with the aggregation weight $\theta_{i} \!=\! \frac{|\mathcal{D}_{i}|}{\sum_{j=1}^{N} |\mathcal{D}_{j}|}$. After updating the global model, the central server dispatches the updated global parameter to all clients for the next iteration's training. The goal of the FL method is to find the optimal global parameter $\boldsymbol{w}^{*}$ to minimize the global loss function $F(\boldsymbol{w})$:
\begin{align} \label{optimal_local_parameter}
\boldsymbol{w}^{*}\!=\!\text{arg}\min _{\boldsymbol{w}} F(\boldsymbol{w})\!=\!\text{arg}\min _{\boldsymbol{w}}  \sum\nolimits_{i=1}^{N} \theta_{i} F_{i}(\boldsymbol{w}). 
\end{align}

\begin{table}[t]
\setlength{\abovecaptionskip}{0cm} 
\caption{Key Notations and Corresponding Meanings}
\begin{center}
\renewcommand\arraystretch{1.2}
\begin{tabular}{p{0.5in}p{2.65in}}
\toprule[0.8pt]
\makecell[l]{\textbf{Symbol}}&\makecell[c]{\textbf{Description}}\\
\midrule[0.8pt]
$N$, $T$ & number of total clients and global training iteration  \\

$\mathcal{D}_{i}$, $|\mathcal{D}_{i}|$ & private data on client $i$ and corresponding datasize   \\

$\theta_{i}$ & local model aggregation weight of client $i$   \\

$K$, $\tau$ & number of selected clients and sampling ratio \\

$x_{i}^{t}$ & the sampling probability of client $i$ at $t$-th global iteration   \\

$\boldsymbol{w_{i}(t)}$, $\!\boldsymbol{w(t)}\!\!\!\!\!$ & the local and global model parameter at $t$-th global iteration \\

$\rho_{i}^{t}$, $\sigma_{i}^{2}(t)$ & privacy budget and variance of client $i$ at $t$-th global iteration \\

$\rho_{L}$, $\rho_{H}$ & the lower and upper boundary of privacy budgets   \\

$\alpha_{i}^{t}$, $R_{t}$ & correction factor of client $i$ and reward at $t$-th global iteration \\

$\phi(t)$ & mean-field estimator at $t$-th global iteration \\

$\gamma$, $\varphi_{i}$ & weight parameters for central server and clients respectively \\

\bottomrule[0.8pt]
\end{tabular}
\label{tab1}
\end{center}
\vspace{-2pt}
\end{table}

\subsection{\texorpdfstring{$\rho$}{}-\texorpdfstring{$z$}{}CDP Mechanism}
To solve the gradient inverse attack \cite{zhu2019deep}, $\rho$-zero-concentrated differential privacy ($\rho$-$z$CDP) mechanism \cite{bun2016concentrated} is proposed for its tight composition bound and suitability for the end-to-end privacy loss analysis of iteration algorithms \cite{hu2020trading}. We define the privacy loss metric by $\mathcal{L}_{p}$. For a randomized~mechanism $ \mathcal{M}\! : \!\mathcal{X} \! \rightarrow \! \mathcal{E}$ with domain $ \mathcal{X}$ and range $ \mathcal{E}$ given two adjacent datasets $ \mathcal{D}$, $\mathcal{D}^{\prime} \! \subseteq \! \mathcal{X}$ with $\mathcal{D}$ and $\mathcal{D}^{\prime}$ sharing the same data size $|\mathcal{D}_{i}|$ but only differing by one sample, after observing an output $o \in \mathbb{R}$, the privacy loss is given as $ \mathcal{L}_{p}\!=\!\ln(\frac{Pr[\mathcal{M}(\mathcal{D})= o]}{Pr[\mathcal{M}(\mathcal{D}^{\prime})= o]})$. Then, the definition of $\rho$-$z$CDP is given as follows.
\begin{definition}\label{definition_dp}
A randomized mechanism $ \mathcal{M}\!:\! \mathcal{X} \! \rightarrow \! \mathcal{E}$ with domain $ \mathcal{X}$ and range $ \mathcal{E}$ satisfies $\rho$-$z$CDP mechanism, if for any coefficient $\alpha \! \in \! (1, \infty)$, $\mathbb{E}[e^{(\alpha-1)\mathcal{L}_{p}}] \leq e^{(\alpha-1)(\rho\alpha)}$ holds. 
\end{definition}

Given any query function $ Q:\mathcal{X} \!\rightarrow\! \mathcal{E}$, the sensitivity of the query function is defined as $ \Delta{Q} \!=\! \max_{\mathcal{D},\mathcal{D}^{\prime}}\|Q(\mathcal{D})\!-\!Q(\mathcal{D}^{\prime}) \|_{2}$ for any two adjacent datasets $\mathcal{D}$, $\mathcal{D}^{\prime} \! \subseteq \! \mathcal{X}$, where the adjacent datasets share the same data size and only differ in the $j$-th data sample. Subsequently, at $t$-th global iteration, by adding the artificial Gaussian noise~$\boldsymbol{n_{i}(t)}  \! \sim \! \mathcal{N}(0, \sigma_{i}^{2}(t))$, the transmitted parameter of client $i$ becomes:
\begin{align} \label{local_gradient_perturbation}
\nabla \widetilde{F}_{i}(\boldsymbol{w(t)}) = \nabla F_{i}(\boldsymbol{w(t)}) + \boldsymbol{n_{i}(t)},
\end{align}
where $\rho$-$z$CDP mechanism is satisfied with the privacy budget $\rho_{i}^{t}$ of each participating client $i$ at $t$-th global iteration following $\rho_{i}^{t} \!=\! \frac{\Delta{Q}^{2}}{2\sigma_{i}^{2}(t)}$, as demonstrated in \cite{dwork2014algorithmic}. Suppose that there exists a clipping threshold $W$ for the $i$-th client's local model parameters at $t$-th global training iteration with adding artificial perturbation ($\|\boldsymbol{w_{i}(t)}\| \!\leq\! W$) as in \cite{mcmahan2017communication}. Consequently, based on the definition of the query function $Q$, we can easily derive the upper bound of the sensitivity $\Delta{Q}$ and quantify the variance of the artificial Gaussian noise $\sigma_{i}^{2}(t)$ as follows.
\begin{proposition} \label{proposition_variance}
At $t$-th global iteration, by leveraging the Gaussian distribution-based $\rho$-$z$CDP technique to perturb the transmitted local parameter, the variance of the Gaussian random noise $\sigma_{i}^{2}(t)$ of $i$-th client is derived as $\sigma_{i}^{2}(t) \!=\! \frac{2 W^{2}}{\rho_{i}^{t} |\mathcal{D}_{i}|^{2}}$.
\end{proposition}

The detailed proof of Proposition \ref{proposition_variance} is shown in Appendix~\ref{proof_proposition_variance}. From Proposition \ref{proposition_variance}, we observe that the privacy budget is inversely proportional to the Gaussian noise scale. In other words, as the privacy budget decreases, the perturbation applied on the local gradient $\nabla F_{i}(\boldsymbol{w(t)})$ of client $i$ increases, and the model performance suffers.

\subsection{Privacy-Aware Client Sampling Model}
We sample the heterogeneous clients in FL systems based on the privacy-aware probability distribution $\boldsymbol{X} \!=\! \{\boldsymbol{X_{t}}, t \in \{0, 1, \\ \ldots,T\}\}$, where the sampling probability vector $\boldsymbol{X_{t}} \!=\! \{x_{1}^{t}, x_{2}^{t}, \\ \ldots, x_{N}^{t}\}$ and the sampling probability of the participating client $i$ is defined by $x_{i}^{t} \!=\! \frac{\rho_{i}^{t}}{\sum_{i=1}^{N} \rho_{i}^{t}}$, where the sampling probability $x_{i}^{t} \!\in\! [0, 1)$ and $\sum_{i=1}^{N} x_{i}^{t} \!=\! 1$ for any global training iteration. According to the sampling probability vector $\boldsymbol{X_{t}}$, the central server dynamically generates the sampled client subset $\mathcal{K}_{t}$ by sampling $K$ times without replacement from $N$ overall clients for achieving optimal model training performance in the FL systems, where set size of the sampled client subset follows $|\mathcal{K}_{t}| \!=\! K \!=\! \tau N$ and the sampling rate $\tau \!\in\! (0,1] $. Following the recent research \cite{luo2022tackling, li2019convergence, luo2024adaptive} and based on the privacy-aware client sampling probabilities above, the global model aggregation weight for each client $i$ is associated with the corresponding sampling probability $x_{i}^{t}$. Thus, the global model update is given by:
\begin{align} \label{sgd_sample}
\boldsymbol{w(t+1)}=\boldsymbol{w(t)} - \eta \sum\nolimits_{i=1}^{K} \frac{\theta_{i}}{K x_{i}^{t}} \nabla F_{i}(\boldsymbol{w(t)}).
\end{align}

For our theoretical analysis in the later content, we employ the following commonly used assumptions and lemmas in the existing literature \cite{wu2021fast, sun2023stability}.  

\begin{assumption} \label{assumption_smoothness}
($\beta$-Lipschitz Smoothness) Local loss function $ F_{i}(\boldsymbol{w})$ is $\beta$-Lipschitz smoothness for each participating client $ i \in\{1,2, \ldots, N\}$ with any $ \boldsymbol{w}, \boldsymbol{w}^{\prime} \in \mathbb{R}^{d}$: 
\begin{align}
& \|\nabla{F_{i}(\boldsymbol{w}})-\nabla{F_{i}(\boldsymbol{w}^{\prime}})\| \leq \beta\|\boldsymbol{w}-\boldsymbol{w}^{\prime}\|, \\
F(\boldsymbol{w}) & - F(\boldsymbol{w}^{\prime}) \!\leq\! \nabla{F(\boldsymbol{w}^{\prime})}^{T}(\boldsymbol{w}\!-\!\boldsymbol{w}^{\prime}) \!+\!\frac{\beta}{2} \|\boldsymbol{w}\!-\!\boldsymbol{w}^{\prime}\|_{2}^{2}.
\end{align}
\end{assumption}

\begin{assumption} \label{assumption_moment}
(First and Second Moment Limits) For any $t \!\in\! \{0,1,\ldots,T\}$, with the global loss function $F(\boldsymbol{w})$ satisfies $F(\boldsymbol{w(t)}) \! \geq \! F(\boldsymbol{w}^{*})$ and the stochastic
gradients satisfy:
\begin{itemize}[itemsep=0pt, leftmargin=15pt, align=right]
\item[(\romannumeral1)] There exists constants $\kappa_{G} \!\geq\! \kappa \!>\! 0$, for all iteration $t$ and each client $i \in \{1, 2, \ldots, N\}$ satisfies,
\begin{align}
\nabla{F(\boldsymbol{w(t)})}^{T}\mathbb{E}[\nabla{F_{i}(\boldsymbol{w(t)})}] &\geq \kappa \|\nabla{F(\boldsymbol{w(t)})}\|_{2}^{2}, \nonumber \\
\|\mathbb{E}[\nabla{F_{i}(\boldsymbol{w(t)})}]\|_{2} &\leq \kappa_{G} \|\nabla{F(\boldsymbol{w(t)})}\|_{2}.
\end{align}
\item[(\romannumeral2)] There exists coefficients $M \!\geq\! 0$, $M_{V} \!\geq\! 0$ for any iteration $t$ and client $i$. The variance of $\nabla{F_{i}(\boldsymbol{w(t)})}$ is defined by $\mathbb{V}[\nabla{F_{i}(\boldsymbol{w(t)})}] \!=\! \mathbb{E}[\|\nabla{F_{i}(\boldsymbol{w(t)})}\|_{2}^{2}]\!-\!\|\mathbb{E}[\nabla{F_{i}(\boldsymbol{w(t)})}]\|_{2}^{2}$,
\begin{align}
\mathbb{V}[\nabla{F_{i}(\boldsymbol{w(t)})}] \leq M + M_{V} \|\nabla{F(\boldsymbol{w(t)})}\|^{2}_{2}.
\end{align}
\end{itemize}
\end{assumption}

\begin{assumption} \label{assumption_convexity}
($\psi$-Strong Convexity) The global loss function $F(\boldsymbol{w})$ satisfies $\psi$-strong convex for any model parameter~vectors $\boldsymbol{w}, \boldsymbol{w}^{\prime} \in \mathbb{R}^{d}$, we hold: 
\begin{align}
\frac{\psi}{2}\|\boldsymbol{w} - \boldsymbol{w^{\prime}}\|_{2}^{2} \leq F(\boldsymbol{w}) - F(\boldsymbol{w}^{\prime}). 
\end{align}
\end{assumption}

\begin{assumption} \label{assumption_subgradient}
(Subgradient for Non-Convexity) The global loss function $F(\boldsymbol{w})$ admits a unique subgradient $g \!\in\! \mathbb{R}^{d}$ for any model parameter vector $\boldsymbol{w}, \boldsymbol{w}^{\prime} \!\in\! \mathbb{R}^{d}$, it satisfies, 
\begin{align}
F(\boldsymbol{w}) - F(\boldsymbol{w}^{\prime}) \geq g^\top (\boldsymbol{w} - \boldsymbol{w}^{\prime}),
\end{align}
\end{assumption}

\begin{assumption}\label{assumption_bounded_gradients}
(Bounded Local Gradient) The local gradient $\nabla F_{i}(\boldsymbol{w})$ attains $D$-bound gradients for client $i \in \{1, 2, \ldots, N\}$ and model parameter $ \boldsymbol{w} \!\in\! \mathbb{R}^{d}$, we have: $\|\nabla F_{i}(\boldsymbol{w})\|_{2} \leq D$. 
\end{assumption}

\begin{lemma} \label{lemma_PL_inequality}
(Polyak-\L{}ojasiewicz inequality) Denote the optimal global parameter by $\boldsymbol{w}^{*}$, the $\beta$-Lipschitz continuous global loss function $F(\boldsymbol{w})$ satisfies Polyak-\L{}ojasiewicz condition if the following holds for $\mu > 0$ and $\boldsymbol{w} \in \mathbb{R}^{d}$: 
\begin{align}
\frac{1}{2}\|\nabla{F(\boldsymbol{w})}\|^{2} \geq \mu(F(\boldsymbol{w})-F(\boldsymbol{w}^{*})).
\end{align}
\end{lemma}

\subsection{Problem Formulation} \label{Systen_Model_and_Problem_Formulation}
\subsubsection{Central Server's Cost Minimization Problem in Stage \texorpdfstring{\uppercase\expandafter{\romannumeral1}}{}}
We formulate the interactions between the central server and clients as a two-stage Stackelberg game. Prior to delineating the cost function of the central server, we initially explore the impact of each client's privacy budget on the accuracy loss of the global model. Suppose that the global loss function $F(\boldsymbol{w})$ satisfies Assumption~\ref{assumption_smoothness} and Lemma~\ref{lemma_PL_inequality} with $\nabla F(\boldsymbol{w(t)})$ attains an upper bound $V$ ($\|\nabla F(\boldsymbol{w(t)})\|_{2} \!\leq\! V$) as in \cite{ding2023incentive}. Then, we denote the dimension of the model parameters as $d$ and formulate the following proposition.
\begin{proposition} \label{proposition_central_server_accuracy_loss}
(Accuracy Loss of Global Model) Given the optimal global model parameter $\boldsymbol{w^{*}}\!$, the accuracy loss of the global model at $t$-th iteration is upper bounded by: 
\begin{align}\label{accuracy_loss}
\text{\small $\mathbb{E}[F(\boldsymbol{w(t)}) \!-\! F(\boldsymbol{w^{*}})] \!\leq\! \frac{\beta}{2\mu^{2}t} \left(V^{2} \!+\! 2 d W^{2} \! \sum\nolimits_{i=1}^{N} \! \frac{\theta_{i}^{2}}{|\mathcal{D}_{i}|^{2}\rho_{i}^{t}} \right) $}. \!\!
\end{align}
\end{proposition}

The detailed proof of Proposition~\ref{proposition_central_server_accuracy_loss} is provided in Appendix~\ref{proof_proposition_central_server_accuracy_loss}. Proposition~\ref{proposition_central_server_accuracy_loss} reveals that the accuracy loss of the global model is substantially influenced by the local model aggregation weight $\theta_i$, the datasize $|\mathcal{D}_{i}|$, and the privacy budget $\rho_i^t$ for each client. We quantify the global model quality by metric $\mathcal{A} \!=\! \sum_{i=1}^{N} \frac{\theta_i^2}{t |\mathcal{D}_{i}|^2 \rho_i^t}$. Notably, a lower value of term $\mathcal{A}$ correlates with the lower accuracy loss in FL models, thereby indicating enhanced global model performance. Define the reward vector decided by the central server as $\boldsymbol{R} \!=\! \{R_{t}, t \!\in\! \{0, 1,\ldots, T\}\}$ and the correction factor vector decided by the clients to dynamically adjust their privacy budgets as $\boldsymbol{\alpha} \!=\! \{\boldsymbol{\alpha_{i}}, i \!\in\! \{1, 2,\ldots, N\}\}$, $\boldsymbol{\alpha_{i}} \!=\! \{\alpha_{i}^{t}, t \!\in\! \{0, 1, \ldots,T\}\}$. Then, the objective of the central server is to decide the optimal reward $R_t$ at $t$-th global iteration to minimize its cost function including the accuracy loss and rewards allocated as follows:
\begin{align}\label{central_server_utility}
U_{t}(R_{t},\boldsymbol{\alpha}) = \sum\nolimits_{i=1}^{K} \left(\gamma \frac{\theta_{i}^{2}}{t |\mathcal{D}_{i}|^{2}\rho_{i}^{t}} + (1 - \gamma) R_{t} \rho_{i}^{t} \right), 
\end{align}
where the weighted aggregation parameter $\gamma \in (0,1)$ balances the trade-off between the accuracy loss and reward term.

\subsubsection{Clients’ Utility Maximization Problem in Stage \texorpdfstring{\uppercase\expandafter{\romannumeral2}}{}}
At $t$-th global iteration, given any $R_{t}$ from the central server, the reward allocated to each sampled client~is proportional to its privacy budget $\rho_{i}^{t}$, defined as $r_{i}^{t} \!=\!  \rho_{i}^{t}R_{t}$. Analogous to \cite{xu2022incentive}, we employ the extensively used quadratic cost function $(\rho_{i}^{t})^{2}$ to quantify the privacy costs. To increase the allocated reward and reduce the privacy cost, each client should dynamically adjust its privacy budget to maximize its overall utility. By introducing the correction factor $\alpha_{i}^{t}$ to adjust the privacy budget of client $i$ over time, the privacy budget $\rho_{i}^{t+1}$ is updated by the average privacy budgets of all clients $i \!\in\! \{1, 2, \ldots, N\}$ and $\rho_{i}^{t}$ as shown in Eq.~(\ref{update_constraint}), which poses a significant challenge due to the lack of privacy budget information of other clients. Moreover, the probability that client $i$ is successfully sampled by the central server at $t$-th global iteration is $(1 - (1 - x_{i}^{t})^{K})$. Hence, the objective of client $i$ is to maximize its own utility $U_{i}$ and the possibility of being sampled by dynamically adjusting the correction factor $\alpha_{i}^{t}$ at each global iteration $t$ according to the reward $R_{t}$ from the central server as follows:
\begin{align}
U_{i}(\boldsymbol{\alpha_{i}}, \boldsymbol{\alpha_{\text{-}i}}, \boldsymbol{R}) \!=&  \sum\nolimits_{t=0}^{T} (1 \!-\! (1 \!-\! x_{i}^{t})^{K}) \nonumber \\
&\times (\rho_{i}^{t}R_{t} \!-\! \varphi_{i} (\rho_{i}^{t})^{2} \!\!-\! (1 \!-\! \varphi_{i}) (\alpha_{i}^{t})^{2}) ,  \label{client_utility} \\
s.t. \quad \rho_{i}^{t+1} = (1 - & \alpha_{i}^{t}) \frac{1}{N} \sum\nolimits_{j=1}^{N}\! \rho_{j}^{t} + \alpha_{i}^{t} \rho_{i}^{t}, \ \alpha_{i}^{t} \in (0,1),  \label{update_constraint}
\end{align}
where $\rho_{i}^{t} \!\in\! [\rho_{L}, \rho_{H}]$, $\boldsymbol{\alpha_{\text{-}i}} \!=\! \{\boldsymbol{\alpha_{1}}, \ldots, \boldsymbol{\alpha_{i-1}}, \boldsymbol{\alpha_{i+1}}, \ldots, \boldsymbol{\alpha_{N}}\}$, the weight parameter $\varphi_{i} \!\in\! (0,1)$ balances the trade-off between the privacy cost and the fluctuation of correction factor $\alpha_{i}^{t}$.

Note that the reward vector $\boldsymbol{R} \!=\! \{R_{t}, t \!\in\! \{0, 1,\ldots, T\}\}$ announced by the central server will affect the clients' designs of the correction factors, which in return affects the central server's cost as shown in Eq.~(\ref{central_server_utility}). Consequently, the two-stage Stackelberg game between the central server in Stage \uppercase\expandafter{\romannumeral1} and the clients in Stage \uppercase\expandafter{\romannumeral2} is constructed as:
\begin{align}
\text{Stage \uppercase\expandafter{\romannumeral1}:} & \  \min_{R_{t}} U_{t}(R_{t},\boldsymbol{\alpha}), \label{central_server_obj} \\
\text{Stage \uppercase\expandafter{\romannumeral2}:} & \  \max_{\substack{\boldsymbol{\alpha_{i}}}} U_{i}(\boldsymbol{\alpha_{i}}, \boldsymbol{\alpha_{\text{-}i}}, \boldsymbol{R}). \label{client_obj}  
\end{align}

\section{Stackelberg Nash Equilibrium Analysis}\label{Stackelberg_Nash_Equilibrium_Analysis}
In this section, we will find the Stackelberge Nash equilibrium profile $(\boldsymbol{R^{*}}, \{\boldsymbol{\alpha_{i}^{*}}$, $i \! \in \! \{1,2, \ldots, N\}\})$, with $ \boldsymbol{R^{*}} \!\!=\! \{R_{t}^{*}, t \! \in \! \{0, 1, \ldots, T \}\}\!$ and $ \boldsymbol{\alpha_{i}^{*}} \!=\! \{\alpha_{i}^{t*}, t \!\in\! \{0,1, \ldots, T\}\}$ for the central server~and clients respectively through the backward induction approach. Firstly, we focus on the followers' strategies and derive each client’s optimal correction factor $\alpha_{i}^{t*}$ at the $t$-th global iteration with any given reward $\boldsymbol{R}$. Then, based on the trade-off between the accuracy loss of the global model and the total reward allocated to sampled clients, we derive the optimal reward $R_{t}^{*}$. Finally, we prove that the optimal strategies form the Stackelberg Nash Equilibrium. 

\subsection{Determine the Optimal Strategy Profile}
According to the discrete-time nonlinear optimization problem with dynamic constraints in Eqs. (\ref{client_utility})-(\ref{update_constraint}), at $t$-th global iteration, the correction factor $\alpha_{i}^{t}$ of client~$i$ is affected by other clients' privacy budget $\rho_{j}^{t}$ ($j \neq i$). During the local training, there is no exchange for private information among clients, indicating that client $i$ can not access other clients' privacy budget $\rho_{j}^{t}$, which makes the multi-client joint correction factor design over time challenging. To tackle the above challenge, we introduce a mean-field term to estimate the average privacy budget, based on which, the correction factor $\alpha_{i}^{t}$ of each client can be derived without requiring other clients' private information during local training. 
\begin{definition}\label{definition_mean_field}
(Mean-Field Term) To figure out the optimal correction factor $\alpha_{i}^{t}$ for each client, we introduce the mean-field term to approximate the average privacy budget $\rho_{i}^{t}$ of all clients at global iteration $t \in\{0,1, \ldots, T\}$:
\begin{align} \label{mean_field}
\phi(t) = \frac{1}{N}\sum\nolimits_{i=1}^{N} \rho_{i}^{t}.
\end{align}
\end{definition}

From the mathematical point of view, the mean-field term $\phi(t)$ is a given function in our FL optimization problem. Then, based on Definition \ref{definition_mean_field}, the sampling probability of client~$i$ at $t$-th global iteration can be further derived as $x_{i}^{t} \!=\! \frac{\rho_{i}^{t}}{\sum_{i=1}^{N} \rho_{i}^{t}} \!=\! \frac{\rho_{i}^{t}}{N \phi(t)}$. Thus, we reformulate the discrete-time nonlinear optimization problem in Eqs. (\ref{client_utility})-(\ref{update_constraint}) as follows:
\begin{align}
U_{i}(\boldsymbol{\alpha_{i}}, \boldsymbol{\alpha_{\text{-}i}}, \boldsymbol{R}) \!=&\  \sum\nolimits_{t=0}^{T} (1 \!-\! (1 \!-\! \frac{\rho_{i}^{t}}{N \phi(t)})^{K}) \nonumber \\
&\ \times (\rho_{i}^{t}R_{t} \!-\! \varphi_{i} (\rho_{i}^{t})^{2} \!\!-\! (1 \!-\!\varphi_{i}) (\alpha_{i}^{t})^{2}),  \label{client_utility_refined} \\
s.t. \quad \rho_{i}^{t+1} = (1 & -\alpha_{i}^{t}) \phi(t) + \alpha_{i}^{t} \rho_{i}^{t}, \ \alpha_{i}^{t} \in (0,1) .\label{update_constraint_refined}
\end{align}


The optimization problem presented in Eqs. (\ref{client_utility_refined})-(\ref{update_constraint_refined}) can be characterized as a discrete-time linear quadratic optimal problem. By constructing the Hamilton function, we can derive the close form expression of the optimal correction factor  $\alpha_{i}^{t*}$ summarized in the following theorem.
\begin{theorem} \label{theorem_optimal_alpha}
(Optimal Correction Factor) In Stage \uppercase\expandafter{\romannumeral2}, given the mean-field term $\phi(t)$ over time, the correction factor of client $i \!\in\! \{1, 2, \ldots, N\}$ at global training iteration $t \!\in\! \{0,1, \ldots, T\}$ is derived as:
\begin{align}\label{optimal_alpha}
\!\!\!\!\text{ $\alpha_{i}^{t*} \!=\! \frac{(\rho_{i}^{t} \!-\! \phi(t))\{\sum_{j=t+2}^{T} [ S(j) \prod_{r=t+1}^{j-1} \alpha_{i}^{r}] \!+\! S(t \!+\! 1)\}}{2 (1 \!-\! \varphi_{i})(1 \!-\! (1 \!-\! \frac{\rho_{i}^{t}}{N \phi(t)})^{K})} $} ,
\end{align}
with $\alpha_{i}^{T} \!=\! 0$, where $S(t) = Q_{i}^{t} R_{t} + M_{i}^{t}$, $Q_{i}^{t} = 1 - (1 - \frac{\rho_{i}^{t}}{N \phi(t)})^{K} + (1 - \frac{\rho_{i}^{t}}{N \phi(t)})^{K - 1} \frac{K \rho_{i}^{t}}{N \phi(t)}$ and $M_{i}^{t} \!=\!  - 2 \varphi_{i} \rho_{i}^{t}(1 \!-\! (1 \!-\! \frac{\rho_{i}^{t}}{N \phi(t)})^{K}) - (1 - \frac{\rho_{i}^{t}}{N \phi(t)})^{K - 1} (\varphi_{i} (\rho_{i}^{t})^{2} + (1 -\varphi_{i}) (\alpha_{i}^{t})^{2})$.
\end{theorem}

\begin{algorithm}[t]
\small 
\caption{Iterative Calculation of Fixed Point $\phi(t)$} 
\label{alg1}
\SetAlgoLined
\KwIn{$N$, $T$, the clipping threshold $\epsilon_{0}$, $m \!=\! 1$, arbitrary initial value for $\epsilon$, $\alpha_{i}^{t}$, $\rho_{i}^{0}$, $\phi_{0}(t)$ and the allocated reward $R_{t}^{0}$, where $t \!\in\! \{0,1, \ldots, T\}$,~$i \!\in\! \{1,2,\ldots,N\}$.}
\KwOut{The fixed point of $\phi(t)$, $t \in \{0,1,\ldots,T\}$.}
\For{$t \leftarrow 0, 1, \ldots, T$}{
    \While{$\epsilon > \epsilon_{0}$}{
        \For{$i \leftarrow 1, 2, \ldots, N$}{
            Update $\rho_{i}^{t}$ according to Eq.~(\ref{update_constraint_refined}) \;
        }
        \For{$t \leftarrow 0, 1, \ldots, T-1$}{
            \For{$i \leftarrow 1, 2, \ldots, N$}{
                Update $\alpha_{i}^{t}$ according to Eq.~(\ref{optimal_alpha}) \;
            }
            Update $R_{t}^{m}$ according to Eq.~(\ref{optimal_r}) \;
            $\alpha_{i}^{T} = 0 $\;
        }
        $\phi_{m}^{\texttt{est}}(t) =  \frac{1}{N}\sum_{i=1}^{N} \rho_{i}^{t}$ \;
        $\phi(t) = \phi_{m}^{\texttt{est}}(t)$, $R_{t} = R_{t}^{m}$ \;
        $\epsilon = \phi_{m}^{\texttt{est}}(t) - \phi_{m-1}^{\texttt{est}}(t)$ \;
        $m = m + 1$\;
    }
}
\end{algorithm}

The proof of Theorem \ref{theorem_optimal_alpha} is presented in Appendix~\ref{proof_theorem_optimal_alpha}. Based on the optimal correction factor $\alpha_{i}^{t*}$ of each client in Theorem~\ref{theorem_optimal_alpha}, we can obtain the optimal reward $R_{t}^{*}$ of the central server in Stage~\uppercase\expandafter{\romannumeral1} by substituting the optimal correction factor in Eq.~(\ref{optimal_alpha}) into Eq.~(\ref{central_server_utility}) as shown in the following lemma.
\begin{lemma}\label{lemma_optimal_reward} 
(Optimal Reward) In Stage \uppercase\expandafter{\romannumeral1}, the optimal reward $R_{t}^{*}$ of the central server at $t$-th iteration can be obtained according to the optimal correction~factor $\alpha_{i}^{t*}$ in Theorem~\ref{theorem_optimal_alpha}:
\begin{align} \label{optimal_r}
\text{\small $R_{t}^{*} \!=\! \frac{1}{2 (1 \!-\! \gamma) \sum_{i=1}^{K} \! G_{i}^{t}}\! \sum\nolimits_{i=1}^{K} \left[\frac{\gamma \theta_{i}^{2} G_{i}^{t}}{t|\mathcal{D}_{i}|^{2}(G_{i}^{t} R_{t}^{*} \!+\! H_{i}^{t})^{2}} \!-\!  H_{i}^{t} \right] $},
\end{align}
where $\rho_{i}^{t} = G_{i}^{t} R_{t} + H_{i}^{t}$, $G_{i}^{t} = (\rho_{i}^{t-1} - \phi(t - 1)) I_{i}^{t}$, $H_{i}^{t} = (1 - J_{i}^{t} ) \phi(t - 1) + J_{i}^{t} \rho_{i}^{t-1}$, $I_{i}^{t} = \frac{(\rho_{i}^{t-1} - \phi(t-1))Q_{i}^{t}}{2 (1 - \varphi_{i})(1 - (1 - \frac{\rho_{i}^{t-1}}{N \phi(t-1)})^{K})}$ and $J_{i}^{t} = \frac{(\rho_{i}^{t-1} - \phi(t-1))\{\sum_{j=t+1}^{T} [S(j) \prod_{r=t}^{j-1} \alpha_{i}^{r}] + M_{i}^{t}\}}{2 (1 - \varphi_{i})(1 - (1 - \frac{\rho_{i}^{t-1}}{N \phi(t-1)})^{K})}$.
\end{lemma}

The detailed proof of Lemma \ref{lemma_optimal_reward} is provided in Appendix~\ref{proof_lemma_optimal_reward}. We observe that Eq.~(\ref{optimal_r}) is a high-ordered non-polynomial which is challenging to acquire a closed-form solution. Still, there are many efficient numerical optimization algorithms (e.g., Gradient Descent, Newton's method, or Quasi-Newton Methods) to obtain a numerical solution of $R_{t}^{*}$. In the following, we theoretically analyze the Stackelberg Nash Equilibrium of the above system. 
\begin{definition}\label{definition_stackelberg}
(Stackelberg Nash Equilibrium, SNE) The~optimal strategy profile $(\boldsymbol{R^{*}}, \{\boldsymbol{\alpha_{i}^{*}}$, $i \!\in\! \{1,2, \ldots, N\}\})$, with~the optimal reward vector $\boldsymbol{R}^{*} = \{R_{t}^{*}, t \in \{0,1,\ldots, T\}\}$ and the optimal correction factor vector $\boldsymbol{\alpha_{i}^{*}} \!=\! \{\alpha_{i}^{t*}, t \!\in\! \{0,1, \ldots, T\}\}$ constitutes a Stackelberg Nash Equilibrium (SNE) if for any reward $R_{t} \in \mathbb{R}$, $t \in \{0,1,\ldots, T\}$ and any correction factor $\alpha_{i}^{t} \in [0,1]$ of any client $i \in \{1,2, \ldots, N\}$,
\begin{align}
U_{t}(R_{t}^{*},\boldsymbol{\alpha^{*}}) &\leq U_{t}(R_{t},\boldsymbol{\alpha^{*}}),  \\
U_{i}(\boldsymbol{\alpha_{i}^{*}}, \boldsymbol{\alpha_{\text{-}i}^{*}}, \boldsymbol{R^{*}}) &\geq U_{i}(\boldsymbol{\alpha_{i}}, \boldsymbol{\alpha_{\text{-}i}^{*}}, \boldsymbol{R^{*}}). 
\end{align}
\end{definition}

\begin{theorem} \label{SNE}
The above two-stage Stackelberg game possesses a Stackelberg Nash Equilibrium.
\end{theorem}
\begin{proof}
Based on Theorem \ref{theorem_optimal_alpha}, we derive the optimal strategy of correction factor $\alpha_{i}^{t*}$ for each client $i \! \in \! \{1, 2, \ldots, N\}$ given any reward $R_{t} \!\in\! \mathbb{R}$. Subsequently, we need to demonstrate that there exists an optimal reward $R_{t}^{*}$ for the central server's cost function $U_{t}(R_{t},\boldsymbol{\alpha})$ according to the optimal correction factor vector $\boldsymbol{\alpha_{i}^{*}}\!=\!\{\alpha_{i}^{t*}, t \! \in \! \{0, 1, \ldots, T\}\}$. The derivation of Lemma~\ref{lemma_optimal_reward} indicates that the cost function of central server $U_{t}(R_{t},\boldsymbol{\alpha})$ is strictly concave with $\partial^{2} U_{t}(R_{t},\boldsymbol{\alpha^{*}})/\partial R_{t}^{2} \!>\! 0$ and $\lim_{R_{t} \rightarrow 0} \frac{\partial U_{t}(R_{t},\boldsymbol{\alpha^{*}})}{ \partial R_{t}} \rightarrow - \infty$, $\lim_{R_{t} \rightarrow +\infty} \frac{\partial U_{t}(R_{t},\boldsymbol{\alpha^{*}})}{ \partial R_{t}} \rightarrow + \infty$. Consequently, the optimal reward $R_{t}^{*}$ for the central server exists by solving the first-order derivation $\partial U_{t}(R_{t},\boldsymbol{\alpha^{*}})/\partial R_{t} \!=\! 0$ based on the optimal correction factor vector $\boldsymbol{\alpha_{i}^{*}}\!=\!\{\alpha_{i}^{t*}, t \! \in \! \\ \{0, 1, \ldots, T\}\}$. Essentially, the strategy $R_{t}^{*}$ and $\boldsymbol{\alpha_{i}^{*}}$ are mutual optimal strategy for each sampled client $i$ and the central server. Thus, the optimal strategy profile $(\boldsymbol{R^{*}},\! \{\boldsymbol{\alpha_{i}^{*}}$,\! $i \!\in\! \{1,2, \ldots, N\}\})$ with $ \boldsymbol{R^{*}} \!\!=\! \{R_{t}^{*}, t \! \in \! \{1, 2, \ldots, T \}\}\!$ and $\boldsymbol{\alpha_{i}^{*}} \!=\! \{\alpha_{i}^{t*}, \! t \!\in\! \{0, 1, \ldots, \\ T\}\}$ of the two-stage Stackelberg game possesses a Stackelberg Nash Equilibrium.  \end{proof}

\begin{algorithm}[t]
\small
\caption{Privacy-Aware Client Sampling in FL}
\label{alg2}
\SetAlgoLined
\KwIn{$N$, $\tau$, $T$, $\phi(t)$, learning rate $\eta$, weight $\theta_{i}$, $\varphi_{i}$ and $\gamma$.}
\KwOut{The global model parameter $\boldsymbol{w(T)}$.}
\For{$t \leftarrow 0, 1, \ldots, T$}{
    The central server decides the optimal reward $R_{t}^{*}$ and samples a client subset $\mathcal{K}_{t}$ according to the probability vector $\boldsymbol{X_{t}}$ \tcp*[r]{Client Sampling} 
    Each client decides the optimal correction factor $\alpha_{i}^{t*}$\;
    \ForEach{sampled client}{
        Local training: $\ \boldsymbol{w_{i}(t \!+\! 1)} \!=\! \boldsymbol{w(t)} \!-\! \eta \nabla F_{i}(\boldsymbol{w(t)})$ \;
        Gradient Perturbation: $\nabla \widetilde{F}_{i}(\boldsymbol{w(t)}) \!=\! \nabla F_{i}(\boldsymbol{w(t)}) \!+\! \boldsymbol{n_{i}(t)}$ \;
    }
    The central server updates the global model with: $w(t + 1) = w(t) - \eta \sum_{i=1}^{K} \! \frac{\theta_{i}}{K x_{i}^{t}} \nabla \widetilde{F}_{i}(\boldsymbol{w(t)})$  \tcp*[r]{Model Aggregation}
}
\end{algorithm}

\subsection{Update of \texorpdfstring{$\phi(t)$}{} for Finalizing Strategy Profile}
Note that the mean-field term $\phi(t)$ is affected by the privacy budget $\rho_{i}^{t}$ of all clients at global iteration $t \!\in\! \{0, 1, \ldots, T\}$, which in return reciprocally impacts the determination of the value of $\rho_{i}^{t}$. In the following discussion, we delve into a rigorous analysis to ascertain the existence of a fixed point for $\phi(t)$ and employ an iterative algorithm to obtain the precise solution, thereby finalizing the above optimal strategy design.
\begin{theorem} \label{theorem_fixed_point}
There exists a fixed point for the mean-field term $\phi(t)$, which can be attained through Algorithm \ref{alg1}.
\end{theorem}

\begin{figure*}[b]
\centering
\vspace{-10pt}
\noindent\rule{\textwidth}{0.5pt}
\vspace{-10pt}
\begin{minipage}{1.0\textwidth}
\begin{align} 
\text{\small $\mathbb{E}[\|\boldsymbol{w(T)} \!-\! \boldsymbol{w^{*}}\|_{2}^{2}] $} \!\leq& \text{\small $(1 \!-\! \frac{\psi \eta [(N \!-\! 1)\rho_{H} \!+\! \rho_{L}]}{K \rho_{L}})^{T}\mathbb{E}[\|\boldsymbol{w(0)} \!-\! \boldsymbol{w^{*}}\|_{2}^{2}] \!+\! \sum\nolimits_{s=0}^{T-1} $} \text{\small $(1 \!-\! \frac{\psi \eta [(N \!-\! 1)\rho_{H} \!+\! \rho_{L}]}{K \rho_{L}})^{s}\{\frac{\beta D^{2} \eta^{3} [(N \!-\! 1)\rho_{H} \!+\! \rho_{L}]}{K \rho_{L}}$} \nonumber \\
&+ \text{\small $\frac{\eta^{2} [(N \!-\! 1)\rho_{H} \!+\! \rho_{L}]^{2}}{K^{2} \rho_{L}^{2}} (D^{2} \!+\! \frac{2 d K W^{2}}{\rho_{L} (\sum_{i=1}^{N} |\mathcal{D}_{i}|)^{2}}) \!+\! \frac{\beta \eta^{3} [(N \!-\! 1)\rho_{H} \!+\! \rho_{L}]^{3}}{K^{3} \rho_{L}^{3}} (D^{2} \!+\! \frac{2 d K W^{2}}{\rho_{L} (\sum_{i=1}^{N} |\mathcal{D}_{i}|)^{2}}) \}$} . \label{optimal_gap_convexity} \\
\text{\small $\mathbb{E}[\|\boldsymbol{w(T)} \!-\! \boldsymbol{w^{*}}\|_{2}^{2}] $} \leq&\text{\small $(1 \!+\! \frac{2 T \eta \|g^\top \|_{2} [(N \!-\! 1)\rho_{H} \!+\! \rho_{L}]}{K \rho_{L}})\mathbb{E}[\|\boldsymbol{w(0)} \!-\! \boldsymbol{w^{*}}\|_{2}^{2}] \!+\! \frac{T \beta \eta^{3} [(N \!-\! 1)\rho_{H} \!+\! \rho_{L}]^{3}}{K^{3} \rho_{L}^{3}}(D^{2} \!+\! \frac{2 d K W^{2}}{\rho_{L} (\sum_{i=1}^{N} |\mathcal{D}_{i}|)^{2}})$} \nonumber \\
&\text{\small $+\frac{T \eta^{2} [(N \!-\! 1)\rho_{H} \!+\! \rho_{L}]^{2}}{K^{2} \rho_{L}^{2}}((T \!-\! 1) \|g^\top\|_{2} D \!+\! D^{2} \!+\! \frac{2 d K W^{2}}{\rho_{L} (\sum_{i=1}^{N} |\mathcal{D}_{i}|)^{2}}) \!+\! \frac{T \beta D^{2} \eta^{3} [(N \!-\! 1)\rho_{H} \!+\! \rho_{L}]}{K \rho_{L}}$} . \label{optimal_gap_non_convexity}
\end{align}
\vspace{5pt}
\end{minipage}
\end{figure*}

We provide the detailed proof of Theorem \ref{theorem_fixed_point} in Appendix~\ref{proof_theorem_fixed_point} and conclude the offline pretraining process in Algorithm~\ref{alg1}, which iteratively calculates the fixed point of mean-field term $\phi(t)$, effectively degrading the high computational complexity to the linear scale, denoted as $O(TNLM)$, where $M$ represents the number of iterative rounds required to achieve the fixed point. At $t$-th global training iteration, with the initial values for $\phi_{0}(t)$, $\rho_{i}^{0}$, $R_{t}^{0}$ and $\alpha_{i}^{t}$, $t \!\in\! \{0,1, \ldots, T\}$, $i \!\in\! \{1,2,  \ldots, N\}$, we calculate the privacy budget $\rho_{i}^{t}$ of each client based on the derivations in Theorem~\ref{theorem_optimal_alpha}~and Eq.~(\ref{update_constraint_refined}), which allows for the iterative refinement of the mean-field term $\phi(t)$. The iteration continues until $\phi_{m+1}(t) \!\rightarrow\! \phi_{m}(t)$ within the preset error $\epsilon$, signifying the identification of the fixed point. Note that, during the online training process, the central server and clients engage in the Stackelberg game at each iteration by leveraging the mean-field term $\phi(t)$. Both parties perform the strategic interaction with the optimal strategy profile $(\boldsymbol{R^{*}}, \{\boldsymbol{\alpha_{i}^{*}}$, $i \!\in\!  \{1,2, \ldots, N\}\})$, accordingly enhancing the performance of the FL global model. The complete workflow of our proposed framework FedPCS is showcased in Algorithm~\ref{alg2}.

\begin{remark}
Prior to the formal FL training, the central server precomputes the mean-field term and disseminates it to the clients along with the initial global model $\boldsymbol{w(0)}$. Consequently, the mean-field terms are regarded as known functions from the clients' view. Moreover, the model information required for mean-field terms calculation can be obtained from the historical training process, thereby mitigating information exchange and reducing the risk of privacy leakage in the formal FL training.
\end{remark}

\section{Convergence Analysis with Dynamic Client Sampling Probabilities} \label{Convergence_Analysis}
In this section, we conduct rigorous theoretical analysis for the convergence upper bound, thus establishing a robust foundation for our proposed FedPCS. To evaluate the FL global model accuracy convergence, we analyze the expected optimality gap $\mathbb{E}[\|\boldsymbol{w(T)} \!-\! \boldsymbol{w^{*}}\|_{2}^{2}]$, which serves as a critical metric and has been comprehensively utilized for the FL model performance assessment \cite{ding2023incentive}. We consider Assumptions~\ref{assumption_convexity}-\ref{assumption_subgradient} for convex and non-convex scenarios, and suppose $i$-th client's local gradient $\nabla F_{i}(\boldsymbol{w})$ satisfies Assumption~\ref{assumption_bounded_gradients} as in \cite{luo2022tackling}. Then, we introduce the following theorem that provides an upper bound for the optimality gap of the global model trained with the partial participation and $\rho$-$z$CDP mechanism, ultimately affecting the training performance.
\begin{theorem} \label{theorem_model_gap}
(Model Optimality Gap) After $T$ iterations of FL model training, for the convex scenario, the expected optimality gap of the global model is upper bounded by Eq.~(\ref{optimal_gap_convexity}); for the non-convex global loss function, the expected optimality gap of the global model is upper bounded by Eq.~(\ref{optimal_gap_non_convexity}).
\end{theorem}

The detailed proof of Theorem~\ref{theorem_model_gap} is presented in Appendix~\ref{proof_theorem_model_gap}. Theorem~\ref{theorem_model_gap} demonstrates that, for our proposed FedPCS algorithm, the upper bound is affected by the ratio $\frac{\rho_{H}}{\rho_{L}}$. The greater the ratio $\frac{\rho_{H}}{\rho_{L}}$, the smaller the upper bound of the optimal gap in Eqs.~(\ref{optimal_gap_convexity})-(\ref{optimal_gap_non_convexity}). Subsequently, based on Assumption~\ref{assumption_moment}, we proceed to formulate a rigorous convergence upper bound for the divergence of the global loss function between the consecutive iterations, which is summarized in the following proposition.
\begin{proposition} \label{proposition_difference}
(One Round Progress) Considering the~artificial Gaussian noise and the partial client participation scenario, the expected convergent upper bound of the global loss function for any consecutive global training iterations is given by:
\begin{align} \label{difference}
&\mathbb{E}[F(\boldsymbol{w(t)}) \!-\! F(\boldsymbol{w(t \!-\! 1)})] \nonumber \\
\leq& - \frac{\eta \kappa}{2} \|\nabla \! F(\boldsymbol{w(t \!-\! 1)})\|_{2}^{2} \!+\! \frac{d \beta \eta^{2}}{2} \sum\nolimits_{i=1}^{N} \! \theta_{i}^{2} \sigma_{i}^{2}(t \!-\! 1) \nonumber \\
&+ \frac{\beta \eta}{2 K}  \sum\nolimits_{i=1}^{N} \! \left[ \frac{\theta_{i}^{2}}{x_{i}^{t}} (D^{2} \!+\! d \sigma_{i}^{2}(t \!-\! 1))\right]  \!+\! \frac{M \beta \eta^{2} N}{2} .
\end{align}
\end{proposition}

The detailed proof of Proposition~\ref{proposition_difference} is provided in Appendix~\ref{proof_proposition_difference}. As delineated by Eq.~(\ref{difference}), it demonstrates that the one-round convergence analysis for our proposed algorithm FedPCS inherently relates to the $\ell_{2}$-norm of the gradient of the global loss function at iteration $t-1$, providing a reliable convergence upper bound. Leveraging the foundational results of Proposition~\ref{proposition_difference}, we extend our analysis to rigorously investigate the upper bounds on the convergence rate and convergence error, as given in the following theorem.
\begin{theorem}\label{theorem_dp_convergence} 
(Convergence Rate \& Convergence Error) With the learning rate satisfying $\eta \!\leq\! \frac{\kappa}{\beta (M_{V}+\kappa_{G}^{2})}$, the gap between the expected global loss function and the optimal global loss function (i.e., convergence rate) satisfies:
\begin{align} \label{convergence_rate}
& \mathbb{E}[F(\boldsymbol{w(T)})] \!-\! F(\boldsymbol{w^{*}}) \nonumber \\
\leq&\ (1 \!-\! \mu \eta \kappa)^{T} \! (\mathbb{E}[F(\boldsymbol{w(0)})] \!-\! F(\boldsymbol{w^{*}})) \!+\! \frac{\beta \eta}{2} \sum\nolimits_{t=0}^{T-1} \! (1 \!-\! \mu \eta \kappa)^{t}  \nonumber \\
& \times\!\left[M \! N \! \eta \!+\! \sum\nolimits_{i=1}^{N} \! \frac{\theta_{i}^{2}D^{2}}{K x_{i}^{t}} \!+\! (\frac{1}{K x_{i}^{t}} \!+\! \eta) d \theta_{i}^{2} \sigma_{i}^{2}(t) \right].
\end{align}
And the expected average-squared $\ell_{2}$-norm of the gradient of the global loss function (i.e., convergence error) satisfies:
\begin{align} \label{convergence_error}
&\frac{1}{T} \sum\nolimits_{t=0}^{T-1} \mathbb{E}[\|\nabla \! F(\boldsymbol{w(t)})\|_{2}^{2}]  \nonumber \\
\leq&\ \frac{2}{\eta \kappa T}(\mathbb{E}[F(\boldsymbol{w(0)})] \!-\! \mathbb{E}[F(\boldsymbol{w^{*}})]) + \frac{\beta \eta}{2T}\!\sum\nolimits_{t=0}^{T-1}  \left[ \frac{M \beta \eta^{2} N}{2} \right. \nonumber \\
&+ \left. \sum\nolimits_{i=1}^{N} \frac{\theta_{i}^{2}}{K x_{i}^{t}} (D^{2} \!+\! d \sigma_{i}^{2}(t)) \!+\! d \eta \theta_{i}^{2} \sigma_{i}^{2}(t) \right] .
\end{align}
\end{theorem}

The detailed proof of Theorem \ref{theorem_dp_convergence} is provided in Appendix~\ref{proof_theorem_dp_convergence}. It is evident that the upper bound of the convergence rate and convergence error both converge to the linear combination of the expected global loss function of the initial model $\boldsymbol{w(0)}$, the optimal model $\boldsymbol{w^{*}}$, and the predetermined parameters, which facilitates the selection of hyperparameters to enhance the convergence property. 
\begin{remark}
(Convex \& Non-Convex Applicability) The~above convergence analysis for our proposed framework FedPCS exclusively relies on the $\beta$-Lipschitz smoothness assumption, which indicates that the effectiveness of our results transcends the convexity paradigm of the global loss function. Specifically, through rigorous derivation, we demonstrate that the convergence upper bounds of FedPCS are applicable for both convex and non-convex domains, which broadens the applicability of our framework as a robust mechanism for addressing distributed optimization challenges. 
\end{remark}

\section{Efficiency Analysis of Equilibrium Strategy via Price of Anarchy Analysis} \label{poa_analysis}
In previous sections, we demonstrate the optimal strategy profile $(\boldsymbol{R^{*}}, \{\boldsymbol{\alpha_{i}^{*}}$, $i \in \{1,2, \ldots, N\}\})$ formulate the Stackelberg Nash Equilibrium (SNE). Subsequently, we further explore the efficiency of the Stackelberg Nash Equilibrium demonstrated in Theorem~\ref{SNE}, which provides profound insights into our proposed two-stage Stackelberg game approach in FL systems.  

\subsection{Social Welfare and Socially Optimal Strategy}

Firstly, we introduce the socially optimal strategy in the FL systems. Each client is tasked with determining the optimal correction factor $\alpha_{i}^{t}$, which will directly affect the privacy budget $\rho_{i}^{t}$ as shown in Eq.~(\ref{update_constraint_refined}). Given the allocated reward $R_{t}$, the long-term objective function under the socially optimal strategy (i.e., social welfare function) is defined as the sum of the profits of each client during total time horizons \cite{aguiar2021network, lu2021non}, thereby, the social welfare function is formulated~as:\footnote{In the following analysis, we consider the optimization problem on the privacy budget $\rho_{i}^{t}$ for convenience, which is equivalent to directly analyze on the correction factor $\alpha_{i}^{t}$ defined in Theorem \ref{theorem_optimal_alpha} according to Eq.~(\ref{update_constraint_refined}).} 
\begin{align} \label{social_welfare}
SW = \max_{\rho_{i}^{t}} \sum\nolimits_{t=0}^{T} \sum\nolimits_{i=1}^{N} \left[\rho_{i}^{t} R_{t} - \varphi_{i} (\rho_{i}^{t})^{2} \right] .
\end{align}

From Eq.~(\ref{social_welfare}), under the socially optimal strategy, the social welfare depends on the optimal privacy budgets of each client over time. Given that Eq.~(\ref{social_welfare}) is convex, consequently, we take the first-order and second-order of Eq.~(\ref{social_welfare}) regarding the privacy budget $\rho_{i}^{t}$ for each client $i \!\in\! \{1, 2, \ldots, N\}$ and obtain the following derivations:
\begin{align} \label{social_welfare_first_and_second_order}
\frac{\partial SW(\rho_{i}^{t})}{\partial \rho_{i}^{t}} = R_{t} - 2 \varphi_{i} \rho_{i}^{t},\ \frac{\partial^{2} SW(\rho_{i}^{t})}{\partial(\rho_{i}^{t})^{2}} = - 2 \varphi_{i} < 0.  
\end{align}

As the inequality $\partial^{2} SW(\rho_{i}^{t})/\partial(\rho_{i}^{t})^{2} \!<\! 0$ holds in Eq.~(\ref{social_welfare_first_and_second_order}), through calculating the first derivative of Eq.~(\ref{social_welfare}) regarding $\rho_{i}^{t}$ equals to 0, i.e., $\partial SW(\rho_{i}^{t})/\partial \rho_{i}^{t} = 0$, we can obtain the optimal privacy budget $\rho_{i}^{t}$ under social optimum as $\rho_{i}^{t*} \!=\! \frac{R_{t}}{2 \varphi_{i}}$. Consequently, the optimal social welfare is derived as follows:
\begin{align} \label{opt_social_welfare}
SW^{(opt)} = \frac{1}{4} \sum\nolimits_{t=0}^{T} R_{t}^{2} \sum\nolimits_{i=1}^{N} \frac{1}{\varphi_{i}}.
\end{align}

Following the definition of the price of anarchy (PoA) \cite{papadimitriou2001algorithms} and the social welfare defined in Eq.~(\ref{social_welfare}), we quantify the efficiency gap between the worst Stackelberg Nash equilibrium and the socially optimal strategies follows.  
\begin{definition} \label{definition_PoA}
(Price of Anarchy, PoA) The Price of Anarchy (PoA) is defined as the ratio of the social welfare achieved under the socially optimal strategy to that achieved by the worst-case Stackelberg Nash Equilibrium, i.e., 
\begin{align} \label{PoA}
PoA = \frac{SW^{(opt)}}{\min_{\substack{\rho_{i}^{t*}, i \in \{1, 2, \ldots, N\}, t \in \{0, 1, \ldots, T\}}} SW^{(Nash)}},
\end{align}
and a higher PoA indicates a greater inefficiency in the system. It is obvious that $\text{PoA} \geq 1$. 
\end{definition}

\subsection{Comparing Random Sampling to Social Optimum}
In this part, we explore the efficiency gap between the Nash Equilibrium achieved by the socially optimal strategy and the random sampling strategy. We prove that the gap is solely determined by the lower boundary of the privacy budgets and present the lower bound of PoA in the following proposition.
\begin{proposition} \label{proposition_random_sampling_poa}
As compared to the optimal social welfare in Eq.~(\ref{opt_social_welfare}), the price of anarchy under the random strategy achieves $\text{PoA}^{\text{(rand)}} \geq \frac{\sum\nolimits_{t=0}^{T} R_{t} \sum\nolimits_{i=1}^{N} \frac{1}{\varphi_{i}}}{4 N \rho_{L}(T+1)}$, which can be arbitrarily large when the lower bound of privacy budget $\rho_{L} \rightarrow 0$.
\end{proposition} 
\begin{proof}
We consider the Stackelberg Nash Equilibrium under the random sampling strategy. Generally, the participants in $\rho$-$z$CDP based FL systems are perceived as egocentric clients and selfishly determine their Privacy Protection Level (PPL), which is negatively correlated with the privacy budget \cite{xu2023personalized}, thereby maximizing their profits and preserving private local data \cite{hu2022incentive}. Without a well-designed monetary mechanism, in the worst case, egocentric clients may gradually decrease the privacy budget $\rho_{i}^{t}$ to the acceptable lowest boundary $\rho_{L}$ by dynamically adjusting the correction factor $\alpha_{i}^{t}$. Thus, from this insight, the social welfare under the Nash Equilibrium of the random sampling strategy is given by: 
\begin{align} \label{social_welfare_random}
SW^{(rand)} &= \sum\nolimits_{t=0}^{T} \sum\nolimits_{i=1}^{N} \left( \rho_{L} R_{t} - \varphi_{i} \rho_{L}^{2} \right) \nonumber \\
&= N\rho_{L} \sum\nolimits_{t=0}^{T} R_{t} - (T+1) \rho_{L}^{2} \sum\nolimits_{i=1}^{N} \varphi_{i} .
\end{align}

Subsequently, according to Definition \ref{definition_PoA} and the social~welfare derived in Eq.~(\ref{social_welfare_random}), the price of anarchy under the random sampling strategy is formulated as follows
\begin{align}
PoA^{(rand)} \!&=\! \frac{\frac{1}{4} \sum\nolimits_{t=0}^{T} R_{t}^{2} \sum\nolimits_{i=1}^{N} \frac{1}{\varphi_{i}}}{N\rho_{L} \sum\nolimits_{t=0}^{T} R_{t} - (T + 1) \rho_{L}^{2} \sum\nolimits_{i=1}^{N} \varphi_{i}} \nonumber \\
&\geq\! \frac{\sum\nolimits_{t=0}^{T} R_{t}^{2} \sum\nolimits_{i=1}^{N} \frac{1}{\varphi_{i}}}{4 N\rho_{L} \sum\nolimits_{t=0}^{T} R_{t}} \!\stackrel{(a)}{\geq}\! \frac{\sum\nolimits_{t=0}^{T} R_{t} \sum\nolimits_{i=1}^{N} \frac{1}{\varphi_{i}}}{4 N \rho_{L} (T \!+\! 1)}, 
\end{align}
where (a) follows Cauchy-Schwarz Inequality. Further, from the clients point, the allocated reward $R_{t}$, $t \in \{0, 1, \ldots, T\}$ is known parameters.  \end{proof}

According to Proposition \ref{proposition_random_sampling_poa}, the random sampling policy performs worse with lower privacy budget $\rho_{L}$. As $\rho_{L} \!\rightarrow\! 0$, $\text{PoA}^{\text{(rand)}}$ approaches infinity, and the learning efficiency of the FL systems becomes arbitrarily bad. Thus, it is critical to design an efficient incentive mechanism to enhance social~welfare.

\subsection{Comparing Privacy-aware Sampling to Social Optimum}
To motivate selfish clients to contribute their private local data and achieve a relatively high client Privacy Protection Level (PPL) simultaneously, we propose a privacy-aware client sampling strategy as shown in Eq.~(\ref{sgd_sample}) that successfully avoids the worst cases of $\text{PoA}^{\text{(rand)}}$ in Proposition \ref{proposition_random_sampling_poa}. In the following, we demonstrate how our strategy effectively bounds the PoA.
\begin{theorem} \label{theorem_privacy_poa}
Our privacy-aware sampling strategy in Eq.~(\ref{sgd_sample}) with the Stackelberg game-based incentive mechanism design achieves $\text{PoA}^{\text{(pri)}} \leq \frac{R_{\text{max}}}{2 N} \sum\nolimits_{i=1}^{N} \frac{1}{\varphi_{i}}$ when the upper boundary of the privacy budget $\rho_{H} \rightarrow + \infty$.
\end{theorem}
\begin{proof}
Leveraging the privacy-aware client sampling strategy delineated in Eq.~(\ref{sgd_sample}) and the two-stage Stackelberg game framework constructed in Eqs.~(\ref{central_server_obj})-(\ref{client_obj}), at each global iteration $t \in \{0, 1, \ldots, T\}$, the egocentric clients will progressively increase their privacy budget $\rho_{i}^{t}$ to maximize their utility and ensure selection by the central server, otherwise, the client's utility will default to zero if not sampled. Consequently, the worst-case scenario of Stackelberg Nash Equilibrium emerges as $\rho_{i}^{t} \rightarrow \rho_{H}$, with the corresponding social welfare under the privacy-aware sampling strategy expressed as follows:
\begin{align} \label{social_welfare_privacy}
SW^{(pri)} = N\rho_{H} \sum\nolimits_{t=0}^{T} R_{t} - (T + 1) \rho_{H}^{2} \sum\nolimits_{i=1}^{N} \varphi_{i} . 
\end{align}

From Definition \ref{definition_PoA} and the social welfare proposed in Eq.~(\ref{social_welfare_privacy}), the upper bound of the price of anarchy with the privacy-aware sampling policy is given as,
\begin{align} \label{poa_privacy}
PoA^{(pri)} &= \frac{\frac{1}{4} \sum\nolimits_{t=0}^{T} R_{t}^{2} \sum\nolimits_{i=1}^{N} \frac{1}{\varphi_{i}}}{N \rho_{H} \sum\nolimits_{t=0}^{T} R_{t} - (T + 1) \rho_{H}^{2} \sum\nolimits_{i=1}^{N} \varphi_{i}} \nonumber \\
&\stackrel{(a)}{\leq} \frac{\sum\nolimits_{t=0}^{T} R_{t}^{2} \sum\nolimits_{i=1}^{N} \frac{1}{\varphi_{i}}}{2(N \sum\nolimits_{t=0}^{T} R_{t} - \frac{(T+1)\sum\nolimits_{i=1}^{N} \varphi_{i}}{\rho_{H}})} \nonumber \\
&\stackrel{(b)}{\leq} \frac{R_{\text{max}} \sum\nolimits_{i=1}^{N} \frac{1}{\varphi_{i}}}{2(N - \frac{(T+1)\sum\nolimits_{i=1}^{N} \varphi_{i}}{\rho_{H} \sum\nolimits_{t=0}^{T} R_{t}})}. 
\end{align}
where (a) follows from the inequality $N \sum\nolimits_{t=0}^{T}  R_{t} \leq (T + 1) \\ \sum\nolimits_{i=1}^{N} \varphi_{i}$; (b) follows from the Cauchy-Schwarz Inequality that $\sum\nolimits_{t=0}^{T} R_{t}^{2} / \sum\nolimits_{t=0}^{T}\! R_{t} \!\leq\! R_{\text{max}}$, where $R_{\text{max}} \!=\! \max_{t \in \{0, 1, \ldots, T\}} R_{t}$ represents the maximum among allocated rewards during the FL training process.~\end{proof}

From Theorem~\ref{theorem_privacy_poa}, our privacy-aware client sampling strategy promotes the egocentric clients to prudently determine the values of the privacy budgets to avoid being unsampled at each global iteration and maximize personal utilities in the long term, which successfully addresses the worst-case in Proposition~\ref{proposition_random_sampling_poa}. Our FedPCS's PoA in Eq.~(\ref{poa_privacy}) decreases with the increment of the upper boundary of the privacy budgets $\rho_{H}$. As $\rho_{H} \rightarrow + \infty$, $\text{PoA}^{\text{(pri)}}$ converges to $\frac{R_{\text{max}}}{2 N} \sum\nolimits_{i=1}^{N} \frac{1}{\varphi_{i}}$.

\section{Adaptive Client Sampling Ratio Determination} \label{adaptive_sampling_ratio}
In this section, we extend our analysis and explore the more realistic scenario that considers dynamic privacy constraints for the optimization of the central server and egocentric clients. The privacy constraint indicates that the total privacy budget of sampled clients is confined within a time-varying range, which instructs the central server to adaptively regulate the dimension of sampled clients' subsets at each global iteration. In the following, we will introduce the reformulated optimization function of the central server and the corresponding adaptive optimal client sampling ratio and allocated rewards.

\begin{figure*}[ht]
\setlength{\abovecaptionskip}{2pt} 
    \centering
    \begin{minipage}{340pt}
    \centerline{\includegraphics[width=1.0\textwidth, trim=50 70 50 70,clip]{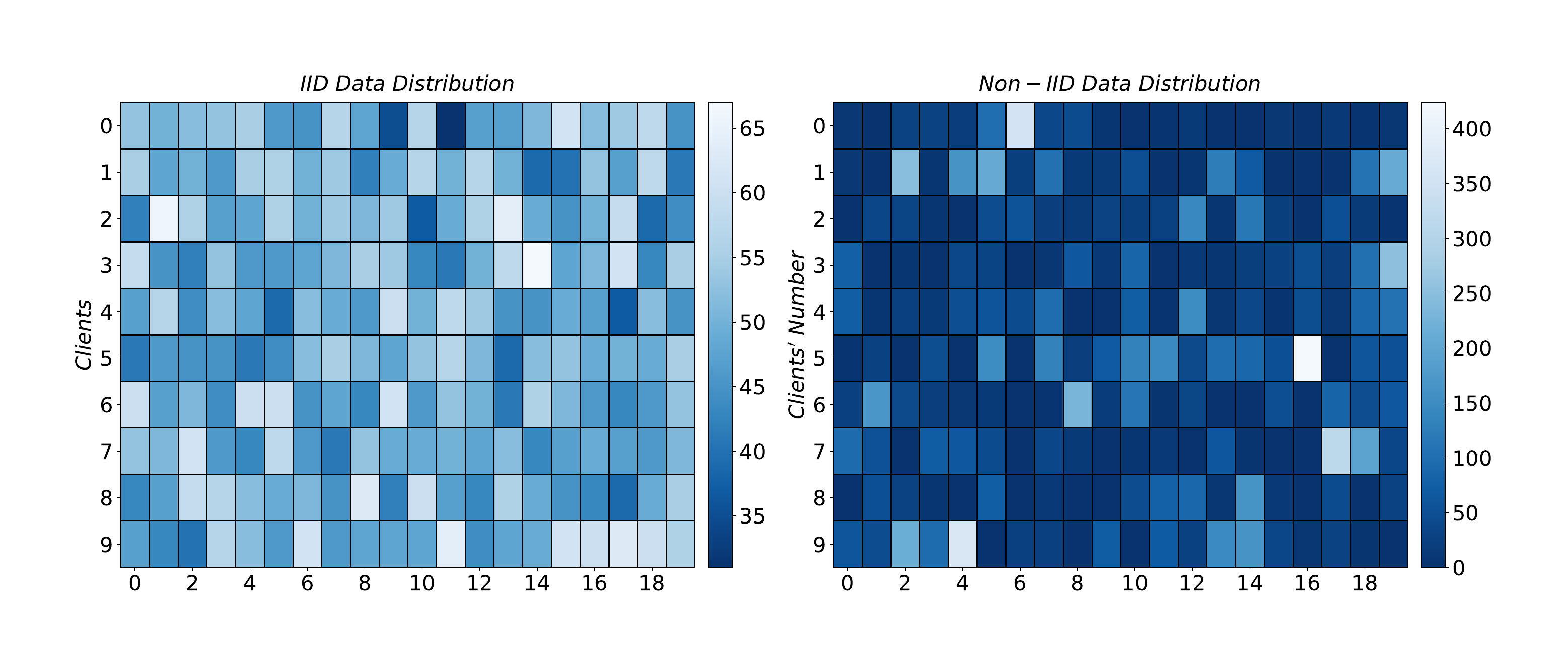}}
    \caption{The data distribution on the CIFAR-10 dataset under IID and Non-IID settings with the client sampling rate $\tau = 0.2$.}
    \label{data_distribution}
    \end{minipage}
    \hspace{2pt}
    \begin{minipage}{160pt}
    \centerline{\includegraphics[width=1.0\textwidth, trim=0 0 0 0,clip]{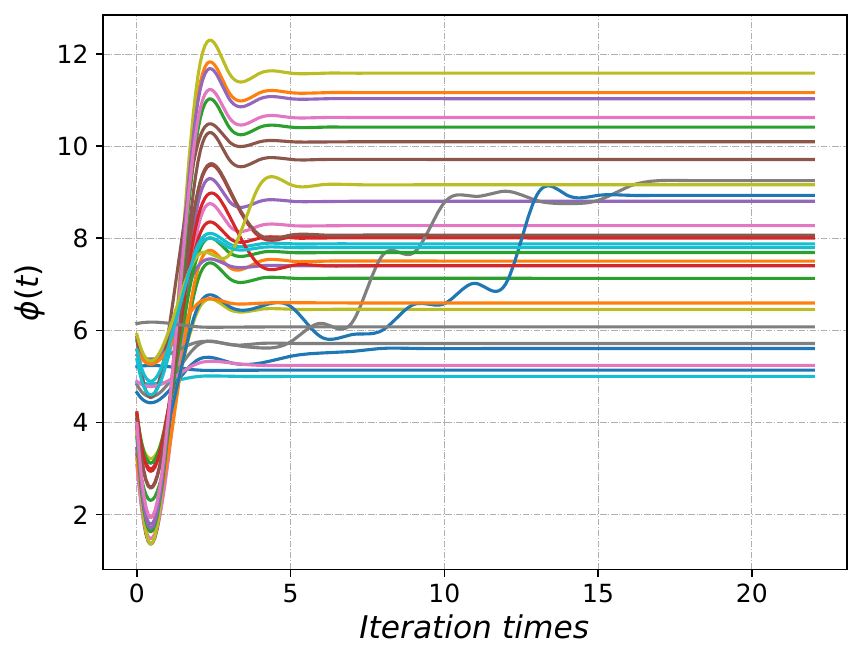}}
    \caption{The iterative convergence for the mean-field estimator $\phi(t)$.}
    \label{phi}
    \end{minipage}
\vspace{-12pt}
\end{figure*}

\subsection{Optimization Objection Reformulation for Central Server}
To satisfy the different requirements during the FL model training, and enhance the comprehensiveness and efficiency of the FL system while balancing privacy protection and model performance,\footnote{In different stages of FL model training, the privacy protection demands may vary. For instance, during the early stages, a higher privacy budget constraint may be needed to explore and acquire preliminary information on the local data, however, in the later stages, the privacy protection demands might become more stringent with a lower privacy budget constraint.} inspired by \cite{sun2022profit, sun2024socially}, we consider the dynamic scenario that the sampled clients’ total privacy budgets at each global iteration cannot exceed the time-various maximum privacy budget of the central server. Furthermore, based on the privacy budget constraint, the central server adaptively adjusts the optimal sampling ratio $\tau(t)$. Consequently, to tackle the cost-minimize problem in Eq.~(\ref{central_server_utility}) by optimizing the time-depend reward $R_{t}$ and the sampling ratio $\tau(t)$, we reformulate the optimization function of the central server as follows.
\begin{align}
U_{t}(R_{t},&\boldsymbol{\alpha}) \!=\! \min_{\tau(t), R_{t}} \sum\nolimits_{i \in \mathcal{K}_{t}} \!\! \left(\gamma \frac{\theta_{i}^{2}}{t |\mathcal{D}_{i}|^{2}\rho_{i}^{t}} \!+\! (1 \!-\! \gamma) R_{t} \rho_{i}^{t} \right), \label{central_server_utility_with_constraint} \\
& s.t.\ \sum\nolimits_{i \in \mathcal{K}_{t}} \rho_{i}^{t} \leq  B_{t}, \  |\mathcal{K}_{t}| = K_{t}= \tau(t) N, \label{central_server_constraint}
\end{align}
where Eq.~(\ref{central_server_constraint}) represents the time-various privacy budget constraint condition and we denote the affordable privacy upper bound by $B_{t}$. In the following, by analyzing the optimization problem in Eqs.~(\ref{central_server_utility_with_constraint})-(\ref{central_server_constraint}) similar to Theorem~\ref{theorem_optimal_alpha} and Lemma~\ref{lemma_optimal_reward}, we derive the optimal dynamic client subset and allocated reward at each global iteration.

\begin{algorithm}[t]
\small
\caption{Adaptive Client Sampling in FedPCS}
\label{alg3}
\SetAlgoLined
\KwIn{$N$, $\tau$, $T$, $\phi(t)$, learning rate $\eta$, weight $\theta_{i}$ $\varphi_{i}$, $\gamma$ and privacy constraint $B_{t}$, where $t \!\in\! \{0, 1, \ldots, T\}$.}
\KwOut{The global model parameter $\boldsymbol{w(T)}$.}
\For{$t \leftarrow 0, 1, \ldots, T$}{
    The central server decides the optimal reward $R_{t}^{*}$ and the adaptive client sampling ratio $\tau^{*}(t)$ based on Eq.~(\ref{opt_reward_and_tau})\tcp*[l]{Adaptive Sampling Ratio Calculation}
    The central server samples a client subset $\mathcal{K}_{t}$ according to the probability vector $\boldsymbol{X_{t}}$\tcp*[r]{Client Sampling} 
    Each client decides the optimal correction factor $\alpha_{i}^{t*}$\;
    \ForEach{sampled client}{
        Local training: $\ \boldsymbol{w_{i}(t \!+\! 1)} \!=\! \boldsymbol{w(t)} \!-\! \eta \nabla F_{i}(\boldsymbol{w(t)})$ \;
        Gradient Perturbation: $\nabla \widetilde{F}_{i}(\boldsymbol{w(t)}) \!=\! \nabla F_{i}(\boldsymbol{w(t)}) \!+\! \boldsymbol{n_{i}(t)}$ \;
    }
    The central server updates the global model with: $w(t + 1) = w(t) - \eta \sum_{i=1}^{K} \! \frac{\theta_{i}}{K_{t} x_{i}^{t}} \nabla \widetilde{F}_{i}(\boldsymbol{w(t)})$  \tcp*[r]{Model Aggregation}
}
\end{algorithm}

\subsection{Adaptive Sampling Ratio and Reward Determination}
Under the privacy budget constraint, the central server aims to find the optimal reward and sampling ratio at $t$-th iteration simultaneously to minimize the cost. Note that the objective function with dynamic constraints in Eqs.~(\ref{central_server_utility_with_constraint})-(\ref{central_server_constraint}) is strictly concave and the constraint set is compact and convex. Thus, we construct a special Lagrangian function that satisfies the Karush–Kuhn–Tucker (KKT) optimality conditions \cite{bertsekas1997nonlinear} to resolve the optimization problem in Eqs.~(\ref{central_server_utility_with_constraint})-(\ref{central_server_constraint}). To obtain the closed-form of optimal reward $R_{t}^{*}$ and sampling ratio $\tau^{*}(t)$, here we analyze the homogeneous scenarios\footnote{Heterogeneous scenarios can also be efficiently solved with numerical optimization tools such as Scipy and CVXPY libraries.} as summarized in the following theorem. 
\begin{theorem} \label{theorem_opt_reward_and_tau}
(Optimal Reward \& Optimal Sampling Ratio)~In Stage~\uppercase\expandafter{\romannumeral1}, with $\upsilon = \frac{\gamma \theta^{2}}{t |\mathcal{D}|^{2}}$, the optimal reward $R_{t}^{*}$ and the optimal sampling ratio $\tau^{*}(t)$ of the central server at $t$-th global iteration is given by:
\begin{align} \label{opt_reward_and_tau}
&R_{t}^{*} = \frac{1}{G^{t}} \left(\frac{2 G^{t} \upsilon}{1 - \gamma}\right)^{\frac{1}{3}} - \frac{H^{t}}{G^{t}},  \\
&\tau^{*}(t) = \frac{1}{N} \left\lfloor B_{t}\left(\frac{2 G^{t} \upsilon}{1 - \gamma}\right)^{-\frac{1}{3}} \right\rfloor. 
\end{align}
\end{theorem}
\begin{proof}
We construct the Lagrangian function to facilitate the optimization of the optimal reward and sampling ratio at $t$-th iteration with the Lagrange multiplier $\delta$ as follows. 
\begin{align}
\!\!\!\! \text{\small $\mathcal{L}(K_{t},R_{t}) \!=\!\! \sum\nolimits_{i=1}^{K_{t}}\!\! \left[\! \frac{\gamma \theta^{2}}{t |\mathcal{D}|^{2}\rho_{i}^{t}} \!+\! (1 \!-\! \gamma) R_{t}\rho_{i}^{t} \!\right] \!\!+\! \delta \! \left(\sum\nolimits_{i=1}^{K_{t}} \! \rho_{i}^{t} \!-\! B_{t} \! \right) $} . \!\!\!
\end{align}

\begin{table*}[t]
\centering
\setlength{\abovecaptionskip}{3pt} 
\renewcommand\arraystretch{1.0}
\caption{Accuracy Comparison on Different Datasets.}
\resizebox{1.0\textwidth}{!}{\begin{tabular}
{c|c|c|c|c|c|c|c|c|c|c|c} 
\toprule[1.2pt]
\multirow{2}{*}{\multirowcell{2}{\centering\textbf{Datasets}}} & \multirow{2}{*}{\multirowcell{2}{\centering\textbf{Distribution}}} & \multicolumn{2}{c|}{\parbox{2cm}{\centering\textbf{Random}}} & \multicolumn{2}{c|}{\parbox{2cm}{\centering\textbf{AOCS}\cite{chen2022optimal}}} & \multicolumn{2}{c|}{\parbox{2cm}{\centering\textbf{Fed-CBS}\cite{zhang2023fed}}} & \multicolumn{2}{c|}{\parbox{2cm}{\centering\textbf{DELTA}\cite{wang2024delta}}} & \multicolumn{2}{c}{\parbox{2cm}{\centering\textbf{FedPCS}(Ours)}}    \\ \cmidrule[0.5pt](l{1pt}r{0pt}){3-12}

&  & \parbox{0.5cm}{{\centering{IID}}} & \parbox{1cm}{\centering Non-IID} & \parbox{0.5cm}{{\centering IID}} & \parbox{1cm}{\centering Non-IID} & \parbox{0.5cm}{{\centering IID}} & \parbox{1cm}{\centering Non-IID} & \parbox{0.5cm}{{\centering IID}} & \parbox{1cm}{\centering Non-IID} & \parbox{0.5cm}{{\centering IID}} & \parbox{1cm}{\centering Non-IID} \\ \cmidrule[0.8pt](l{1pt}r{0pt}){1-12}

\multirow{3}{*}{\multirowcell{2}{\textbf{Fashion} \\ \textbf{-MNIST}}}
        & $\tau=0.2$ & 74.69$\pm$0.78 & 72.79$\pm$0.43 & 78.62$\pm$0.26 & 76.32$\pm$0.31 & 78.87$\pm$0.34 & 76.88$\pm$0.25 & 82.71$\pm$0.11 & 80.31$\pm$0.07 & \textbf{84.43$\pm$0.18} & \textbf{83.35$\pm$0.22} \\ \cmidrule[0.5pt](l{1pt}r{0pt}){2-12}

        & $\tau=0.5$ & 80.51$\pm$0.07 & 77.69$\pm$0.08 & 82.25$\pm$0.29 & 80.95$\pm$0.19 & 82.63$\pm$0.54 & 80.43$\pm$0.49 & 83.76$\pm$0.16 & 81.48$\pm$0.23 & \textbf{86.20$\pm$0.34} & \textbf{85.14$\pm$0.28} \\ \cmidrule[0.5pt](l{1pt}r{0pt}){2-12}
        
        & $\tau=0.8$ & 82.20$\pm$0.22 & 80.28$\pm$0.09 & 82.23$\pm$0.02 & 81.17$\pm$0.06 & 83.72$\pm$0.10 & 81.59$\pm$0.04 & 84.14$\pm$0.48 & 83.05$\pm$0.37 & \textbf{87.33$\pm$0.16} & \textbf{85.82$\pm$0.13} \\ \midrule[0.8pt]

\multirow{3}{*}{\multirowcell{2}{\textbf{CIFAR-10}}}
        & $\tau=0.2$ & 56.31$\pm$0.21 & 54.76$\pm$0.38 & 58.63$\pm$0.16 & 56.35$\pm$0.29 & 57.92$\pm$0.17 & 56.17$\pm$0.08 & 60.91$\pm$0.29 & 59.03$\pm$0.14  & \textbf{62.21$\pm$0.15} & \textbf{60.56$\pm$0.10} \\ \cmidrule[0.5pt](l{1pt}r{0pt}){2-12}

        & $\tau=0.5$ & 60.34$\pm$0.07 & 58.75$\pm$0.12 & 63.26$\pm$0.07 & 61.13$\pm$0.09 & 61.81$\pm$0.29 & 60.88$\pm$0.17 & 63.64$\pm$0.39 & 62.02$\pm$0.20 & \textbf{65.39$\pm$0.40} & \textbf{63.82$\pm$0.18} \\ \cmidrule[0.5pt](l{1pt}r{0pt}){2-12}
        
        & $\tau=0.8$ & 62.19$\pm$0.16 & 61.13$\pm$0.20 & 64.10$\pm$0.22 & 63.11$\pm$0.19 & 65.08$\pm$0.57 & 63.62$\pm$0.43 & 63.91$\pm$0.26 & 61.87$\pm$0.14 & \textbf{67.18$\pm$0.39} & \textbf{65.50$\pm$0.26} \\  \midrule[0.8pt]

\multirow{3}{*}{\multirowcell{2}{\textbf{SVHN}}}
        & $\tau=0.2$ & 82.64$\pm$0.04 & 80.63$\pm$0.06 & 84.91$\pm$0.07 & 82.44$\pm$0.12 & 83.69$\pm$0.04 & 81.77$\pm$0.02 & 85.76$\pm$0.04 & 83.64$\pm$0.03 & \textbf{87.53$\pm$0.05} & \textbf{86.06$\pm$0.08} \\ \cmidrule[0.5pt](l{1pt}r{0pt}){2-12}

        & $\tau=0.5$ & 83.53$\pm$0.03 & 81.94$\pm$0.04 & 85.29$\pm$0.06 & 83.18$\pm$0.08 & 84.86$\pm$0.05 & 82.92$\pm$0.02 & 86.29$\pm$0.03 & 84.57$\pm$0.05 & \textbf{87.79$\pm$0.03} & \textbf{85.92$\pm$0.04} \\ \cmidrule[0.5pt](l{1pt}r{0pt}){2-12}
        
        & $\tau=0.8$ & 82.09$\pm$0.02 & 81.12$\pm$0.03 & 85.68$\pm$0.03 & 83.85$\pm$0.13 & 83.05$\pm$0.24 & 82.17$\pm$0.20 & 86.87$\pm$0.04 & 84.95$\pm$0.09 & \textbf{87.81$\pm$0.02} & \textbf{86.15$\pm$0.02} \\  \midrule[0.8pt]

\multirow{3}{*}{\multirowcell{2}{\textbf{CIFAR-100}}}
        & $\tau=0.2$ & 41.34$\pm$0.06 & 40.27$\pm$0.39 & 44.36$\pm$0.29 & 41.10$\pm$0.38 & 44.87$\pm$0.10 & 43.32$\pm$0.15 & 46.17$\pm$0.29 & 44.65$\pm$0.23 & \textbf{47.58$\pm$0.15} & \textbf{46.28$\pm$0.29} \\ \cmidrule[0.5pt](l{1pt}r{0pt}){2-12}

        & $\tau=0.5$ & 45.38$\pm$0.09 & 44.64$\pm$0.08 & 46.70$\pm$0.09 & 46.51$\pm$0.19 & 47.35$\pm$0.07 & 46.52$\pm$0.28 & 47.36$\pm$0.22 & 46.89$\pm$0.37 & \textbf{48.72$\pm$0.12} & \textbf{47.92$\pm$0.13} \\ \cmidrule[0.5pt](l{1pt}r{0pt}){2-12}
        
        & $\tau=0.8$ & 47.07$\pm$0.01 & 46.65$\pm$0.22 & 47.36$\pm$0.28 & 45.96$\pm$0.11 & 47.87$\pm$0.08 & 46.78$\pm$0.06 & 47.71$\pm$0.23 & 46.57$\pm$0.14 & \textbf{49.37$\pm$0.02} & \textbf{47.50$\pm$0.05} \\  \midrule[0.8pt]

\multirow{3}{*}{\multirowcell{2}{\textbf{CINIC-10}}}
        & $\tau=0.2$ & 40.81$\pm$0.22 & 40.39$\pm$0.25 &42.70$\pm$0.14 & 42.27$\pm$0.33 & 42.69$\pm$0.11 & 41.20$\pm$0.07 & 45.03$\pm$0.26 & 42.56$\pm$0.19 & \textbf{47.47$\pm$0.23} & \textbf{44.41$\pm$0.25} \\ \cmidrule[0.5pt](l{1pt}r{0pt}){2-12}

        & $\tau=0.5$ & 45.30$\pm$0.19 & 44.53$\pm$0.22 & 47.26$\pm$0.03 & 44.83$\pm$0.12 & 46.70$\pm$0.18 & 45.21$\pm$0.25 & 48.23$\pm$0.25 & 45.41$\pm$0.22 & \textbf{49.35$\pm$0.09} & \textbf{47.08$\pm$0.08} \\ \cmidrule[0.5pt](l{1pt}r{0pt}){2-12}
        
        & $\tau=0.8$ & 47.62$\pm$0.22 & 45.65$\pm$0.14 & 47.86$\pm$0.29 & 46.24$\pm$0.30 & 47.56$\pm$0.13 & 46.43$\pm$0.27 & 48.85$\pm$0.14 & 46.81$\pm$0.21 & \textbf{50.01$\pm$0.12} & \textbf{48.76$\pm$0.04} \\  \midrule[0.8pt]

\multirow{3}{*}{\multirowcell{2}{\textbf{Tiny} \\ \textbf{-ImageNet}}}
        & $\tau=0.2$ & 60.55$\pm$0.10 & 58.86$\pm$0.02 & 63.15$\pm$0.24 & 61.60$\pm$0.12 & 62.72$\pm$0.15 & 61.97$\pm$0.13 & 64.38$\pm$0.11 & 63.96$\pm$0.20 & \textbf{66.38$\pm$0.04} & \textbf{65.71$\pm$0.22} \\ \cmidrule[0.5pt](l{1pt}r{0pt}){2-12}

        & $\tau=0.5$ & 61.11$\pm$0.21 & 60.62$\pm$0.11 & 64.67$\pm$0.17 & 64.57$\pm$0.12 & 63.78$\pm$0.13 & 63.28$\pm$0.20 & 66.13$\pm$0.22 & 64.98$\pm$0.29 & \textbf{67.50$\pm$0.07} & \textbf{66.28$\pm$0.09} \\ \cmidrule[0.5pt](l{1pt}r{0pt}){2-12}
        
        & $\tau=0.8$ & 61.92$\pm$0.11 & 61.58$\pm$0.30 & 66.33$\pm$0.09 & 65.71$\pm$0.06 & 64.59$\pm$0.08 & 64.62$\pm$0.15 & 66.61$\pm$0.06 & 66.34$\pm$0.19 & \textbf{68.18$\pm$0.12} & \textbf{67.85$\pm$0.05} \\
         
\bottomrule[1.2pt]
\end{tabular}}
\vspace{-12pt}
\label{accuracy}
\end{table*}

\begin{figure*}[ht]
\setlength{\abovecaptionskip}{3pt} 
\centering
\subfloat[Correlation of $\rho_{i}^{t}$ with Random]{
    \label{Correlation_random}
    \includegraphics[width=0.245\textwidth, trim=5 5 5 5,clip]{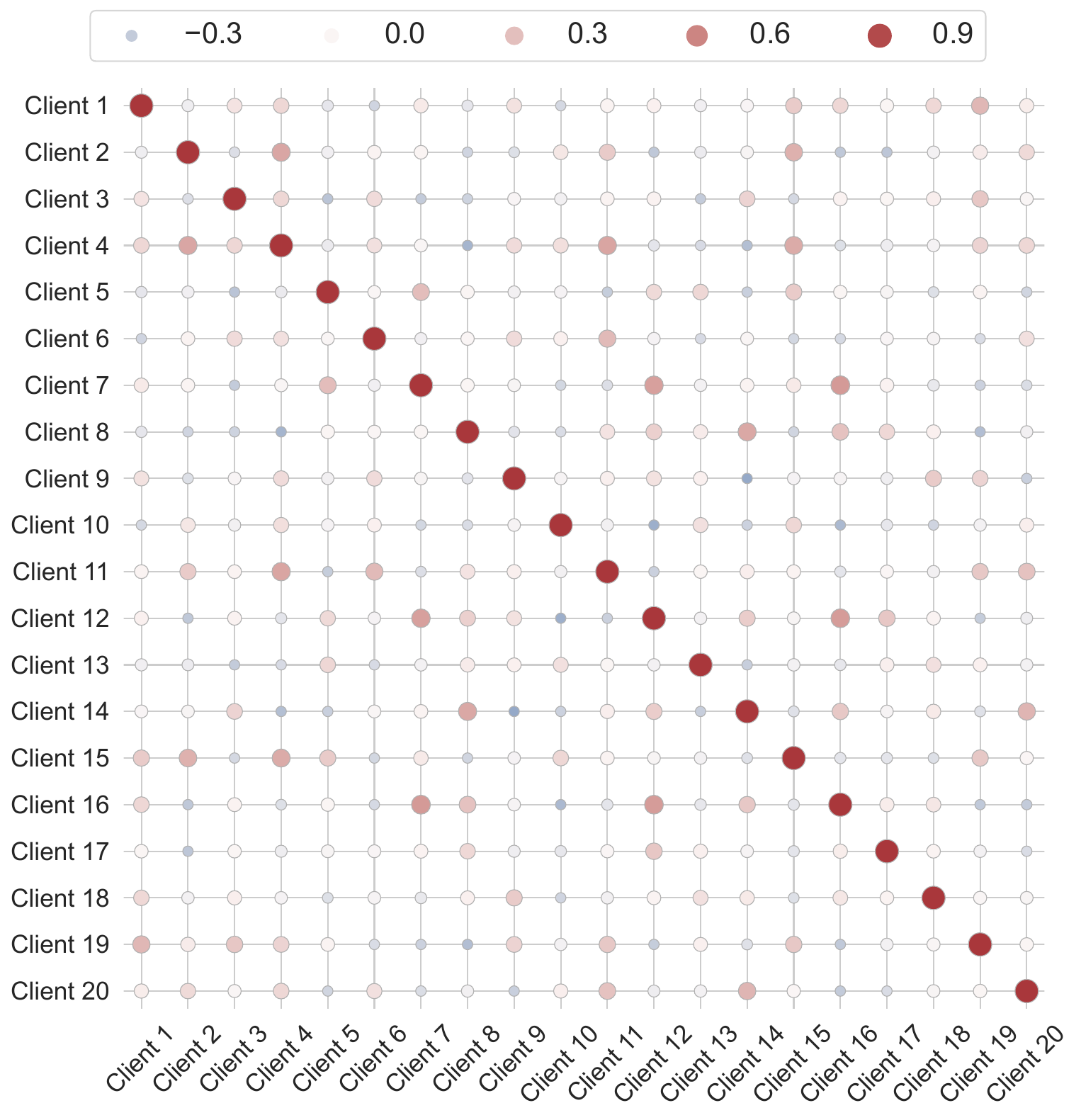}}
\subfloat[Correlation of $\rho_{i}^{t}$ with FedCBS]{
    \label{Correlation_cbs}
    \includegraphics[width=0.245\textwidth, trim=5 5 5 5,clip]{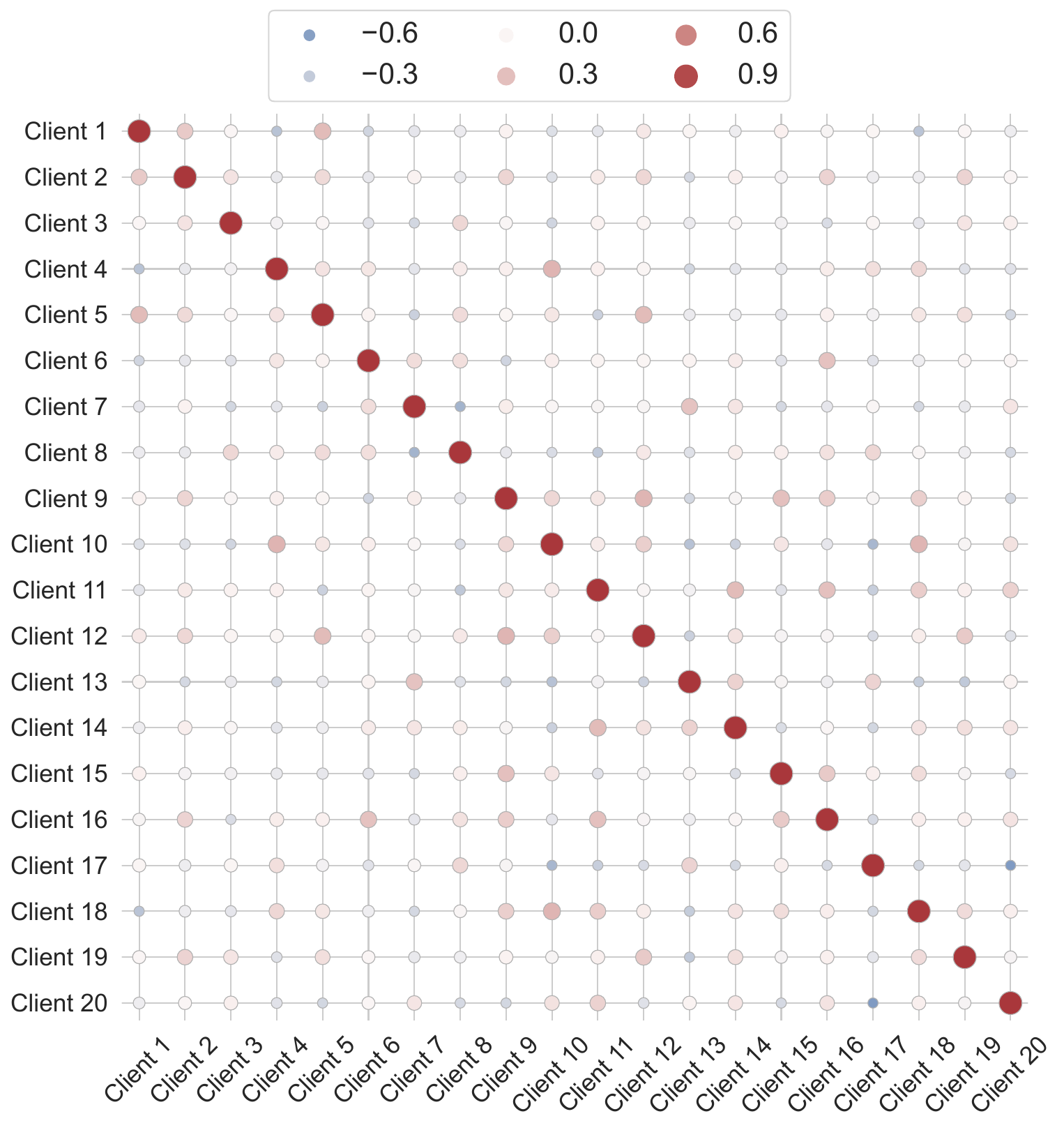}}
\subfloat[Correlation of $\rho_{i}^{t}$ with DELTA]{
    \label{Correlation_delta}
    \includegraphics[width=0.245\textwidth, trim=5 5 5 5,clip]{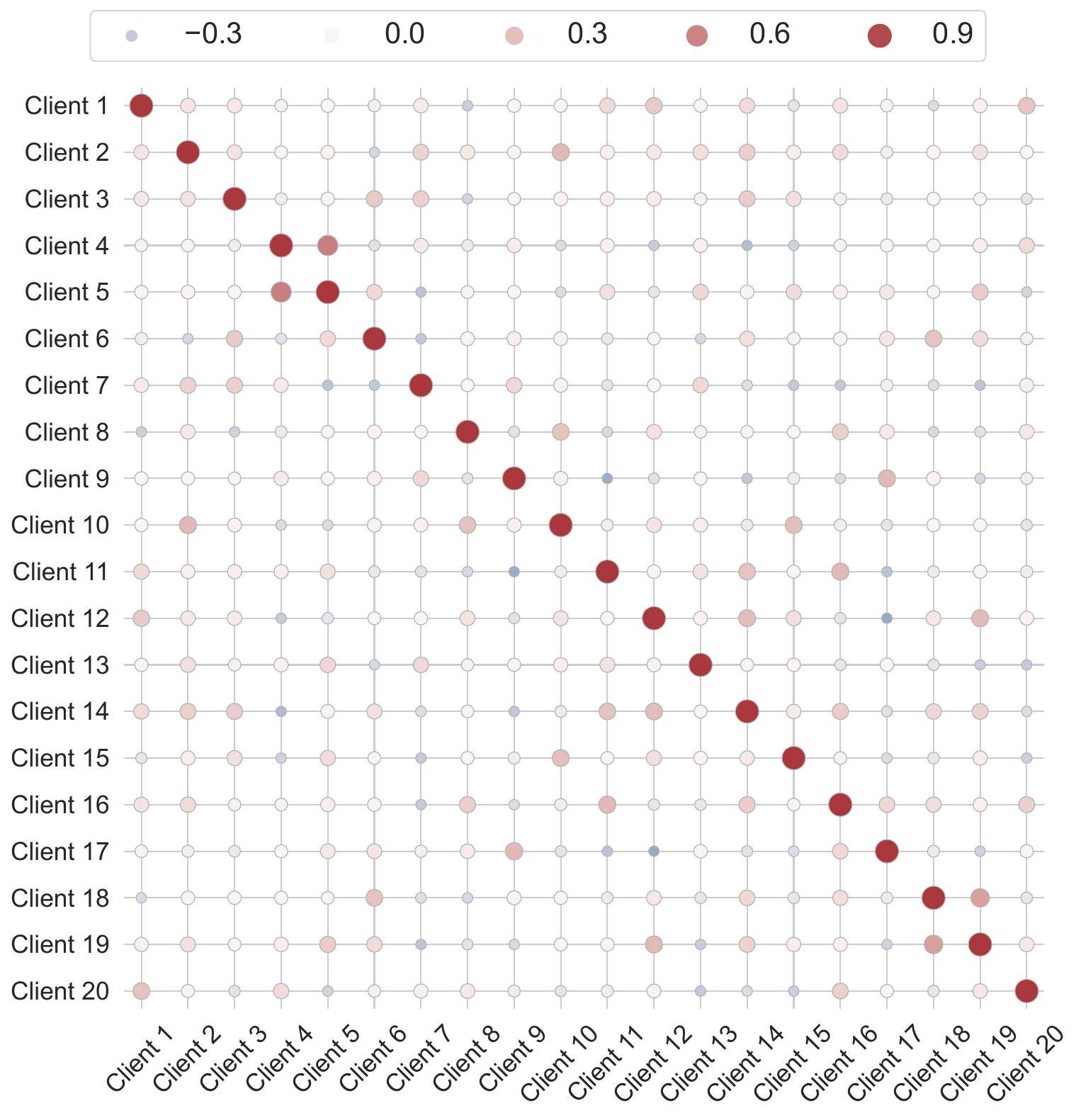}}
\subfloat[Correlation of $\rho_{i}^{t}$ with FedPCS]{
    \label{Correlation_pcs}
    \includegraphics[width=0.245\textwidth, trim=5 5 5 5,clip]{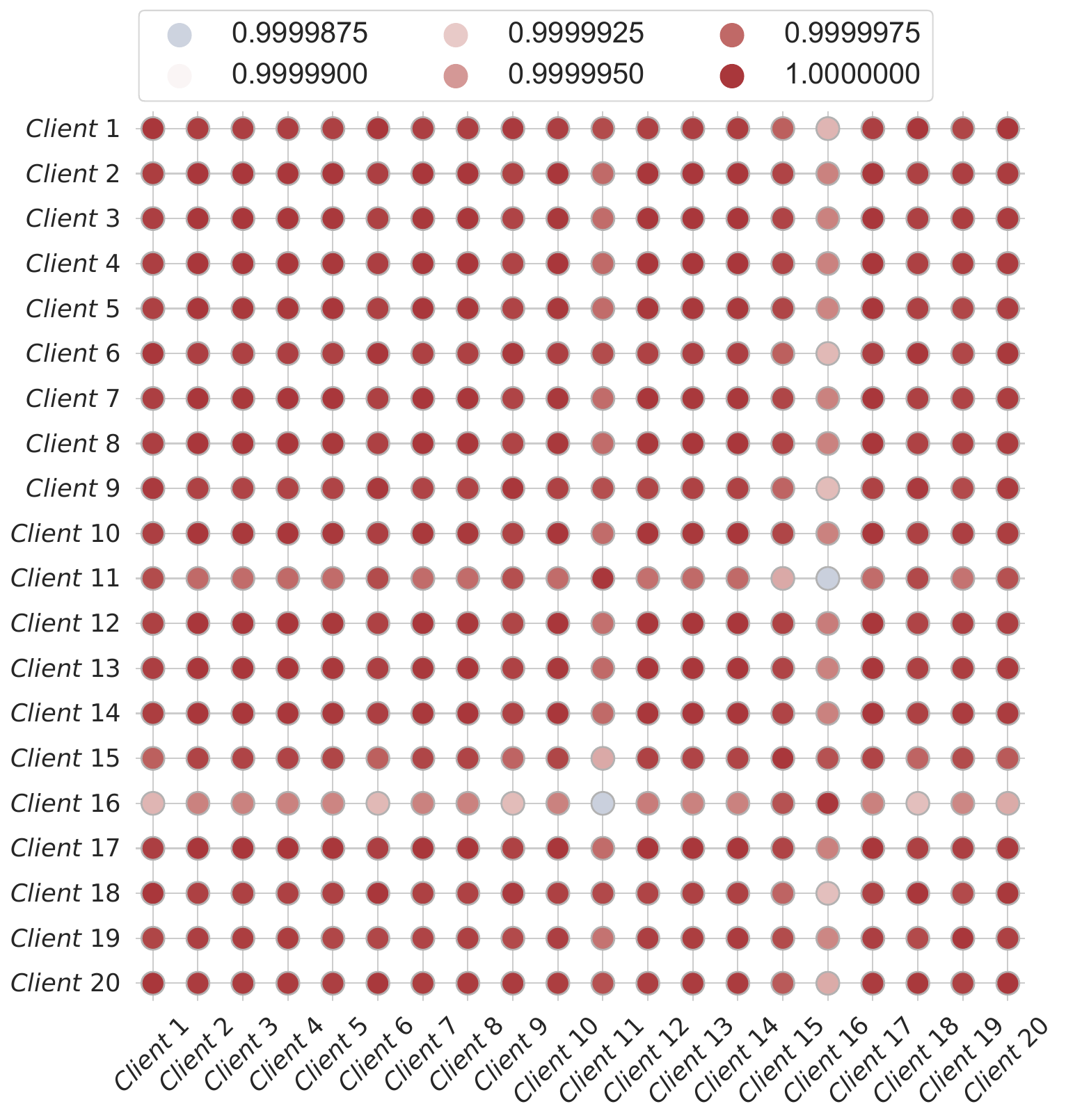}}
\caption{The Pearson correlation coefficient of various schemes on Fashion-MNIST dataset with sampling rate $\tau = 0.2$.}
\label{correlation}
\vspace{-12pt}
\end{figure*}

Then, as the optimization function of the central server in Eq.~(\ref{central_server_utility_with_constraint}) is concave, thus its optimal solution satisfies the KKT optimality conditions. According to the equation $\rho_{i}^{t} \!=\! G_{i}^{t} R_{t} \!+\! H_{i}^{t}$ in Lemma \ref{lemma_optimal_reward} and the stationary point condition, we derive the first-order derivatives of the Lagrangian function $\mathcal{L}(K_{t}, R_{t})$ regarding the reward $R_{t}$ and the size of client subset $K_{t}$:
\begin{align}\label{stationary_point_condition}
\frac{\partial \mathcal{L}}{\partial R_{t}} \!&=\! \frac{- G^{t} \upsilon}{(G^{t} R_{t} \!+\! H^{t})^{2}} \!+\! (1 \!-\! \gamma) (2 G^{t} R_{t} \!+\! H^{t}) \!+\! \delta G^{t} \!=\! 0, \\
\frac{\partial \mathcal{L}}{\partial K_{t}} \!&=\! \frac{\upsilon}{G^{t} R_{t} \!+\! H^{t}} \!+\! \left[(1 \!-\! \gamma)R_{t} \!+\! \delta \right] (G^{t} R_{t} \!+\! H^{t}) \!=\! 0.
\end{align}
where  and the second-order derivatives regarding the reward $\partial \mathcal{L}^{2}/\partial^{2} R_{t} > 0$ holds. Furthermore, based on the complementary slackness condition, the Lagrange multiplier satisfies $\delta \!\geq\! 0$ (If $\delta \!=\! 0$, the constraints on the feasible region become ineffective, indicating that the equivalent formulation is unconstrained) and the following equation holds:
\begin{align}\label{complementary_slackness}
\text{$\delta \left(\sum\nolimits_{i=1}^{K_{t}} \rho_{i}^{t} - B_{t} \right) = 0 $} .
\end{align}

From Eqs.~(\ref{stationary_point_condition})-(\ref{complementary_slackness}) and with the Lagrange multiplier $\delta = \upsilon (\frac{2 G^{t} \upsilon}{1 - \gamma})^{- \frac{2}{3}} - \frac{1 - \gamma}{G^{t}}(2(\frac{2 G^{t} \upsilon}{1 - \gamma})^{\frac{1}{3}}- H^{t}) \geq 0$ holds, we can obtain the expression of the optimal reward $R_{t}^{*}$ and the optimal sampling ratio $\tau^{*}(t)$ in Eq.~(\ref{opt_reward_and_tau}).  \end{proof}

According to Theorem \ref{theorem_opt_reward_and_tau},  with the optimal strategy profile, participating clients achieve their maximal utilities and the central server reaches its minimum cost simultaneously. Therefore, we can conclude that the two-stage Stackelberg game can still possess a Stackelberg Nash Equilibrium as demonstrated in Theorem~\ref{theorem_fixed_point} under the privacy budget constraint. The interaction between the central server and clients with the dynamic client subset and time-various privacy constraints within the FL system is summarized in Algorithm \ref{alg3}.

\begin{figure*}[t]
\setlength{\abovecaptionskip}{2pt} 
    \centering
    \begin{minipage}{350pt}
        \begin{minipage}{350pt}
            \centerline{\includegraphics[width=1.0\textwidth, trim=0 5 0 5,clip]{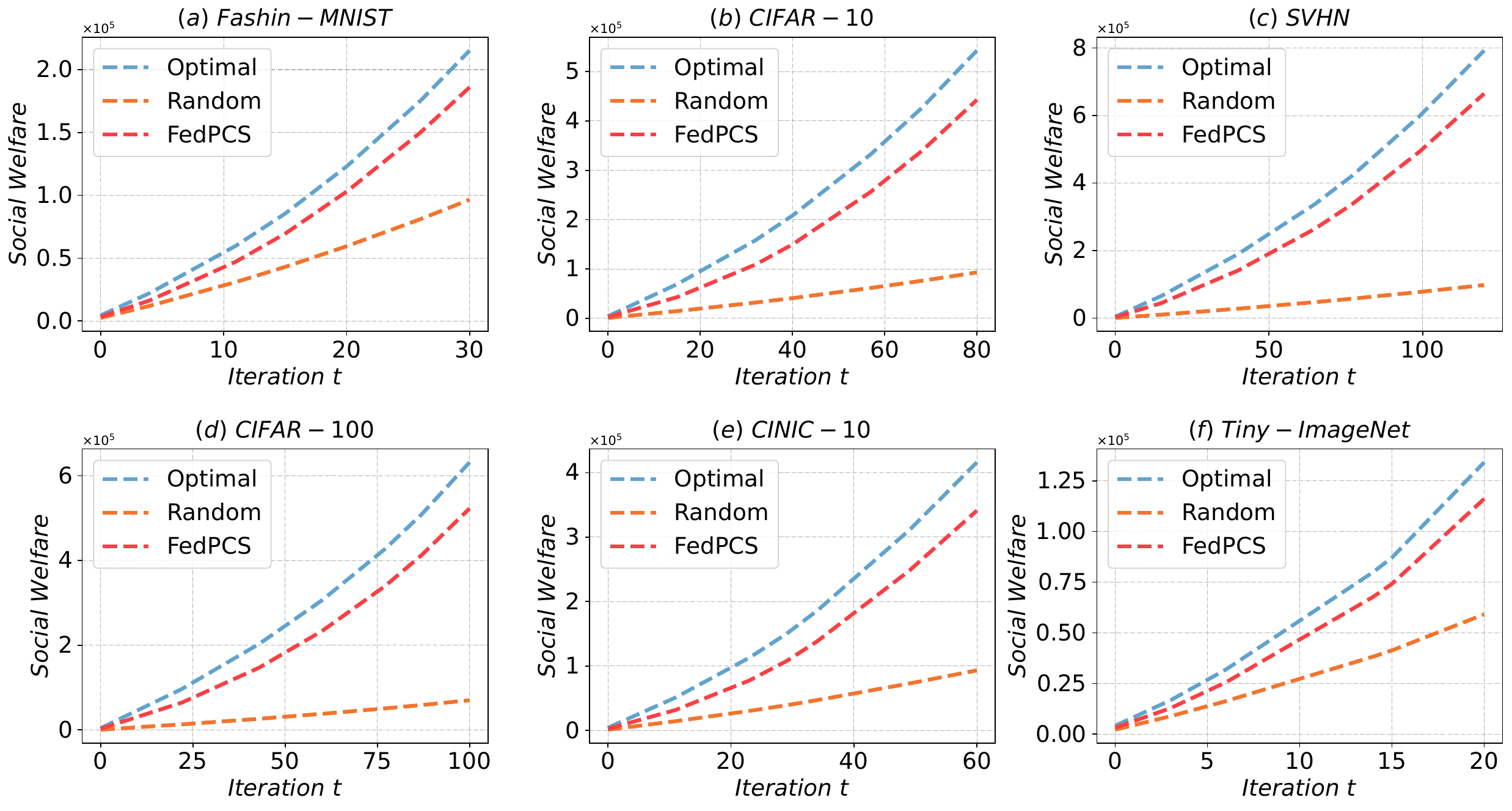}}
            \caption{The social welfare comparison on Fashion-MNIST/CIFAR-10/SVHN/CIFAR-100/CINIC-10/Tiny-ImageNet datasets under different client sampling strategies.}
            \label{social_welfare_comparison}
        \end{minipage}
        \begin{minipage}{350pt}
            \centerline{\includegraphics[width=1.0\textwidth, trim=50 70 50 35,clip]{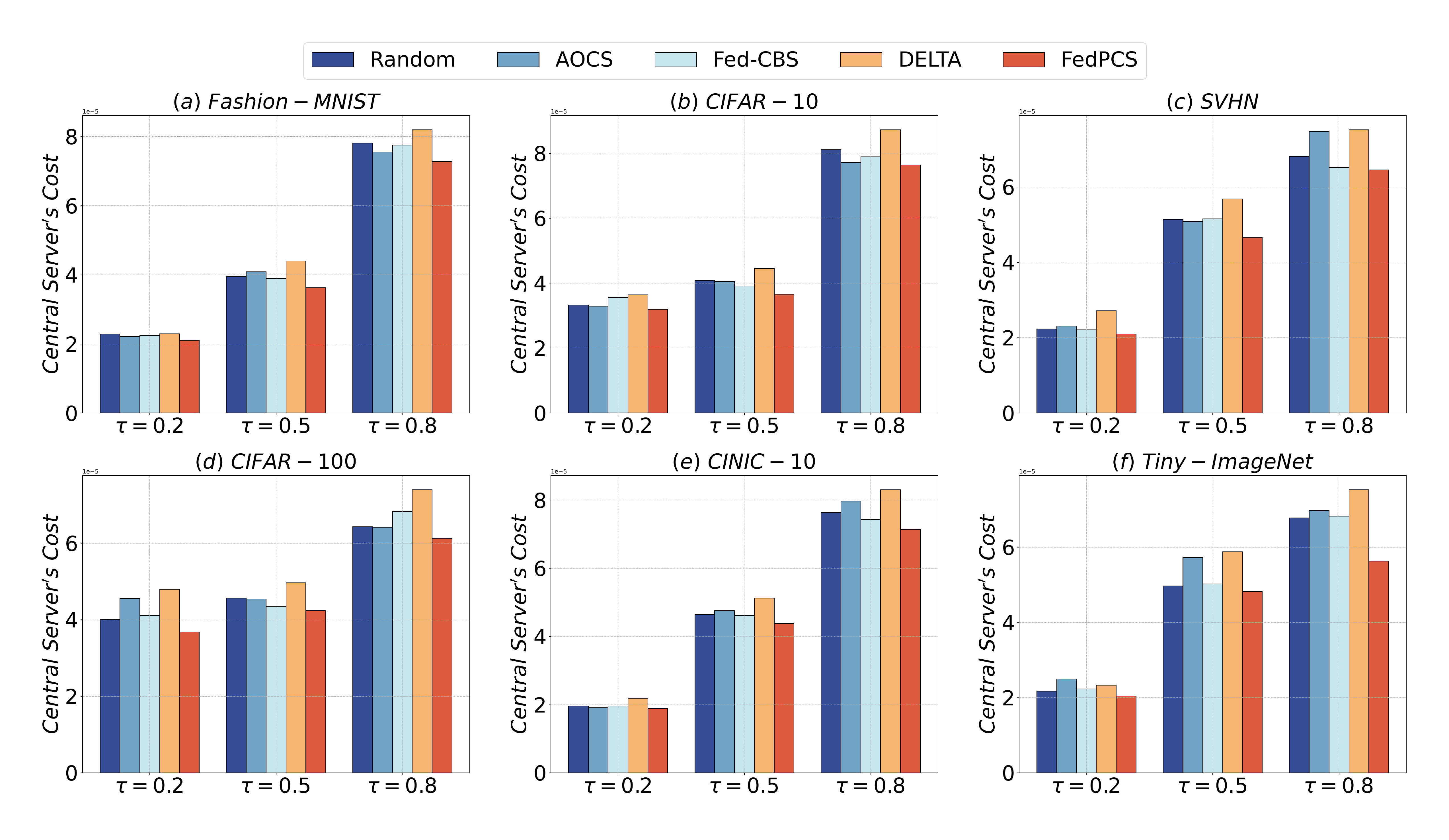}}
            \caption{The central server's cost comparison conducted on Fashion-MNIST/CIFAR-10/SVHN/CIFAR-100/CINIC-10/Tiny-ImageNet datasets with $\tau = \{0.2, 0.5, 0.8\}$.}
            \label{central_server_cost}
        \end{minipage}
    \end{minipage}
    \hspace{2pt}
    \begin{minipage}{150pt}
        \begin{minipage}{150pt}
            \centerline{\includegraphics[width=1.0\textwidth, trim=0 0 0 0,clip]{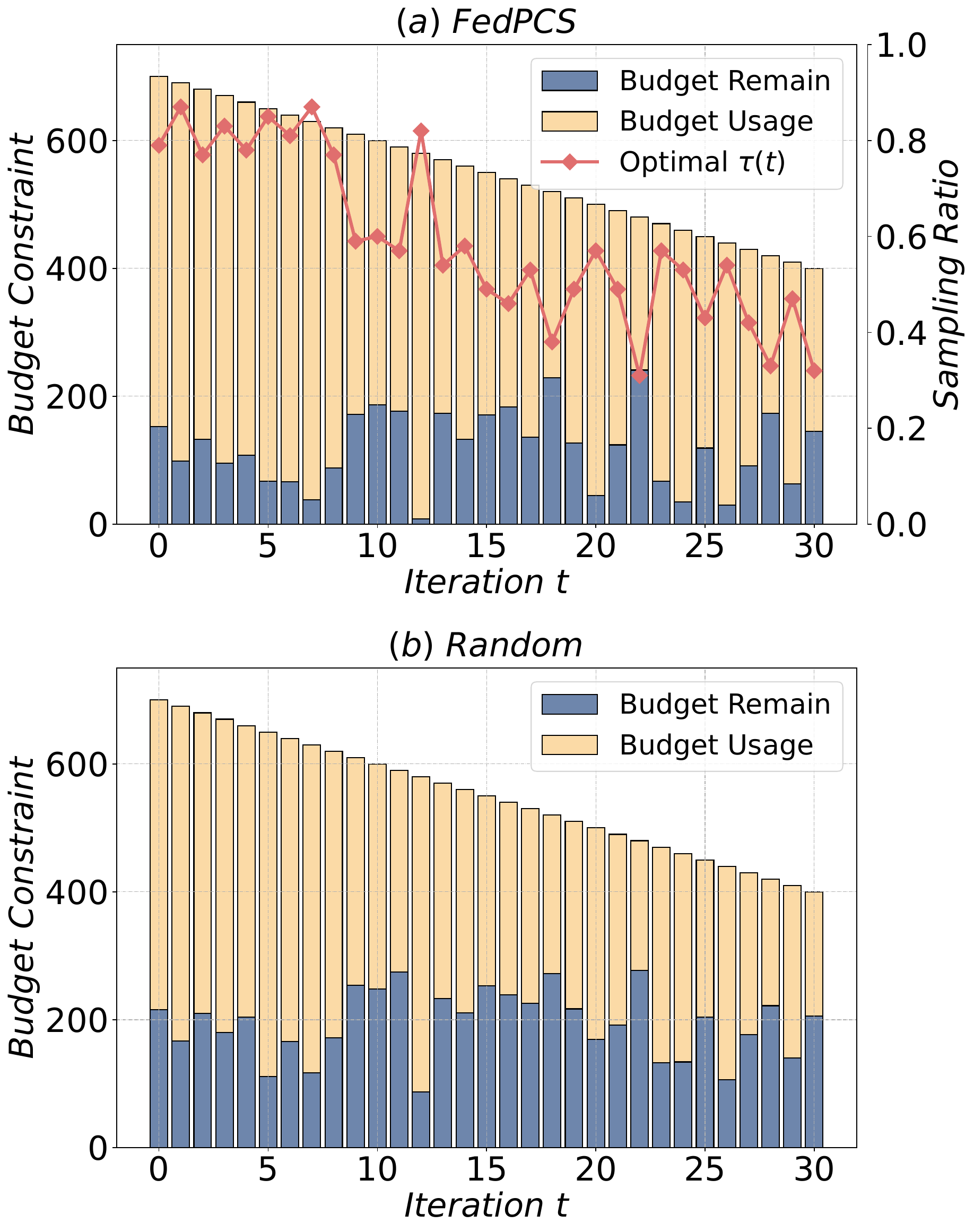}}
            \caption{Privacy budget usage comparison on Fashion-MNIST.}
            \label{opt_K_fmnist}
        \end{minipage}
        \begin{minipage}{150pt}
            \centerline{\includegraphics[width=1.0\textwidth, trim=0 0 0 0,clip]{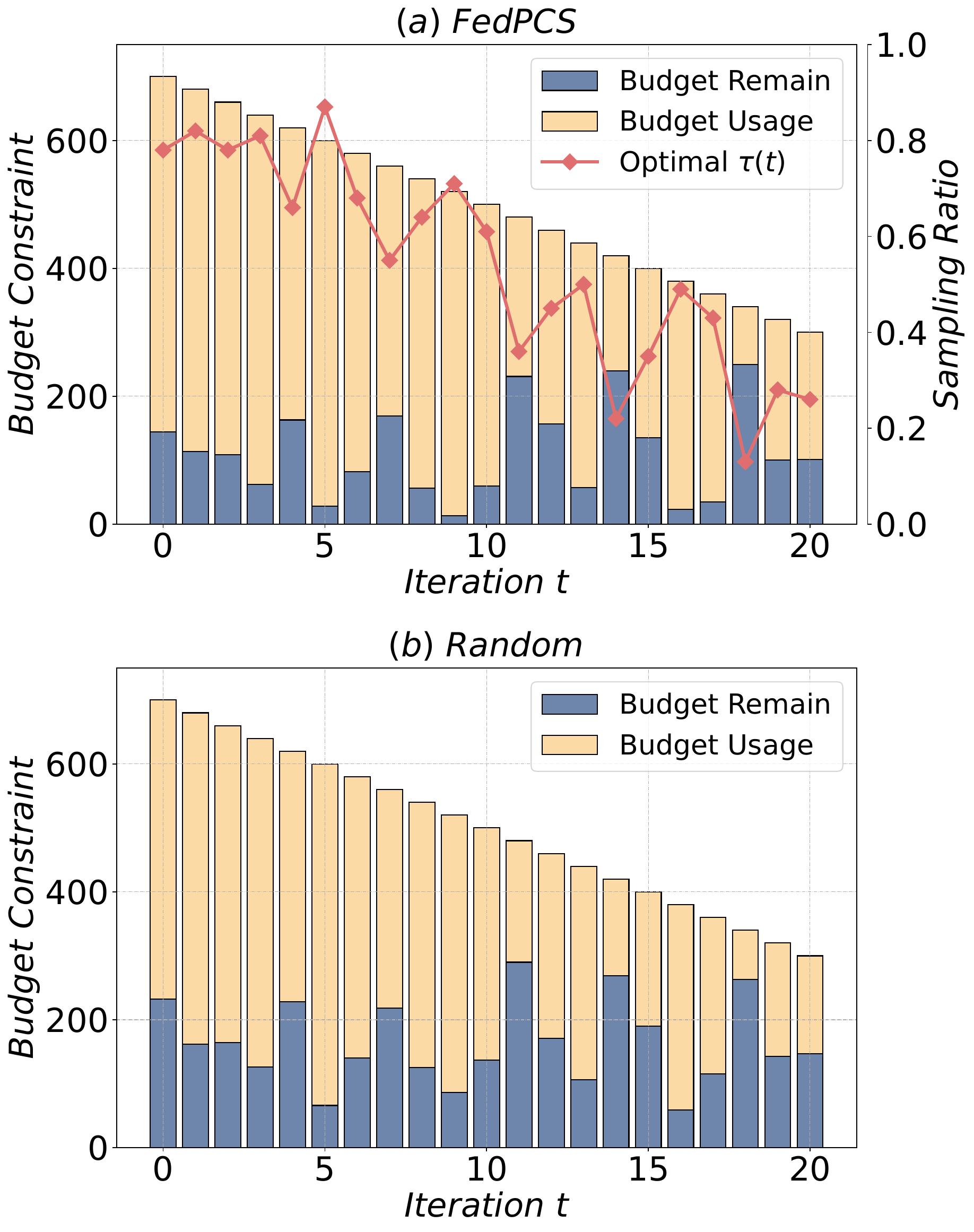}}
            \caption{Privacy budget usage comparison on Tiny-ImageNet.}
            \label{opt_K_tinyimagenet}
        \end{minipage}
    
    \end{minipage}
\vspace{-10pt}
\end{figure*}

\section{Experiments} \label{Experiments}
In this section, we perform comprehensive numerical experiments to validate our analytical results and assess the effectiveness of our proposed FedPCS framework. We first introduce the experiment setups and then provide the experimental results compared with the state-of-the-art benchmarks on different real-world datasets. 

\begin{table}[t]
\setlength{\abovecaptionskip}{0pt} 
\renewcommand\arraystretch{1.0}
\caption{Basic Information of Datasets}
\begin{center}
\resizebox{0.48\textwidth}{!}{\begin{tabular}{c|c|c|c|c}
\toprule[1pt]
\textbf{Datasets}&\textbf{Traing Set Size}&\textbf{Test Set Size}&\textbf{Class} &\textbf{Image Size}  \\ \cmidrule[0.5pt](l{1pt}r{0pt}){1-5}

\textbf{Fashion-MNIST} & 60,000 & 10,000  & 10 & 1 $\times$ 28 $\times$ 28  \\ \cmidrule[0.5pt](l{1pt}r{0pt}){1-5}

\textbf{CIFAR-10} & 50,000 & 10,000  & 10 & 3 $\times$ 32 $\times$ 32  \\ \cmidrule[0.5pt](l{1pt}r{0pt}){1-5}

\textbf{SVHN} & 73,257 & 26,032  & 10 & 3 $\times$ 32 $\times$ 32  \\ \cmidrule[0.5pt](l{1pt}r{0pt}){1-5}

\textbf{CIFAR-100} & 50,000 & 10,000  & 100 & 3 $\times$ 32 $\times$ 32  \\ \cmidrule[0.5pt](l{1pt}r{0pt}){1-5}

\textbf{CINIC-10} & 90,000 & 90,000  & 10 & 3 $\times$ 32 $\times$ 32  \\ \cmidrule[0.5pt](l{1pt}r{0pt}){1-5}

\textbf{Tiny-ImageNet} & 100,000 & 10,000  & 200 & 3 $\times$ 64 $\times$ 64  \\ 
\bottomrule[1pt]
\end{tabular}}
\label{basic_information_of_datasets}
\end{center}
\vspace{-6pt}
\end{table}

\subsection{Experimental Setups}
\subsubsection{Experimental Environment} 
We conduct our experiment on a workstation equipped with Ubuntu 22.04.4; CPU: Intel(R) Xeon(R) Gold 6133 @ 2.50GHz; RAM: 32GB DDR4 2666 MHz; GPU: NVIDIA GeForce RTX 4090; CUDA version 12.2.

\subsubsection{Datasets and Local Model Architecture}
We conduct our experiments on six different real-world datasets, i.e., Fashion-MNIST \cite{xiao2017fashion}, CIFAR-10\&CIFAR-100 \cite{krizhevsky2009learning}, SVHN \cite{netzer2011reading}, CINIC-10 \cite{darlow2018cinic} and Tiny-ImageNet\footnote{\href{https://www.kaggle.com/c/tiny-imagenet}{https://www.kaggle.com/c/tiny-imagenet}.}, to validate the performance of our proposed FedPCS. The basic information regarding the datasets is summarized in TABLE~\ref{basic_information_of_datasets}. We partition the training data for each participating client into IID and Non-IID distribution based on Dirichlet distribution \cite{hsu2019measuring} and default the Dirichlet parameter $\alpha_{\text{Dir}} = 0.5$ to validate the robustness of our proposed FedPCS. As shown in Fig.~\ref{data_distribution}, we display the heat map of IID and Non-IID datasets on CIFAR-10 respectively. The color of the heat map is shallow under the IID distribution with a close range of label numbers, while for the Non-IID distribution, the heat map is much deeper with a wider range of label quantity, indicating a higher level of data heterogeneity. 

We utilize different model architectures for local training. For Fashion-MNIST and CIFAR-10, we implement the convolutional neural network (CNN) models similar to the architectures in \cite{mcmahan2017communication}; for CINIC-10, we adopt the CNN same as \cite{ye2023feddisco}; for SVHN, CIFAR-100, and Tiny-ImageNet datasets, we utilize the VGG-11, VGG-16 and ResNet-18 for local training. Besides, similar to \cite{feng2023towards}, we use ResNet-18 pre-trained on ImageNet-1K as the backbone of the experiments on Tiny-ImageNet.


\subsubsection{Hyperparameter settings}
We apply the stochastic gradient descent (SGD) optimizer for local model training and summarize the learning rate, batchsize, global iteration, and local training epoch specific to each dataset in TABLE~\ref{hyperparameter}. Regarding the parameters for
the incentive-based FL system, we default the total participating client number $N = 100$, the sampling rate $\tau = \{0.2, 0.5, 0.8\}$, the boundary condition for the privacy budget $\rho_{i}^{t} \in [0.01, 12]$, aggregation parameter $\gamma \in (0,1)$ for the central server's cost function in Eq.~(\ref{central_server_utility}) and the clipping threshold $\epsilon_{0} \!=\! 10^{-3}$ for obtaining the fixed point of $\phi(t)$ in Algorithm \ref{alg2}. Additionally, to enable a heterogeneous client set for model training, we randomly generate the weight parameter $\varphi_{i}$, which follows the uniform distribution $\mathbb{U}(0,1)$.

\begin{table}[t]
\setlength{\abovecaptionskip}{0pt} 
\caption{FL Training Hyperparameter Settings}
\renewcommand\arraystretch{1.0}
\begin{center}
\resizebox{0.48\textwidth}{!}{\begin{tabular}{c|c|c|c|c}
\toprule[1pt]
\textbf{Datasets} & \textbf{Learning Rate}&\textbf{Batchsize} & \textbf{Global Iteration} & \textbf{Local Epoch}  \\ \cmidrule[0.5pt](l{1pt}r{0pt}){1-5}

\textbf{Fashion-MNIST} & 0.1 & 32  & 30 & 5 \\ \cmidrule[0.5pt](l{1pt}r{0pt}){1-5}

\textbf{CIFAR-10} & 0.1 & 32  & 80 & 5 \\ \cmidrule[0.5pt](l{1pt}r{0pt}){1-5}

\textbf{SVHN} & 0.1 & 256  & 120 & 5 \\ \cmidrule[0.5pt](l{1pt}r{0pt}){1-5}

\textbf{CIFAR-100} & 0.01 & 64  & 100 & 5 \\ \cmidrule[0.5pt](l{1pt}r{0pt}){1-5}

\textbf{CINIC-10} & 0.01 & 64  & 60 & 5 \\ \cmidrule[0.5pt](l{1pt}r{0pt}){1-5}

\textbf{Tiny-ImageNet} & 0.01 & 64  & 20 & 10 \\ 
\bottomrule[1pt]
\end{tabular}}
\label{hyperparameter}
\end{center}
\vspace{-6pt}
\end{table}

\begin{figure*}[t]
\setlength{\abovecaptionskip}{3pt} 
\centerline{\includegraphics[width=0.98\textwidth, trim=0 10 0 5,clip]{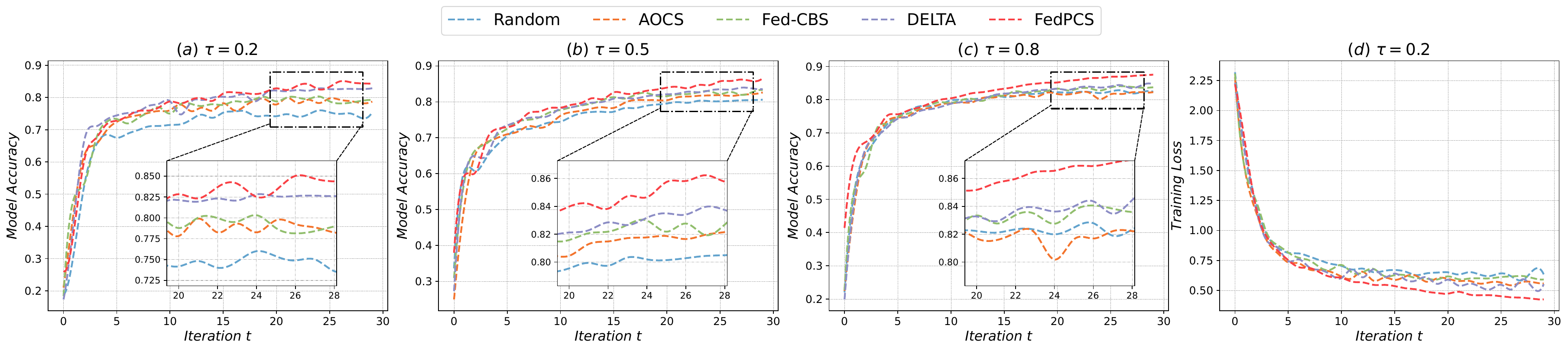}}
\caption{Accuracy and loss curve on Fashion-MNIST dataset with different sampling rate $\tau$.}
\label{FMNIST}
\vspace{-13pt}
\end{figure*}

\begin{figure*}[t]
\setlength{\abovecaptionskip}{3pt} 
\centerline{\includegraphics[width=0.98\textwidth, trim=0 10 0 5,clip]{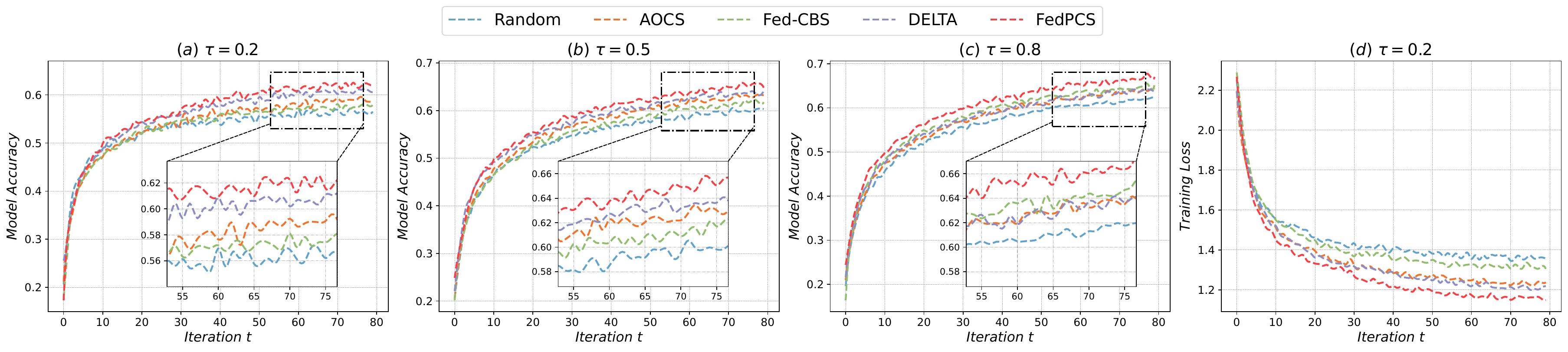}}
\caption{Accuracy and loss curve on CIFAR-10 dataset with different sampling rate $\tau$.}
\label{CIFAR10}
\vspace{-13pt}
\end{figure*}

\begin{figure*}[t]
\setlength{\abovecaptionskip}{3pt} 
\centerline{\includegraphics[width=0.98\textwidth, trim=0 10 0 5,clip]{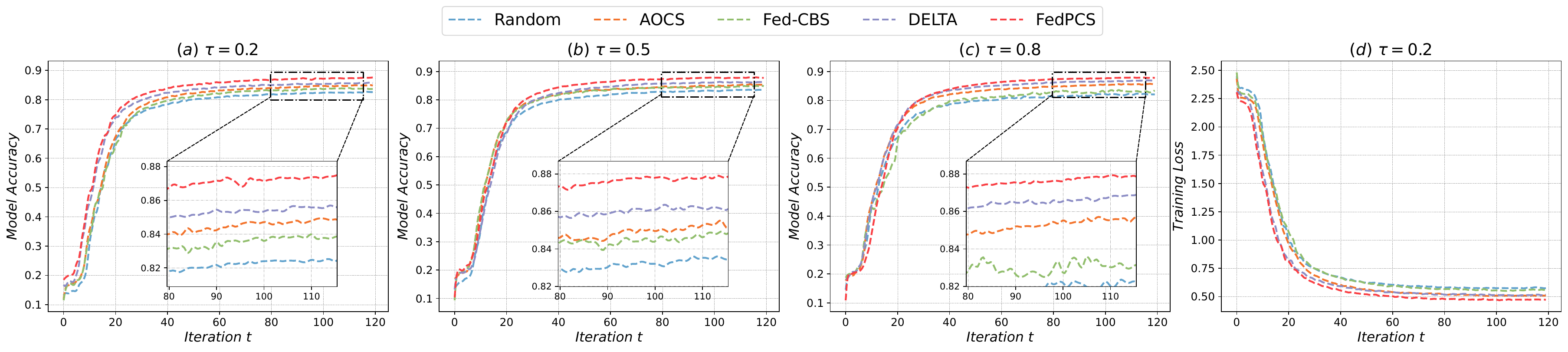}}
\caption{Accuracy and loss curve on SVHN dataset with different sampling rate $\tau$.}
\label{SVHN}
\vspace{-13pt}
\end{figure*}

\begin{figure*}[t]
\setlength{\abovecaptionskip}{2pt} 
\centerline{\includegraphics[width=0.98\textwidth, trim=0 10 0 5,clip]{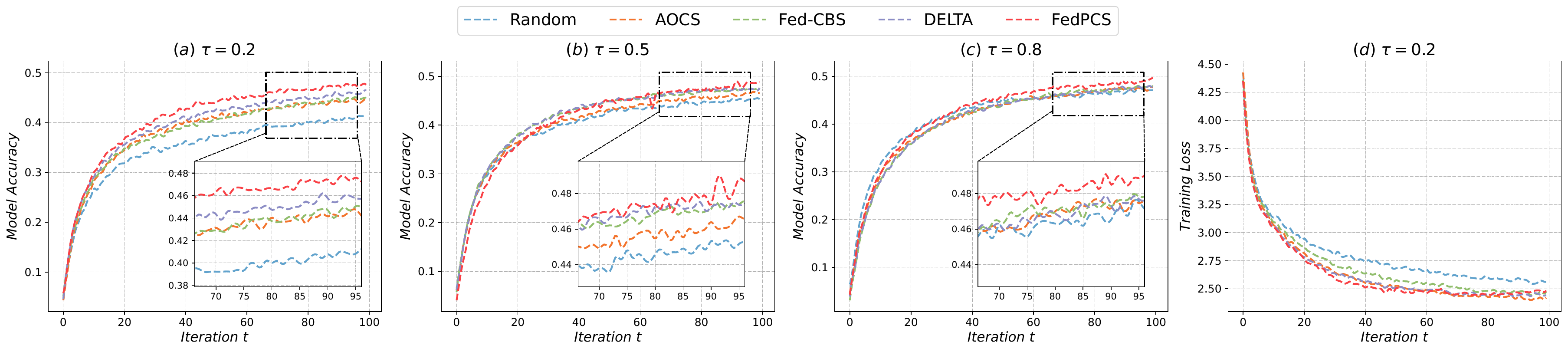}}
\caption{Accuracy and loss curve on CIFAR-100 dataset with different sampling rate $\tau$.}
\label{CIFAR100}
\vspace{-16pt}
\end{figure*}

\subsubsection{Baselines}
We consider several baseline client sampling algorithms in FL systems. The Random client sampling strategy, similar to FedAvg \cite{mcmahan2017communication}, selects clients for global model training with a fixed sampling probability $x_{i}^{t} \!=\! \frac{1}{N}$. The AOCS algorithm \cite{chen2022optimal} reduces the communication burden in FL systems and defines the sampling probability according to the norm of the local update by $x_{i}^{t} \!=\! \frac{\|(m + l - n)\nabla F_{i}(\boldsymbol{w(t)})\|}{\sum_{j=1}^{l} \|\nabla F_{j}(\boldsymbol{w(t)})\|}$. The Fed-CBS mechanism \cite{zhang2023fed} serves as a heterogeneity-aware sampling scheme that utilizes a privacy-preserving quadratic class-imbalance degree (QCID) metric for the client sampling probability design, i.e., the sampling probability $x_{i}^{t} \propto\frac{[QCID(\mathcal{M}_{i-1})]^{\beta_{i-1}}}{[QCID(\mathcal{M}_{i})]^{\beta_{i}}}$, where the parameter $\beta_{i-1} < \beta_{i}$. The DELTA algorithm \cite{wang2024delta} characterizes the influence of the client divergence and local variance, and samples the representative clients with valuable information for global model aggregation. The client sampling probability is defined by $x_{i}^{t} = \frac{\sqrt{\alpha_{1} \zeta_{G,i,t}^{2} + \alpha_{2} \sigma_{L,i}^{2}}}{\sum_{j=1}^{m} \sqrt{\alpha_{1}\zeta_{G,j,t}^{2} + \alpha_{2} \sigma_{L,j}^{2}}}$, where $\zeta_{G,i,t}$ and $\sigma_{L,i}$ represents the gradient diversity and local variance respectively.

\subsection{Utility Evaluation Analysis}
\vspace{-3pt}
It is displayed in Fig. \ref{phi} that $\phi(t)$ typically converges to the fixed point mostly within 5 iterations, which validates the linear computation complexity and fast convergence of Algorithm~\ref{alg1} as elaborated in Section \ref{Stackelberg_Nash_Equilibrium_Analysis}-B. Moreover, we introduce the Pearson correlation coefficient to analyze the optimal strategies of participating clients and present a heat map of the privacy budget $\rho_{i}^{t}$ for each client, as shown in Fig. \ref{correlation}. Generally, the warmer color indicates higher parameter correlations and smaller divergence among the strategies of egocentric clients, thereby increasing the global utilities and enhancing the FL model performance. As illustrated in Fig. \ref{correlation}\subref{Correlation_random}-\subref{Correlation_delta}, the correlation of each client's strategy within baseline methods is primarily highest with its own parameter while relatively lower with the parameters from other clients. This suggests that egocentric clients prioritize personal utility, which may degrade the performance of the FL systems. For our proposed framework FedPCS, as depicted in Fig. \ref{correlation}\subref{Correlation_pcs}, the high correlation among parameters indicates that egocentric clients achieve the trade-off between the individual profit and global model performance, which demonstrates the effectiveness and superiority of our proposed FedPCS. To validate the derivations in Proposition~\ref{proposition_random_sampling_poa} and Theorem~\ref{theorem_privacy_poa}, we conduct a comparative analysis of social welfare across different datasets during the training process, as illustrated in Fig.~\ref{social_welfare_comparison}. The curves demonstrate that the privacy-aware client sampling strategy in FedPCS closely approximates the socially optimal strategy throughout the training process, while a significant gap is observed in comparison to the random sampling policy, reflecting the superiority of our proposed FedPCS framework.

\begin{figure*}[t]
\setlength{\abovecaptionskip}{3pt} 
\centerline{\includegraphics[width=0.98\textwidth, trim=0 10 0 5,clip]{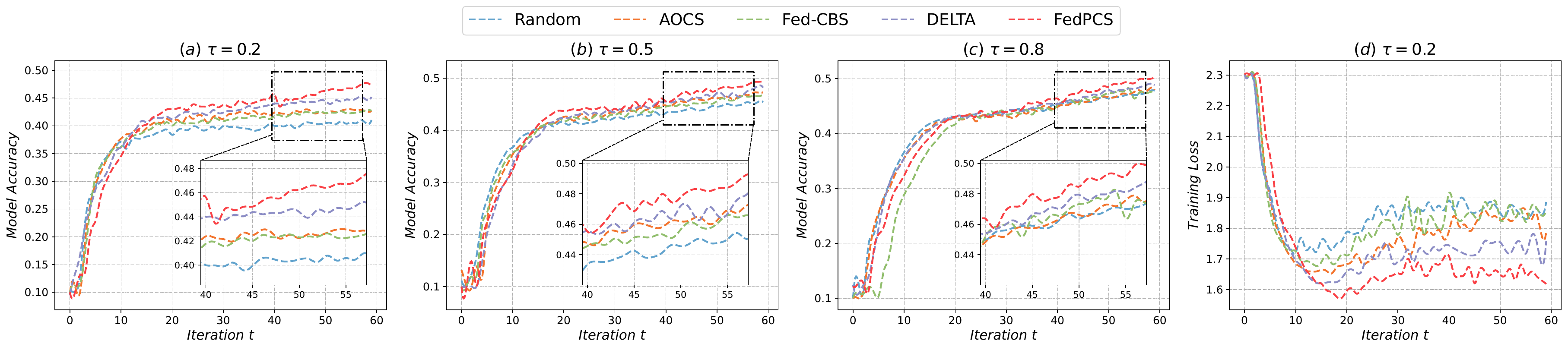}}
\caption{Accuracy and loss curve on CINIC-10 dataset with different sampling rate $\tau$.}
\label{CINIC10}
\vspace{-13pt}
\end{figure*}

\begin{figure*}[t]
\setlength{\abovecaptionskip}{3pt} 
\centerline{\includegraphics[width=0.98\textwidth, trim=0 10 0 5,clip]{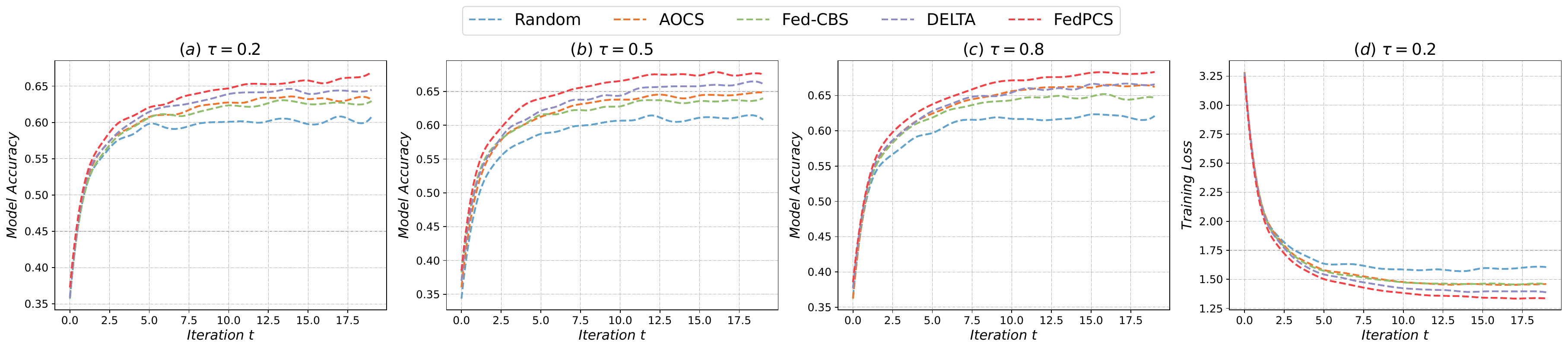}}
\caption{Accuracy and loss curve on Tiny-ImageNet dataset with different sampling rate $\tau$.}
\label{Tinyimagenet}
\vspace{-13pt}
\end{figure*}

\begin{figure*}[t]
\setlength{\abovecaptionskip}{-2pt} 
\centerline{\includegraphics[width=0.97\textwidth, trim=60 50 60 50,clip]{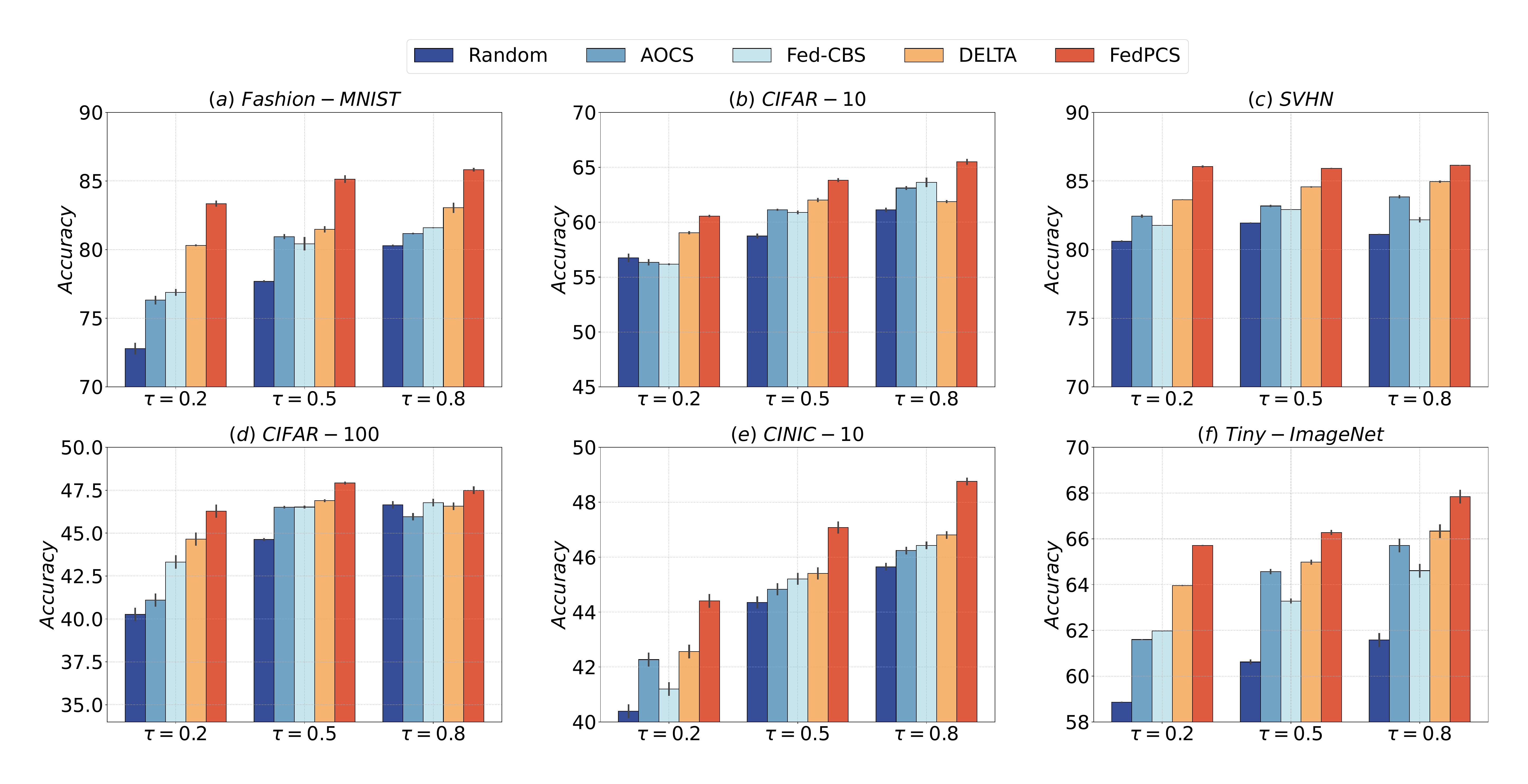}}
\caption{Accuracy comparison on all datasets with different sampling rates under Non-IID settings.}
\label{Non_IID}
\vspace{-16pt}
\end{figure*}

To demonstrate the effectiveness of our proposed FL framework FedPCS in cost minimization for the central server, we compare the cost in Eq.~(\ref{central_server_utility}) on Fashion-MNIST/CIFAR-10/SVHN/CIFAR-100/CINIC-10/Tiny-ImageNet datasets and benchmark FedPCS against several baseline methods under various sampling rate settings ($\tau \!=\! \{0.2, 0.5, 0.8\}$). As illustrated in Fig. \ref{central_server_cost}, based on the optimal reward derived in Lemma~\ref{lemma_optimal_reward}, FedPCS consistently yields the lowest cost values for the central server, underscoring the efficiency of FedPCS in minimizing the central server's cost. Furthermore, as the number of sampled clients increases, the disparity between FedPCS and baseline methods becomes more pronounced, indicating that our method exhibits superior performance regardless of client scale. These advantages become increasingly prominent with larger client populations, highlighting the robustness and effectiveness of FedPCS under large-scale FL systems. Furthermore, we explore the scenarios with the time-varying privacy budget constraints $B_{t} \!\in\! \{300, 700\}$ and $B_{t} \!\in\! \{400, 700\}$ with uniform decrements for Fashion-MNIST and Tiny-ImageNet datasets respectively, as illustrated in Figs.~\ref{opt_K_fmnist}-\ref{opt_K_tinyimagenet}. The results demonstrate that, with the decrement of the privacy budget constraint, the central server gradually reduces the optimal client sampling ratio $\tau(t)$ as derived in Theorem~\ref{theorem_opt_reward_and_tau}. We also analyze the privacy usage ratio of FedPCS, as shown in Fig.~\ref{opt_K_fmnist}(a) and Fig.~\ref{opt_K_tinyimagenet}(a), which consistently achieves a higher privacy usage ratio compared to the random strategy depicted in Fig.~\ref{opt_K_fmnist}(b) and Fig.~\ref{opt_K_tinyimagenet}(b), indicating the privacy efficiency and robustness of our proposed FedPCS framework under dynamic scenarios.

\subsection{Model Training Performance Assessment}
Based on the optimal strategy profile derived in Theorem \ref{theorem_optimal_alpha} and Lemma \ref{lemma_optimal_reward}, we evaluate the performance of our proposed FedPCS framework against the baseline by considering the model accuracy and training loss. Table~\ref{accuracy} shows the FL model accuracy for different real-world datasets (Fashion-MNIST, CIFAR-10, SVHN, CIFAR-100, CINIC-10, and Tiny-ImageNet) under varying sampling rates ($\tau \!=\! \{0.2, 0.5, 0.8\}$) and IID and Non-IID settings. Our experiment results demonstrate that FedPCS consistently outperforms the baseline methods in all client scales. Additionally, we display the prediction accuracy and loss curves observed during the model training process, as illustrated in Figs. \ref{FMNIST}-\ref{Tinyimagenet}. The curves indicate that with the increment of client sampling rates, the accuracy enhancement becomes more pronounced, and FedPCS consistently surpasses the benchmarks. The training loss curves further corroborate that FedPCS achieves a faster reduction in training loss and consistently attains lower loss values compared to the baseline approaches. The results under Non-IID settings, as illustrated in Fig.~\ref{Non_IID}, further confirm the demonstrated superiority.

\section{Conclusion} \label{Conclusion}
In this paper, we propose a game-theoretic framework with a privacy-aware client sampling strategy in federated learning to achieve clients' privacy preservation and enhancement of FL model performance. First, we derive a pioneering upper bound for the accuracy loss of FL models incorporating privacy-aware client sampling probabilities, based on which, a two-stage Stackelberg game is constructed to model the interaction between the central server and clients. We analytically demonstrate that our optimal strategy profile for the central server and clients achieves the Stackelberg Nash Equilibrium. Moreover, to address the challenge of approximating other clients' local information, a mean-field estimator is introduced for the client's optimal strategy design. We conduct rigorous theoretical analyses regarding the existence and convergence of the fixed point for the mean-field estimator and establish convergence upper bound for FedPCS. Through PoA analysis, we theoretically demonstrate the limitation of the random sampling strategy and the efficacy of FedPCS in comparison to the socially optimal policy. We further extend our model and analysis to dynamic privacy constraint scenarios and obtain the closed-form adaptive optimal client sampling ratio and reward at each global iteration. Extensive experimental results validate the superiority and robustness of our proposed framework FedPCS compared to state-of-the-art baselines under IID and Non-IID datasets. 



\bibliographystyle{IEEEtran}
\bibliography{reference}




\newpage

\appendices

\section{Proof of Proposition \ref{proposition_variance}} \label{proof_proposition_variance}
\begin{proof}
For any participating client $i$ with two adjacent datasets $\mathcal{D}_{i}$ and $\mathcal{D}_{i}^{'}$ and hypothesise that there exists a clipping threshold $W$ for the $i$-th client's local model parameters at $t$-th global iteration with adding artificial perturbation ($\|\boldsymbol{w_{i}(t)}\| \!\leq\!  W$), the sensitivity of the query function $Q$ regarding the input $\mathcal{D}_{i}$ and $\mathcal{D}_{i}^{'}$ can be obtained as follows:
\begin{align} \label{sensitivity}
\Delta{Q} =& \max_{\mathcal{D}_{i},\mathcal{D}_{i}^{'}}\|Q(\mathcal{D}_{i}) - Q(\mathcal{D}_{i}^{'}) \|_{2} \nonumber \\
=& \max_{\mathcal{D}_{i},\mathcal{D}_{i}^{'}} \frac{1}{|\mathcal{D}_{i}|} \| \text{arg}\min_{\boldsymbol{w}} \sum_{j \in \mathcal{D}_{i}}  f_{i}(\boldsymbol{w},\mathcal{D}_{i}) - \sum_{j \in \mathcal{D}_{i}^{'}} \! f_{i}(\boldsymbol{w},\mathcal{D}_{i}^{\prime})\|_{2} \nonumber \\
=& \frac{2 W}{|\mathcal{D}_{i}|}. 
\end{align}

Given the constraint of the privacy budget $\rho_{i}^{t}$ aforementioned and combing the expression of the sensitivity of the query function $Q$ in Eq.~(\ref{sensitivity}), the proposition holds.    \end{proof}

\section{Proof of Proposition \ref{proposition_central_server_accuracy_loss}} \label{proof_proposition_central_server_accuracy_loss}
\begin{proof}
Since the global loss function $F(\boldsymbol{w(t)})$ satisfies the $\beta$-Lipschitz smoothness assumption and Polyak-\L{}ojasiewicz condition, assume the expected squared $\ell_{2}$-norm of the perturbed global model satisfies $\mathbb{E}[\|\nabla \! \widetilde{F}(\boldsymbol{w(t)})\|_{2}^{2}] \!\leq\! G^{2}(t)$ \cite{rakhlin2011making}, the upper bound of the accuracy loss of the global model $\mathbb{E}[F(\boldsymbol{w(t)}) \!-\! F(\boldsymbol{w^{*}})] \!\leq\! \frac{\beta G^{2}(t)}{2 \mu^{2} t}$ holds according to \cite{rakhlin2011making}. From the perturbed local gradient in Eq.~(\ref{local_gradient_perturbation}), the global parameter with artificial Gaussian random noise can be derived as $\nabla \widetilde{F}(\boldsymbol{w(t)}) \!=\! \sum_{i=1}^{N} \theta_{i} (\nabla F_{i}(\boldsymbol{w(t)}) \!+\! \boldsymbol{n_{i}(t)}) \!=\! \nabla F(\boldsymbol{w(t)}) \!+\! \sum_{i=1}^{N} \theta_{i} \boldsymbol{n_{i}(t)}$. Further, note that the added Gaussian random noise possesses zero mean and $\mathbb{E}[\|\! \sum_{i=1}^{N} \! \theta_{i} \boldsymbol{n_{i}(t)}\|_{2}^{2}] \!=\! d  \sum_{i=1}^{N} \theta_{i}^{2} \sigma_{i}^{2}(t)$, where the variance $\sigma_{i}^{2}(t)$ is affected by $i$-th client's privacy budget $\rho_{i}^{t}$ based on Proposition~\ref{proposition_variance}. Consequently, the upper boundary of $\mathbb{E}[\|\nabla \widetilde{F}(\boldsymbol{w(t)})\|_{2}^{2}]$ can be derived as follows:
\begin{align}\label{global_gradient_expectation}
\mathbb{E}[\|\nabla \widetilde{F}(\boldsymbol{w(t)})\|_{2}^{2}] \! &= \! \mathbb{E}[\|\nabla F(\boldsymbol{w(t)}) \!+\! \sum\nolimits_{i=1}^{N} \!\theta_{i} \boldsymbol{n_{i}(t)} \|_{2}^{2}] \nonumber \\ 
&= \mathbb{E}[\|\nabla F(\boldsymbol{w(t)})\|_{2}^{2}] \!+\! \mathbb{E}[\|\sum\nolimits_{i=1}^{N} \! \theta_{i} \boldsymbol{n_{i}(t)}\|_{2}^{2}] \nonumber \\ 
&\leq V^{2} + d \sum\nolimits_{i=1}^{N} \theta_{i}^{2} \sigma_{i}^{2}(t) \triangleq G^{2}(t). 
\end{align}

Thus, based on Eq.~(\ref{global_gradient_expectation}), we can derive the upper bound of the accuracy loss of the global model as expressed in Eq.~(\ref{accuracy_loss}).~\end{proof} 

\section{Proof of Theorem \ref{theorem_optimal_alpha}} \label{proof_theorem_optimal_alpha}
\begin{proof}
Following the discrete-time Maximum Principle and mean-field term $\phi(t)$, the Hamilton function is constructed as: 
\begin{align} \label{hamilton}
H(t) &= (1 \!-\! (1 \!-\! \frac{\rho_{i}^{t}}{N \phi(t)})^{K})(\rho_{i}^{t}R_{t} \!-\! \varphi_{i} (\rho_{i}^{t})^{2} \!-\! (1 \!-\!\varphi_{i}) (\alpha_{i}^{t})^{2}) \nonumber \\
&+ \lambda(t + 1)(1 - \alpha_{i}^{t})(\phi(t) - \rho_{i}^{t}), 
\end{align} 
where $\phi(t)$ can be viewed as a given function, which is not affected by the correction factor $\alpha_{i}^{t}$. Based on the properties of discrete-time Hamilton function in \cite{di2012discrete}, to obtain the close form expression of $\alpha_{i}^{t}$, we consider the first-order and second-order constraints on the control vector $\alpha_{i}^{t}$ for each client $i \!\in\! \{1,2,\ldots,N\}$ as follows:
\begin{align}
\frac{\partial H(t)}{\partial \alpha_{i}^{t}} =& - 2 (1 - \varphi_{i})(1 - (1 - \frac{\rho_{i}^{t}}{N \phi(t)})^{K}) \alpha_{i}^{t} \nonumber \\
&+ \lambda(t+1)(\rho_{i}^{t} - \phi(t)) \label{hamilton_first_order} \\
\frac{\partial^{2} H(t)}{\partial(\alpha_{i}^{t})^{2}} =& - 2 (1 - \varphi_{i})(1 - (1 - \frac{\rho_{i}^{t}}{N \phi(t)})^{K}) \alpha_{i}^{t} < 0. \label{hamilton_second_order}
\end{align}

From Eq.~(\ref{hamilton_second_order}), the second-order derivative $\partial^{2} H(t)/\partial(\alpha_{i}^{t})^{2} \!<\! 0$, which indicates that the Hamilton function is strictly concave in the feasible region of $p_{i}(t)$ for all time horizons $t \!\in\! \{0,1, \ldots, \\ T\} $. The maximum of the Hamilton function holds by calculating the first derivative of Eq.~(\ref{hamilton}) regarding $\alpha_{i}^{t}$ equals to 0, i.e., $\partial H(t)/\partial(\alpha_{i}^{t}) \!=\! 0$. Consequently, we obtain the expression of the correction factor $\alpha_{i}^{t}$ as follows:
\begin{align} \label{alpha}
\alpha_{i}^{t} = \frac{\lambda(t+1)(\rho_{i}^{t} - \phi(t))}{2 (1 - \varphi_{i})(1 - (1 - \frac{\rho_{i}^{t}}{N \phi(t)})^{K})}.
\end{align}

Based on the differential property concerning the impulse vector $\lambda$ and the state vector $\rho$ in Hamilton function \cite{qi2021linear}, we have the following derivation:
\begin{align} \label{lambda_difference}
& \text{\small $ \lambda(t + 1) \!-\! \lambda(t) = - \frac{\partial H(t)}{\partial \rho_{i}^{t}} $} \nonumber \\
&= \text{\small $ \lambda(t+1)(1 \!-\! \alpha_{i}^{t}) - (1 \!-\! (1 \!-\! \frac{\rho_{i}^{t}}{N \phi(t)})^{K}) (R_{t} \!-\! 2 \varphi_{i} \rho_{i}^{t}) $} \nonumber \\
&- \text{\small $ \frac{K}{N \phi(t)} (1 \!-\! \frac{\rho_{i}^{t}}{N \phi(t)})^{K - 1} (\rho_{i}^{t}R_{t} \!-\! \varphi_{i} (\rho_{i}^{t})^{2} \!-\! (1 \!-\!\varphi_{i}) (\alpha_{i}^{t})^{2}). $}
\end{align}

From Eq.~(\ref{lambda_difference}), we can simplify the expression of $\lambda(t)$ as:
\begin{align} \label{lamda_iteration}
\lambda(t) = \alpha_{i}^{t} \lambda(t + 1) + S(t), \  S(t) = Q_{i}^{t} R_{t} + M_{i}^{t},
\end{align}
where $Q_{i}^{t} \!=\! 1 \!-\! (1 \!-\! \frac{\rho_{i}^{t}}{N \phi(t)})^{K} \!+\! \frac{K \rho_{i}^{t}}{N \phi(t)} (1 \!-\! \frac{\rho_{i}^{t}}{N \phi(t)})^{K - 1}$ and $M_{i}^{t} \!=\!  - 2 \varphi_{i} \rho_{i}^{t}(1 \!-\! (1 \!-\! \frac{\rho_{i}^{t}}{N \phi(t)})^{K}) - (1 \!-\! \frac{\rho_{i}^{t}}{N \phi(t)})^{K - 1} (\varphi_{i} (\rho_{i}^{t})^{2} \!+\! (1 \!-\!\varphi_{i}) (\alpha_{i}^{t})^{2})$. Following the boundary condition \cite{di2012discrete} for the final value at global iteration $T$ of the impulse vector $\lambda$, we have:
\begin{align}
&\text{\small $ \lambda(T) \!=\! \frac{\partial S(\rho_{i}^{T})}{\partial (\rho_{i}^{T})} \!=\! \frac{K}{N \phi(T)} (1 \!-\! \frac{\rho_{i}^{T}}{N \phi(T)})^{K \!- 1} (\rho_{i}^{T}R_{T} \!-\! \varphi_{i} (\rho_{i}^{T})^{2})  $} \nonumber \\
&+ \text{\small $ \!(1 \!-\! (1 \!-\! \frac{\rho_{i}^{T}}{N \phi(T)})^{K} \!) (R_{T} \!-\! 2 \varphi_{i} \rho_{i}^{T}) \!=\! Q_{i}^{T} \! R_{T} \!+\! M_{i}^{T} \!\!=\! S(T), \! $}
\end{align}
where $S(\cdot)$ is a weighting function of the state variable at the global iteration $T$ with the expression by:
\begin{align}
S(\rho_{i}^{T}) = (1 \!-\! (1 \!-\! \frac{\rho_{i}^{T}}{N \phi(T)})^{K})(\rho_{i}^{T} R_{T} \!-\! \varphi_{i} (\rho_{i}^{T})^{2}).
\end{align}

By iteratively aggregating the formulations in Eq.~(\ref{lamda_iteration}), we obtain the expression of $\lambda(t)$, $0 \!\in\! \{0,1, \ldots, T-1\}$ as:
\begin{align} \label{lambda_t}
\lambda(t) = \sum_{j=t+1}^{T} \left[ S(j) \prod_{r=t}^{j-1} \alpha_{i}^{k}) \right] + S(t)
\end{align}

Then, according to the formula of $\alpha_{i}^{t}$ in Eq.~(\ref{alpha}) and $\lambda(t)$ in Eq.~(\ref{lambda_t}), we can finalize the optimal correction factor $\alpha_{i}^{t*}$ as shown in Eq.~(\ref{optimal_alpha}). Hence, the theorem holds.  \end{proof}

\section{Proof of Lemma \ref{lemma_optimal_reward}} \label{proof_lemma_optimal_reward}
\begin{proof}
Based on Theorem~\ref{theorem_optimal_alpha}, we observe that the optimal correction factor $\alpha_{i}^{t*}$ can be concluded as an analytic function regarding the given reward $R_{k}$, $k \in \{t+1, t+2, \ldots, T\}$, i.e., 
\begin{align} \label{alpha_with_R}
\alpha_{i}^{t*} \!=\! I_{i}^{t+1} R_{t+1} + J_{i}^{t+1}.
\end{align}
where reward $R_{k}$, $k \!\in\! \{t+2, t+3, \ldots, T\}$ can be considered as known functions, {\small $I_{i}^{t+1} \!=\! \frac{(\rho_{i}^{t} - \phi(t))Q_{i}^{t+1}}{2 (1 - \varphi_{i})(1 - (1 - \frac{\rho_{i}^{t}}{N \phi(t)})^{K})}$ } and {\small $J_{i}^{t+1} \!=\! \frac{(\rho_{i}^{t} - \phi(t))\{\sum_{j=t+2}^{T} [S(j) \prod_{r=t+1}^{j-1} \alpha_{i}^{r}] + M_{i}^{t+1}\}}{2 (1 - \varphi_{i})(1 - (1 - \frac{\rho_{i}^{t}}{N \phi(t)})^{K})}$}. By substituting Eq.~(\ref{alpha_with_R}) into Eq.~(\ref{update_constraint_refined}), we obtain the iteration function regarding the state variable, i.e., the privacy budget $\rho_{i}^{t}$ at $t$-th global training iteration as follows:
\begin{align} \label{rho_iteration}
\rho_{i}^{t} &= (1 \!-\! I_{i}^{t} R_{t} \!-\! J_{i}^{t} ) \phi(t \!-\! 1) \!+\! (I_{i}^{t} R_{t} \!+\! J_{i}^{t}) \rho_{i}^{t-1} \nonumber \\
&= (\rho_{i}^{t-1} \!-\! \phi(t \!-\! 1)) I_{i}^{t}R_{t} \!+\! (1 \!-\! J_{i}^{t} ) \phi(t \!-\! 1) \!+\! J_{i}^{t} \rho_{i}^{t-1} \nonumber \\
&= G_{i}^{t} R_{t} + H_{i}^{t} .
\end{align}
where $G_{i}^{t} = (\rho_{i}^{t-1} - \phi(t - 1)) I_{i}^{t} $ and $H_{i}^{t} = (1 - J_{i}^{t} ) \phi(t - 1) + J_{i}^{t} \rho_{i}^{t-1}$. Then, based on Eq.~(\ref{rho_iteration}), we reformulate Eq.~(\ref{central_server_utility}) as:
\begin{align} \label{central_server_utility_with_alpha}
\footnotesize \!\!\!\!\!\!\! U_{t}(R_{t},\boldsymbol{\alpha^{*}}) \!=\!\! \sum_{i=1}^{K} \! \left[\gamma \frac{\theta_{i}^{2}}{t |\mathcal{D}_{i}|^{2}(G_{i}^{t} R_{t} \!+\! H_{i}^{t})} \!+\! (1 \!-\! \gamma)( G_{i}^{t} R_{t}^{2} \!+\! H_{i}^{t} R_{t}) \right], 
\end{align}
where the optimal correction factor vector $\boldsymbol{\alpha_{i}^{*}} \!=\! \{\alpha_{i}^{t*}, t \!\in\! \{0, 1, \\ \ldots, T\}\}$ and Eq.~(\ref{central_server_utility_with_alpha}) is seen as a high-ordered non-polynomial function regarding $R_{t}$. Consequently, the first-order and the second-order of the central server’s cost function in Eq.~(\ref{central_server_utility_with_alpha}) can be derived as follows:
\begin{align}
&\!\!\small \frac{\partial U_{t}(R_{t}, \! \boldsymbol{\alpha^{*}})}{\partial R_{t}} \!=\!\! \sum_{i=1}^{K} \! \left[(1 \!-\! \gamma) (2 G_{i}^{t} R_{t} \!+\! H_{i}^{t}) \!-\! \frac{\gamma \theta_{i}^{2} G_{i}^{t}}{t |\mathcal{D}_{i}|^{2}(G_{i}^{t} R_{t} \!+\! H_{i}^{t})^{2}}\right], \label{central_server_first_order} \\
&\!\!\small \frac{\partial^{2} U_{t}(R_{t}, \! \boldsymbol{\alpha^{*}})}{\partial R_{t}^{2}} \!=\! \sum_{i=1}^{K}\! \left[ 2 G_{i}^{t}(1 \!-\! \gamma) \!+\! \frac{2 \gamma \theta_{i}^{2} (G_{i}^{t})^{2}}{t |\mathcal{D}_{i}|^{2}(G_{i}^{t} R_{t} \!+\! H_{i}^{t})^{2}}\right] \!>\! 0, \!\!  \label{central_server_second_order}
\end{align}

Then, to demonstrate the existence of the minimum for the utility in Eq.~(\ref{central_server_utility_with_alpha}),  we analyze the root of the first-order derivation as shown in Eq.~(\ref{central_server_first_order}) in the feasible region of $R_{t}$:
\begin{align} \label{limit}
\footnotesize \lim_{R_{t} \rightarrow 0} \frac{\partial U_{t}(R_{t},\boldsymbol{\alpha^{*}})}{ \partial R_{t}} \rightarrow - \infty, \ \lim_{R_{t} \rightarrow +\infty} \frac{\partial U_{t}(R_{t},\boldsymbol{\alpha^{*}})}{ \partial R_{t}} \rightarrow + \infty.
\end{align}

From Eq.~(\ref{limit}), we verify the existence of the solution in the feasible region of $R_{t} \in (0, +\infty)$ for the high-ordered non-polynomial equation in Eq.~(\ref{central_server_utility_with_alpha}). By calculating the equation $\partial U_{t}(R_{t},\boldsymbol{\alpha^{*}})/\partial R_{t} \!=\! 0$, we can obtain the function of optimal reward $R_{t}^{*}$ as shown in Eq.~(\ref{optimal_r}). Thus, the theorem holds.   \end{proof}

\section{Proof of Theorem \ref{theorem_fixed_point}} \label{proof_theorem_fixed_point}
\begin{proof} 
Based on Definition \ref{definition_mean_field} and Theorem \ref{theorem_optimal_alpha}, for each~client $i \!\in\! \{1, 2, \ldots, N\}$, we substitute the mean-field term $\phi(t) \!=\! \frac{1}{N}\sum_{i=1}^{N} \rho_{i}^{t}$ where $i \!\in\! \{1,2 ,\ldots, N\}$, $t \!\in\! \{0,1, \ldots, T\}$ into Eq.~(\ref{update_constraint_refined}). Then, the expression of state variable $\rho_{i}^{t}$ in Eq.~(\ref{update_constraint_refined}) can be further reformulated as follows:
\begin{align} \label{rho_iteration_with_estimators}
\rho_{i}^{t} &= \frac{1}{N} \sum_{k=1}^{N} \rho_{k}^{t-1} + (\rho_{i}^{t-1} - \frac{1}{N} \sum_{k=1}^{N} \rho_{k}^{t-1}) \nonumber \\
&\times \!\frac{(\rho_{i}^{t-1} \!-\! \phi(t \!-\! 1))\{\sum_{j=t+1}^{T} [S(j) \prod_{r=t}^{j-1}  \alpha_{i}^{k}] \!+\! S(t)\}}{2 (1 \!-\! \varphi_{i})(1 \!-\! (1 \!-\! \frac{\rho_{i}^{t-1}}{N \phi(t-1)})^{K})}.
\end{align}

From Eq.~(\ref{rho_iteration_with_estimators}), we notice that the privacy budget $\rho_{i}^{t}$ of client $i$ at $t$-th global iteration is a continuous function of $\{\rho_{i}^{t} \mid t \!\in\! \\ \{0,1, \ldots, T\}, i \!\in\! \{1, 2, \ldots, N\}\}$ of all clients over time. Then, we define a continuous mapping from $\{\rho_{i}^{t} \mid t \!\in\! \{0,1,\ldots, T\}, \\ i \!\in\! \{1, 2, \ldots, N\}\}$ to $i$-th client's privacy budget $\rho_{i}^{t}$ in Eq.~(\ref{rho_iteration_with_estimators}) at $t$-th global training iteration by:
\begin{align} \label{rho_mapping}
\Theta_{i}^{t}(\{\rho_{i}^{t} \mid t \in \{0,1, \ldots, T\}, i \in \{1, 2, \ldots, N\}\}) = \rho_{i}^{t}.
\end{align} 

To summarize any possible mapping $\Theta_{i}^{t}$, we define the following vector functions as a mapping from $\{\rho_{i}^{t} \mid t \!\in\! \{0,1, \\ \ldots, T\}, i \!\in\! \{1, 2, \ldots, N\}\}$ to all clients’ privacy budget set over time for $t \!\in\! \{0,1, \ldots, T\}$:
\begin{align} \label{mapping}
& \Theta(\{\rho_{i}^{t} \mid t \in \{0,1, \ldots, T\}, i \in \{1, 2, \ldots, N\}\}) \nonumber \\
 = & \ (\Theta_{1}^{0}(\{\rho_{i}^{t} \mid t \in \{0,1, \ldots, T\}, i \in \{1, 2, \ldots, N\}\}), \ldots, \nonumber \\
& \Theta_{1}^{T}(\{\rho_{i}^{t} \mid t \in \{0,1, \ldots, T\}, i \in \{1, 2, \ldots, N\}\}), \ldots, \nonumber \\
& \Theta_{N}^{0}(\{\rho_{i}^{t} \mid t \in \{0,1, \ldots, T\}, i \in \{1, 2, \ldots, N\}\}), \ldots, \nonumber \\
& \Theta_{N}^{T}(\{\rho_{i}^{t} \mid t \in \{0,1, \ldots, T\}, i \in \{1, 2, \ldots, N\}\})).
\end{align}

Thus, the fixed point to $\Theta_{i}^{t}(\{\rho_{i}^{t} \mid t \!\in\! \{0,1, \ldots, T\}, i \!\in\! \{1, 2, \\ \ldots, N\}\}) = \Theta_{i}^{t}(\{\rho_{i}^{t} \mid t \in \{0, 1, \ldots, T\}, i \in \{1, 2, \ldots, N\}\})$ in Eq.~(\ref{mapping}) should be reached to make the mean-field term $\phi(t)$ replicate $\frac{1}{N}\sum_{i=1}^{N}\rho_{i}^{t}$. Subsequently, according to the iteration formulation of the privacy budget $\rho_{i}^{t}$ in Eq.~(\ref{update_constraint_refined}) and boundary condition in Eq.~(\ref{update_constraint_refined}). Assume that a tighter boundary condition of correction factor $\alpha_{i}^{t}$ exists, i.e., $0 < \alpha_{L} \leq \alpha_{i}^{t} \leq \alpha_{H} < 1$ with $\alpha_{L} = \min\{\alpha_{i}^{t}\}$ and $\alpha_{H} = \max\{\alpha_{i}^{t}\}$ for $t \!\in\! \{0,1, \ldots, T\}, i \!\in\! \{1, 2, \ldots, N\}$, thereby, we derive the boundary for Eq.~(\ref{rho_mapping})~by:
\begin{align} \label{mapping_boundary}
(1 \!-\! \alpha_{H})^{t} &+ \rho_{L} \! \sum_{s=1}^{t-1} \alpha_{H} (1 \!-\! \alpha_{H})^{s} \nonumber \\[-6pt]
&\leq \Theta_{i}^{t} \leq (1 \!-\! \alpha_{H})^{t} \!+\! \rho_{H} \! \sum_{s=1}^{t-1}  \alpha_{H}(1 \!-\! \alpha_{H})^{s}.
\end{align}

Then, we define a continuous space $\Omega = [\rho_{L}, \rho_{H}] \times \ldots \times[(1 \!-\! \alpha_{H})^{T} + \rho_{L} \sum_{s=1}^{T-1} \alpha_{H} (1 - \alpha_{H})^{s}, (1 - \alpha_{H})^{T} + \rho_{H} \sum_{s=1}^{T-1} \alpha_{H}(1 \!-\! \alpha_{H})^{s}]$ for $\rho_{i}^{t}$ in $T$ dimensions. From Eq.~(\ref{rho_mapping}), each mapping $\Theta_{i}^{t}$ is continuous in the continuous space $\Omega$, consequently, $\Theta$ is a continuous mapping from $\Omega$ to $\Omega$. According to Brouwer's fixed-point theorem, $\Theta$ has a fixed point in $\Omega$. 

According to Eq.~(\ref{update_constraint_refined}) and mapping $\Theta_{i}^{t}(\{\rho_{i}^{t} \mid t \in \{0,1, \ldots, \\ T\}, i \!\in\! \{1, 2, \ldots, N\}\})$ in Eq.~(\ref{mapping}), the iteration function of the privacy budget can be characterized as $\rho_{i}^{t+1} \!=\! \Theta_{i}^{t}(\rho_{i}^{t})$.~Subsequently, we define the Complete Metric Space $\Omega$ for all privacy budget parameters $\rho$ and the distance function $\Gamma$ in the metric space to quantify the divergence between two state variables can be defined as the squared $\ell_{2}$-norm, i.e., $\Gamma \!=\! \|\rho_{i} - \rho_{j}\|_{2}^{2}$. By substituting Eq.~(\ref{update_constraint}) into the distance function $\Gamma$, we have the following derivations:
\begin{align}
&\Gamma(\Theta_{i}^{t}(\rho_{i}^{t}), \Theta_{j}^{t}(\rho_{j}^{t})) = \|\rho_{i}^{t+1} - \rho_{j}^{t+1}\|_{2}^{2} \nonumber \\
&= \|(\frac{1 - \alpha_{i}^{t}}{N} \sum_{k=1}^{N} \rho_{k}^{t}  \!+\! \alpha_{i}^{t} \rho_{i}^{t}) \!-\! (\frac{1 - \alpha_{j}^{t}}{N} \sum_{k=1}^{N} \rho_{k}^{t} \!+\! \alpha_{j}^{t} \rho_{j}^{t})\|_{2}^{2} \nonumber \\
&= \| (\alpha_{i}^{t} \rho_{i}^{t} - \alpha_{j}^{t} \rho_{j}^{t}) - \frac{1}{N} (\alpha_{i}^{t} \sum_{k=1}^{N} \rho_{k}^{t}   \!-\! \alpha_{j}^{t} \sum_{k=1}^{N} \rho_{k}^{t}) \|_{2}^{2} \nonumber \\
&\leq \|\alpha_{i}^{t} \rho_{i}^{t} - \alpha_{j}^{t} \rho_{j}^{t} \|_{2}^{2} \leq \sqrt{\nu} \|\rho_{i}^{t} - \rho_{j}^{t}\|_{2}^{2},
\end{align}
where the coefficient $\sqrt{\nu}$ is chosen to satisfy the inequality $0 \!\leq\! \sqrt{\nu} \!<\! 1$ with $\sqrt{\nu} \!=\! \max\{\frac{|\alpha_{L}\rho_{L} - \alpha_{H}\rho_{H}|}{|\rho_{L} - \rho_{H}|}, \frac{|\alpha_{H}\rho_{L} - \alpha_{L}\rho_{H}|}{|\rho_{L} - \rho_{H}|}\}$ for $t \in \{0, 1, \ldots, T\}$, ensuring that the mapping $\Theta_{i}^{t}$ acts as a contraction mapping. Following Banach’s Fixed Point Theorem, we establish the existence of a unique fixed point in the iterative computational process, designated as $\rho_{i}^{t*}$. Furthermore, the inherent contraction characteristic of $\Theta_{i}^{t}$ guarantees that each iteration diminishes the disparity between successive elements in the series by a multiplicative factor of $\sqrt{\nu}$. Consequently, given any positive value $\iota$, there exists a positive integer $\mathcal{Z}$ such that for all $u, v > \mathcal{Z}$ satisfies $\Gamma(\Theta_{u}^{t}(\rho_{u}^{t}), \Theta_{v}^{t}(\rho_{v}^{t})) < \iota$,
which demonstrates that the sequence generated by $\Theta_{i}^{t}$ can be rigorously defined as a Cauchy sequence. Based on the principles of Complete Metric Spaces, we derive that the~sequence generated by $\Theta_{i}^{t}$ will converge to the unique fixed point $\rho_{i}^{t*}$ above. Thus, the theorem holds.  \end{proof}

\section{Proof of Theorem \ref{theorem_model_gap}} \label{proof_theorem_model_gap}
\begin{proof}
Before proof, we introduce the following lemma for the convenience of our theoretical analysis.
\begin{lemma} \label{lemma_Jensen_inequality}
Given $m$ vectors $x_{1}, x_{2}, \ldots, x_{m} \in \mathbb{R}^{d}$, based on Jensen's inequality, we have:
\begin{align}
\| \sum_{i=1}^{m} x_{i}\|^{2} \leq m \sum_{i=1}^{m}\|x_{i}\|^{2}.
\end{align}
\end{lemma}

According to Eqs. (\ref{local_gradient_perturbation})-(\ref{sgd_sample}), the stochastic gradient descent of the global model update rule considering the specific probability-based client sampling model and artificial Gaussian noise is given by:
\begin{align} \label{sgd_sample_noise}
\boldsymbol{w(t+1)}=\boldsymbol{w(t)} - \eta \sum_{i=1}^{K} \frac{\theta_{i}}{K x_{i}^{t}} \nabla \widetilde{F}_{i}(\boldsymbol{w(t)})
\end{align}

Consequently, based on Eq.~(\ref{sgd_sample_noise}), the squared $\ell_{2}$-norm of the difference between the global and the optimal model parameter is calculated as follows:
\begin{align} \label{model_gap}
\|\boldsymbol{w(t \!+\! 1)} &\!-\! \boldsymbol{w^{*}}\|_{2}^{2} \!=\! \|\boldsymbol{w(t)} \!-\! \boldsymbol{w^{*}}\|_{2}^{2} \!+\! \underbrace{\|\eta \! \sum_{i=1}^{K} \! \frac{\theta_{i}}{K x_{i}^{t}} \nabla \! \widetilde{F}_{i}(\boldsymbol{w(t)})\|_{2}^{2}}_{T_{1}} \nonumber \\
&  \underbrace{- 2 \langle\boldsymbol{w(t)} \!-\! \boldsymbol{w^{*}}, \eta  \sum_{i=1}^{K} \frac{\theta_{i}}{K x_{i}^{t}} \nabla \widetilde{F}_{i}(\boldsymbol{w(t)}) \rangle}_{T_{2}}.
\end{align}

Note that according to the boundary condition of the privacy budget $\rho_{i}^{t} \!\in\! [\rho_{L}, \rho_{H}]$ the sampling probability $x_{i}^{t}$ satisfies
\begin{align} \label{probability_boundary}
\frac{\rho_{L}}{(N-1)\rho_{H} + \rho_{L}} \leq x_{i}^{t} \leq \frac{\rho_{H}}{(N-1)\rho_{L} + \rho_{H}}.
\end{align}

Then, we bound the two terms on the right-hand side of the equality respectively. For the term $T_{1}$, we take the expectation on both sides of the equation and the derivation is given by:
\begin{align} \label{T1}
\mathbb{E}[T_{1}] &= \mathbb{E}[\|\eta \sum_{i=1}^{K} \frac{\theta_{i}}{K x_{i}^{t}} \nabla \widetilde{F}_{i}(\boldsymbol{w(t)})\|_{2}^{2}] \nonumber \\
&= \mathbb{E}[\|\eta \sum_{i=1}^{K} \frac{\theta_{i}}{K x_{i}^{t}} (\nabla F_{i}(\boldsymbol{w(t)}) + \boldsymbol{n_{i}(t)})\|_{2}^{2}] \nonumber \\
&= \mathbb{E}[\|\eta \sum_{i=1}^{K} \! \frac{\theta_{i}}{K x_{i}^{t}} \nabla F_{i}(\boldsymbol{w(t)}) \|_{2}^{2}] + \mathbb{E}[\|\eta \sum_{i=1}^{K} \! \frac{\theta_{i} \boldsymbol{n_{i}(t)}}{K x_{i}^{t}} \|_{2}^{2}] \nonumber \\
&\stackrel{(a)}{\leq} \frac{\eta^{2} [(N \!-\! 1)\rho_{H} \!+\! \rho_{L}]^{2}}{K^{2} \rho_{L}^{2}} (D^{2} \!+\! \frac{2 d K W^{2}}{\rho_{L} (\sum_{i=1}^{N} |\mathcal{D}_{i}|)^{2}}),
\end{align}
where (a) follows Jensen's inequality in Lemma \ref{lemma_Jensen_inequality}, Assumption~\ref{assumption_bounded_gradients}, Proposition \ref{proposition_variance} and Eq.~(\ref{probability_boundary}). For the term $T_{2}$ in Eq.~(\ref{model_gap}), taking the expectation on both sides of the function, we have the following derivation:
\begin{align} \label{T2}
\!\!\!\!\mathbb{E}[T_{2}] \!=& \ \mathbb{E}[- 2 \langle\boldsymbol{w(t)} \!-\! \boldsymbol{w^{*}}, \eta  \sum_{i=1}^{K} \frac{\theta_{i}}{K x_{i}^{t}} \nabla \widetilde{F}_{i}(\boldsymbol{w(t)}) \rangle] \nonumber \\
=&\ \mathbb{E}[- 2 \langle\boldsymbol{w(t)} \!-\! \boldsymbol{w^{*}}, \eta  \sum_{i=1}^{K} \frac{\theta_{i}}{K x_{i}^{t}} (\nabla F_{i}(\boldsymbol{w(t)}) \!+\! \boldsymbol{n_{i}(t)}) \rangle] \nonumber \\
\stackrel{(a)}{\leq}&\ \mathbb{E}[- 2 \langle\boldsymbol{w(t)} \!-\! \boldsymbol{w^{*}}, \eta  \sum_{i=1}^{K} \frac{\theta_{i}}{K x_{i}^{t}} \nabla F_{i}(\boldsymbol{w(t)}) \rangle] \nonumber \\
\leq&\ \text{\small $- \frac{2 \eta [(N \!-\! 1)\rho_{H} \!+\! \rho_{L}]}{K \rho_{L}} \sum_{i=1}^{K} \! \theta_{i}  \mathbb{E}[\langle\boldsymbol{w(t)} \!-\! \boldsymbol{w^{*}}, \nabla F_{i}(\boldsymbol{w(t)}) \rangle] $}  \nonumber \\
\stackrel{(b)}{\leq}& -\frac{2 \eta [(N \!-\! 1)\rho_{H} \!+\! \rho_{L}]}{K \rho_{L}}(F(\boldsymbol{w(t)}) \!-\! F(\boldsymbol{w^{*}})) \nonumber \\
&\!+ \!\frac{\beta \eta [(N \!\!-\! 1)\rho_{H} \!+\! \rho_{L}]}{K \rho_{L}}\mathbb{E}[\sum_{i=1}^{K} \! \theta_{i} \|\boldsymbol{w_{i}(t)} \!-\! \boldsymbol{w(t)}\|_{2}^{2}], \!\!\!
\end{align}
where inequality (a) holds according to the equation below, in which the artificial Gaussian noise satisfies $\boldsymbol{n_{i}(t)}  \! \sim \! \mathcal{N}(0, \sigma_{i}^{2}(t))$ and $\mathbb{E}[(\boldsymbol{n_{i}(t)}] = 0$, thereby, we have:
\begin{align}
\mathbb{E}[\langle\boldsymbol{w(t)} \!-\! \boldsymbol{w^{*}}, \boldsymbol{n_{i}(t)} \rangle] &= \mathbb{E}[(\boldsymbol{w(t)} \!-\! \boldsymbol{w^{*}})^\top \boldsymbol{n_{i}(t)}] \nonumber \\
&= (\boldsymbol{w(t)} \!-\! \boldsymbol{w^{*}})^\top \mathbb{E}[(\boldsymbol{n_{i}(t)}] \!=\! 0,
\end{align}
and (b) follows the fact that:
\begin{align}
&- \sum_{i=1}^{K} \theta_{i} \langle\boldsymbol{w(t)} - \boldsymbol{w^{*}}, \nabla F_{i}(\boldsymbol{w(t)}) \rangle \nonumber \\
=& \small - \!\! \sum_{i=1}^{K} \! \theta_{i} [\langle\boldsymbol{w(t)} \!-\! \boldsymbol{w_{i}(t)}, \!\nabla \! F_{i}(\boldsymbol{w(t)}) \rangle \!+\! \langle \boldsymbol{w_{i}(t)} \!-\! \boldsymbol{w^{*}} , \! \nabla \! F_{i}(\boldsymbol{w(t)})\rangle] \nonumber \\
\leq& - \sum_{i=1}^{K} \theta_{i} [F_{i}(\boldsymbol{w(t)}) - F_{i}(\boldsymbol{w_{i}(t)}) - \frac{\beta}{2} \|\boldsymbol{w_{i}(t)} - \boldsymbol{w(t)}\|_{2}^{2} \nonumber \\
&+ \langle \boldsymbol{w_{i}(t)} - \boldsymbol{w^{*}} , \nabla F_{i}(\boldsymbol{w(t)})\rangle] \nonumber \\
\leq& - \sum_{i=1}^{K} \theta_{i} [F_{i}(\boldsymbol{w(t)}) - F_{i}(\boldsymbol{w_{i}(t)}) + F_{i}(\boldsymbol{w_{i}(t)}) - F_{i}(\boldsymbol{w^{*}})] \nonumber \\
&+  \frac{\beta}{2} \sum_{i=1}^{K} \theta_{i} \|\boldsymbol{w_{i}(t)} - \boldsymbol{w(t)}\|_{2}^{2} \nonumber \\
\leq& - (F(\boldsymbol{w(t)}) \!-\! F(\boldsymbol{w^{*}})) \!+\! \frac{\beta}{2} \underbrace{\sum_{i=1}^{K} \theta_{i} \|\boldsymbol{w_{i}(t)} \!-\! \boldsymbol{w(t)}\|_{2}^{2}}_{T_{3}}.
\end{align}

Next, we bound the expectation of term $T_{3}$. Based on Eq.~(\ref{sgd}) and Eq.~(\ref{sgd_sample_noise}), the upper bound of expected $T_{3}$ is given~by:
\begin{align} \label{T3}
\mathbb{E}[T_{3}] &\!= \mathbb{E}[\sum_{i=1}^{K} \theta_{i} \|\boldsymbol{w_{i}(t)} \!-\! \boldsymbol{w(t)}\|_{2}^{2}] \nonumber \\
&= \mathbb{E}[\sum_{i=1}^{K} \theta_{i} \|\boldsymbol{w(t \!-\! 1)} \!-\! \eta \nabla F_{i}(\boldsymbol{w(t \!-\! 1)}) \!-\! \boldsymbol{w(t)}\|_{2}^{2}] \nonumber \\
&\leq \text{\small $\mathbb{E}[\sum_{i=1}^{K}\! \theta_{i} (\|\boldsymbol{w(t \!-\! 1)} \!-\! \boldsymbol{w(t)}\|_{2}^{2} \!+\! \|\eta \nabla F_{i}(\boldsymbol{w(t \!-\! 1)})\|_{2}^{2})] $} \nonumber \\
&\leq \mathbb{E}[\sum_{i=1}^{K} \theta_{i} \|\eta \sum_{i=1}^{K} \frac{\theta_{i}}{K x_{i}^{t}} \nabla \widetilde{F}_{i}(\boldsymbol{w(t)})\|_{2}^{2}] + \eta^{2} D^{2} \nonumber \\
&\stackrel{(a)}{\leq} \!\! \text{\small $\frac{\eta^{2} [(N \!\!-\! 1)\rho_{H} \!+\! \rho_{L}]^{2}}{K^{2} \rho_{L}^{2}} (D^{2} \!\!+\! \frac{2 d K W^{2}}{\rho_{L} (\sum_{i=1}^{N}\! |\mathcal{D}_{i}|)^{2}\!}) \!+\! \eta^{2} \! D^{2}\!\! $}, \!
\end{align}
where (a) follows the Jensen's inequality in Lemma \ref{lemma_Jensen_inequality}. Subsequently, by combining the upper bound of terms $\mathbb{E}[T_{1}]$, $\mathbb{E}[T_{2}]$ and $\mathbb{E}[T_{3}]$ in Eqs.~(\ref{T1})-(\ref{T2}) and Eq.~(\ref{T3}), the expected optimality gap $\mathbb{E}[\|\boldsymbol{w(t \!+\! 1)} \!-\! \boldsymbol{w^{*}}\|_{2}^{2}]$ is concluded in Eq.~(\ref{optimal_gap_middle}). 

Then, we analyze the upper bound of the expected optimality gap of the global model under convexity and non-convexity scenarios. For the convexity scenario in Assumption \ref{assumption_convexity}, Eq.~(\ref{optimal_gap_middle}) can be further derived as Eq.~(\ref{optimal_gap_convexity_former}). Through iteratively aggregating both sides of Eq.~(\ref{optimal_gap_convexity_former}), we can obtain the upper bound of the expected optimality gap after $T$-iteration of the global model training process as shown in Eq.~(\ref{optimal_gap_convexity}). 

\begin{figure*}[ht]
\centering

\begin{minipage}{1.0\textwidth}
\begin{align} \label{optimal_gap_middle}
\!\mathbb{E}[\|\boldsymbol{w(t \!+\! 1)} \!-\! \boldsymbol{w^{*}}\|_{2}^{2}] \!\leq&\ \mathbb{E}[\|\boldsymbol{w(t)} \!-\! \boldsymbol{w^{*}}\|_{2}^{2}] \!-\! \frac{2 \eta [(N \!-\! 1)\rho_{H} \!+\! \rho_{L}]}{K \rho_{L}}(F(\boldsymbol{w(t)}) \!-\! F(\boldsymbol{w^{*}})) \!+\! \frac{\beta D^{2} \eta^{3} [(N \!-\! 1)\rho_{H} \!+\! \rho_{L}]}{K \rho_{L}} \nonumber \\
&\!+ \frac{\eta^{2} [(N \!-\! 1)\rho_{H} \!+\! \rho_{L}]^{2}}{K^{2} \rho_{L}^{2}} (D^{2} \!+\! \frac{2 d K W^{2}}{\rho_{L} (\sum_{i=1}^{N} |\mathcal{D}_{i}|)^{2}}) \!+\! \frac{\beta \eta^{3} [(N \!-\! 1)\rho_{H} \!+\! \rho_{L}]^{3}}{K^{3} \rho_{L}^{3}} (D^{2} \!+\! \frac{2 d K W^{2}}{\rho_{L} (\sum_{i=1}^{N} |\mathcal{D}_{i}|)^{2}}) . 
\end{align}
For Convexity Assumption:
\begin{align} \label{optimal_gap_convexity_former}
\!\mathbb{E}[\|\boldsymbol{w(t \!+\! 1)} \!-\! \boldsymbol{w^{*}}\|_{2}^{2}] \!\leq&\ (1 \!-\! \frac{\psi \eta [(N \!-\! 1)\rho_{H} \!+\! \rho_{L}]}{K \rho_{L}})\mathbb{E}[\|\boldsymbol{w(t)} \!-\! \boldsymbol{w^{*}}\|_{2}^{2}] \!+\! \frac{\beta D^{2} \eta^{3} [(N \!-\! 1)\rho_{H} \!+\! \rho_{L}]}{K \rho_{L}} \nonumber \\
&\!+ \frac{\eta^{2} [(N \!-\! 1)\rho_{H} \!+\! \rho_{L}]^{2}}{K^{2} \rho_{L}^{2}} (D^{2} \!+\! \frac{2 d K W^{2}}{\rho_{L} (\sum_{i=1}^{N} |\mathcal{D}_{i}|)^{2}}) \!+\! \frac{\beta \eta^{3} [(N \!-\! 1)\rho_{H} \!+\! \rho_{L}]^{3}}{K^{3} \rho_{L}^{3}} (D^{2} \!+\! \frac{2 d K W^{2}}{\rho_{L} (\sum_{i=1}^{N} |\mathcal{D}_{i}|)^{2}}) . 
\end{align}
For Non-Convexity Assumption:
\begin{align} \label{optimal_gap_non_convexity_former}
\!\mathbb{E}[\|\boldsymbol{w(t \!+\! 1)} \!-\! \boldsymbol{w^{*}}\|_{2}^{2}] \!\leq&\ \mathbb{E}[\|\boldsymbol{w(t)} \!-\! \boldsymbol{w^{*}}\|_{2}^{2}] \!-\! \frac{2 \eta \|g^\top\|_{2} [(N \!-\! 1)\rho_{H} \!+\! \rho_{L}]}{K \rho_{L}}\mathbb{E}[\|\boldsymbol{w(t)} \!-\! \boldsymbol{w^{*}}\|_{2}] \!+\! \frac{\beta D^{2} \eta^{3} [(N \!-\! 1)\rho_{H} \!+\! \rho_{L}]}{K \rho_{L}} \nonumber \\
&\!+ \frac{\eta^{2} [(N \!-\! 1)\rho_{H} \!+\! \rho_{L}]^{2}}{K^{2} \rho_{L}^{2}} (D^{2} \!+\! \frac{2 d K W^{2}}{\rho_{L} (\sum_{i=1}^{N} |\mathcal{D}_{i}|)^{2}}) \!+\! \frac{\beta \eta^{3} [(N \!-\! 1)\rho_{H} \!+\! \rho_{L}]^{3}}{K^{3} \rho_{L}^{3}} (D^{2} \!+\! \frac{2 d K W^{2}}{\rho_{L} (\sum_{i=1}^{N} |\mathcal{D}_{i}|)^{2}}) \nonumber \\
\leq&\ \mathbb{E}[\|\boldsymbol{w(t)} \!-\! \boldsymbol{w^{*}}\|_{2}^{2}] \!+\! \frac{2 t \eta^{2} \|g^\top \!\|_{2} D [(N \!\!-\! 1)\rho_{H} \!+\! \rho_{L}]^{2}}{K^{2} \rho_{L}^{2}} \!+\! \frac{2 \eta \|g^\top \!\|_{2} [(N \!\!-\! 1)\rho_{H} \!+\! \rho_{L}]}{K \rho_{L}}\mathbb{E}[\|\boldsymbol{w(0)} \!-\! \boldsymbol{w^{*}}\|_{2}]   \nonumber \\
& \!+  \text{\small $\frac{\beta D^{2} \eta^{3} [(N \!\!-\! 1)\rho_{H} \!+\! \rho_{L}]}{K \rho_{L}} \!+\! (\frac{\eta^{2} [(N \!-\! 1)\rho_{H} \!+\! \rho_{L}]^{2}}{K^{2} \rho_{L}^{2}} \!+\! \frac{\beta \eta^{3} [(N \!-\! 1)\rho_{H} \!+\! \rho_{L}]^{3}}{K^{3} \rho_{L}^{3}}) (D^{2} \!+\! \frac{2 d K W^{2}}{\rho_{L} (\sum_{i=1}^{N} |\mathcal{D}_{i}|)^{2}}) $}.
\end{align}
\end{minipage}

\vspace{5pt}
\noindent\rule{\textwidth}{0.5pt}
\vspace{-25pt}
\end{figure*}

Subsequently, we consider the upper bound of Eq.~(\ref{optimal_gap_middle}) under the non-convexity scenario in Assumption \ref{assumption_subgradient}. We first introduce the following proposition to facilitate our demonstration.
\begin{proposition} \label{proposition_global_model_optimal_gap}
Suppose that the local gradient $\nabla F_{i}(\boldsymbol{w})$ attains $D$-bound gradients, the expected $\ell_{2}$-norm of the difference between the global and optimal model is upper bounded by
\begin{align} \label{global_model_optimal_gap}
\!\!\!\! \small \mathbb{E}[\|\boldsymbol{w(t)} \!-\! \boldsymbol{w^{*}}\|_{2}] \!\leq\! \mathbb{E}[\|\boldsymbol{w(0)} \!-\! \boldsymbol{w^{*}}\|_{2}] \!+\! \frac{t \eta D [(N \!\!-\! 1)\rho_{H} \!+\! \rho_{L}]}{K \rho_{L}}. \!
\end{align}
\end{proposition}
\begin{proof}
For $t$-th global iteration, $t \in \{0, 1, \ldots, T\}$, based on the refined global model update rule regarding the client sampling probability vector and artificial Gaussian noise in Eq.~(\ref{sgd_sample_noise}), we have the following derivation:
\begin{align} \label{global_model_gap}
&\mathbb{E}[\|\boldsymbol{w(t \!+\! 1)} \!-\! \boldsymbol{w^{*}}\|_{2}] \nonumber \\
=& \ \mathbb{E}[\|\boldsymbol{w(t)} \!-\! \boldsymbol{w^{*}} \!-\! \eta \sum_{i=1}^{K} \frac{\theta_{i}}{K x_{i}^{t}} \nabla \widetilde{F}_{i}(\boldsymbol{w(t)})\|_{2}] \nonumber \\
\leq& \ \mathbb{E}[\|\boldsymbol{w(t)} \!-\! \boldsymbol{w^{*}}\|_{2}] \!+\! \mathbb{E}[\|\eta \sum_{i=1}^{K} \frac{\theta_{i}}{K x_{i}^{t}} \nabla \widetilde{F}_{i}(\boldsymbol{w(t)})\|_{2}] \nonumber \\
\stackrel{(a)}{\leq}& \ \mathbb{E}[\|\boldsymbol{w(t)} \!-\! \boldsymbol{w^{*}}\|_{2}] \!+\! \frac{\eta D [(N \!-\! 1)\rho_{H} \!+\! \rho_{L}]}{K \rho_{L}},
\end{align}
where (a) follows the Cauchy-Schwarz inequality and Assumption~\ref{assumption_bounded_gradients}. Through iteratively aggregating both sides of inequality in Eq.~(\ref{global_model_gap}), we can derive the expression as shown in Eq.~(\ref{global_model_optimal_gap}), and thus, the proposition holds.  \end{proof}

Thus, according to Proposition \ref{proposition_global_model_optimal_gap} and Assumption \ref{assumption_subgradient}, Eq.~(\ref{optimal_gap_middle}) can be further derived as Eq.~(\ref{optimal_gap_non_convexity_former}).Through iteratively aggregating both sides of Eq.~(\ref{optimal_gap_non_convexity_former}), we can obtain the upper bound of the expected optimality gap of the global model under the non-convexity scenario as showcased in Eq.~(\ref{optimal_gap_non_convexity}). Hence, the theorem holds.   \end{proof}

\section{Proof of Proposition \ref{proposition_difference}} \label{proof_proposition_difference}
\begin{proof}
We first analyze the situation with artificial Gaussian noise under full client participation. According to Assumption \ref{assumption_smoothness}, we have the following inequality regarding the difference between the global loss function for any consecutive iterations:
\begin{align}\label{global_loss_function_difference}
F(\boldsymbol{w(t)}) \!&-\! F(\boldsymbol{w(t \!-\! 1)}) \!\leq\! \frac{\beta \eta^{2}}{2} \|\sum_{i=1}^{N} \theta_{i} \nabla \widetilde{F}_{i}(\boldsymbol{w(t \!-\! 1)}))\|_{2}^{2} \nonumber \\
&+ \! \nabla F(\boldsymbol{w(t \!-\! 1)})^\top (-\eta \! \sum_{i=1}^{N} \theta_{i} \nabla \widetilde{F}_{i}(\boldsymbol{w(t \!-\! 1)})).
\end{align}

Then, taking expectations on both sides of Eq.~(\ref{global_loss_function_difference}), Eq.~(\ref{global_loss_function_difference}) can be rewritten as follows:
\begin{align} \label{expectations_difference}
&\mathbb{E}[F(\boldsymbol{w(t)})] \!-\! F(\boldsymbol{w(t \!-\! 1)}) \nonumber \\
\leq& \small - \eta \nabla\! F(\boldsymbol{w(t \!-\! 1)})^\top \!({\sum_{i=1}^{N} \! \theta_{i} \mathbb{E}[\nabla \! F_{i}(\boldsymbol{w(t \!-\! 1)}) \!+\! \theta_{i} \boldsymbol{n_{i}(t \!-\! 1)}]}) \nonumber \\
&+ \small \frac{\beta \eta^{2}}{2} ({\mathbb{E}[\|\!\sum_{i=1}^{N}\! \theta_{i} \nabla \! F_{i}(\boldsymbol{w(t \!-\! 1)})\|_{2}^{2}] \!+\! \mathbb{E}[\|\!\sum_{i=1}^{N}\! \theta_{i} \boldsymbol{n_{i}(t \!-\! 1)}\|_{2}^{2}]}) \nonumber \\
&+ \small  \beta \eta^{2} {\mathbb{E}[\sum_{i=1}^{N}\!\theta_{i} \boldsymbol{n_{i}(t \!-\! 1)} \! \sum_{i=1}^{N} \theta_{i} \nabla F_{i}(\boldsymbol{w(t \!-\! 1)})]} \nonumber \\
\leq& - \eta \nabla F(\boldsymbol{w(t \!-\! 1)})^\top ({\sum_{i=1}^{N} \theta_{i} \mathbb{E}[ \nabla F_{i}(\boldsymbol{w(t \!-\! 1)})]}) \nonumber \\
&+ \frac{\beta \eta^{2}}{2} (\mathbb{E}[\|\sum_{i=1}^{N} \! \theta_{i} \nabla F_{i}(\boldsymbol{w(t\!-\!1)})\|_{2}^{2}]\!+\! d \sum_{i=1}^{N} \theta_{i}^{2} \sigma_{i}^{2}(t \!-\! 1)) \nonumber \\
\stackrel{(a)}{\leq}&  -\eta \kappa   \|\nabla F(\boldsymbol{w(t \!-\! 1)})\|_{2}^{2} + \frac{d \beta \eta^{2}}{2}  \sum_{i=1}^{N} \theta_{i}^{2} \sigma_{i}^{2}(t-1) \nonumber \\
&+ \frac{\beta \eta^{2} N}{2} (M +(M_{V}+\kappa_{G}^{2})\|\nabla F(\boldsymbol{w(t-1)})\|_{2}^{2}) \nonumber \\
\leq& \small - \! \frac{\eta \kappa}{2} \|\nabla \! F(\boldsymbol{w(t \!-\! 1)})\|_{2}^{2} \!+\! \frac{d \beta \eta^{2}}{2} \!\! \sum_{i=1}^{N} \! \theta_{i}^{2} \! \sigma_{i}^{2}(t \!-\! 1) \!+\! \frac{M \! \beta \eta^{2} \! N}{2} , \!
\end{align}
where (a) follows the fact regarding the squared $\ell_{2}$-norm of the weighted local gradients of all clients: 
\begin{align}
&\mathbb{E}[\|{\sum_{i=1}^{N}} \theta_{i} \nabla F_{i}(\boldsymbol{w(t)})\|_{2}^{2}] \nonumber \\
\leq& \ N \mathbb{E}[\sum_{i=1}^{N} \|\theta_{i} \nabla F_{i}(\boldsymbol{w(t)})\|_{2}^{2}] = N \sum_{i=1}^{N} \theta_{i}^{2} \mathbb{E}[\|\nabla F_{i}(\boldsymbol{w(t)})\|_{2}^{2}] \nonumber \\
\stackrel{(a)}{=}& \ N \sum_{i=1}^{N} \theta_{i}^{2} (\mathbb{V}[\nabla{F_{i}(\boldsymbol{w(t)})}] + \|\mathbb{E}[\nabla{F_{i}(\boldsymbol{w(t)})}]\|_{2}^{2}) \nonumber \\
\stackrel{(b)}{\leq}& \  NM \!+\! N(M_{V} \!+\! \kappa_{G}^{2})\|\nabla F(\boldsymbol{w(t)})\|_{2}^{2}
\end{align}
where (a) follows the definition of the variance of $ \nabla{F_{i}(\boldsymbol{w(t)})}$ in Assumption~\ref{assumption_moment} and (b) follows the inequalities in Assumption~\ref{assumption_moment} (\romannumeral1) and (\romannumeral2). 

Then, as to the partial client participation scenarios with dynamic client sampling probability vector $\boldsymbol{X_{t}} \!=\! \{x_{i}^{t}, i \!\in\! \{1, 2, \\ \ldots, N\}\}$, based on the Lemma 1 in \cite{luo2024adaptive}, we can derive that, at $t$-th global iteration, the sampling variance between the actual aggregated global model $\boldsymbol{w(t)}$ and the virtual global model $\boldsymbol{\hat{w}(t)}$ is bounded by:
\begin{align}
&s^{2}(t) = \mathbb{E} \| \boldsymbol{w(t)} - \boldsymbol{\hat{w}(t)} \|_{2}^{2} \nonumber \\
&\leq \frac{\eta}{K} \sum_{i=1}^{N} \frac{\theta_{i}^{2}}{x_{i}^{t}} \mathbb{E}\|\widetilde{F}_{i}(\boldsymbol{w(t)})\|_{2}^{2} \!\leq\! \frac{\eta}{K}  \sum_{i=1}^{N} \frac{\theta_{i}^{2}}{x_{i}^{t}} (D^{2} \!+\! d \sigma_{i}^{2}(t)). 
\end{align}

Based on Assumption \ref{assumption_smoothness}, the upper bound of $\mathbb{E}[F(\boldsymbol{w_{p^{t}}(t)})] - \mathbb{E}[F(\boldsymbol{\hat{w}(t)})]$ is given as follows:
\begin{align} \label{global_loss_function_sampling_variance}
\mathbb{E}[F(\boldsymbol{w_{p^{t}}(t)})] &- \mathbb{E}[F(\boldsymbol{\hat{w}(t)})] \leq \frac{\beta}{2} \mathbb{E} \| \boldsymbol{w(t)} - \boldsymbol{\hat{w}(t)} \|_{2}^{2} \nonumber \\
&\leq \frac{\beta \eta}{2 K} \sum_{i=1}^{N} \frac{\theta_{i}^{2}}{x_{i}^{t}} (D^{2} \!+\! d \sigma_{i}^{2}(t)). 
\end{align}

Inspired by \cite{luo2022tackling,luo2024adaptive}, from the upper bound of the difference between the global loss function for any consecutive iterations under the full client participation in Eq.~(\ref{expectations_difference}) and the sampling variance of the global loss function in Eq.~(\ref{global_loss_function_sampling_variance}), the expected divergence of the global loss function with the partial client participation for iteration $t$ and $t+1$ can be derived as displayed in Eq.~(\ref{difference}). Consequently, the theorem holds. \end{proof}

\section{Proof of Theorem \ref{theorem_dp_convergence}} \label{proof_theorem_dp_convergence}
\begin{proof}
According to the expected difference between the global loss function for any consecutive iterations under partial client participation scenarios in Eq.~(\ref{difference}), we have:
\begin{align}
&\mathbb{E}[F(\boldsymbol{w(t)})] \!-\! F(\boldsymbol{w(t \!-\! 1)}) \nonumber \\
\leq&  \!-\! \frac{\eta \kappa}{2} \|\nabla \! F(\boldsymbol{w(t \!-\! 1)})\|_{2}^{2} \!+\! \frac{d \beta \eta^{2}}{2} \sum_{i=1}^{N} \theta_{i}^{2} \! \sigma_{i}^{2}(t \!-\! 1) \nonumber \\
&+ \frac{\beta \eta}{2 K} \sum_{i=1}^{N} \frac{\theta_{i}^{2}}{x_{i}^{t}} (D^{2} \!+\! d \sigma_{i}^{2}(t \!-\! 1)) \!+\! \frac{M \beta \eta^{2} N}{2}  \label{consecutive_upper_bound} \\ 
\stackrel{(a)}{\leq}&  \!-\! \mu \eta \kappa (F(\boldsymbol{w(t \!-\! 1)}) \!-\! F(\boldsymbol{w^{*}})) \!+\! \frac{d \beta \eta^{2}}{2} \sum_{i=1}^{N} \theta_{i}^{2} \! \sigma_{i}^{2}(t \!-\! 1) \nonumber \\
&+ \frac{\beta \eta}{2 K} \sum_{i=1}^{N} \frac{\theta_{i}^{2}}{x_{i}^{t}} (D^{2} \!+\! d \sigma_{i}^{2}(t \!-\! 1)) \!+\! \frac{M \beta \eta^{2} N}{2}, \label{expectations_difference_partial}
\end{align}
where (a) follows Polyak-\L{}ojasiewicz inequality in Lemma~\ref{lemma_PL_inequality}. Then, taking expectations and subtracting the optimal global loss function $F(\boldsymbol{w^{*}})$ on both sides of Eq.~(\ref{expectations_difference_partial}), we have:
\begin{align} \label{convergence_rate_former}
&\mathbb{E}[F(\boldsymbol{w(t)})] \!-\! F(\boldsymbol{w^{*}}) \nonumber \\
&\leq (1 - \mu \eta \kappa) (\mathbb{E}[F(\boldsymbol{w(t-1)})] \!-\! F(\boldsymbol{w^{*}})) \!+\! \frac{M \beta \eta^{2} N}{2} \nonumber \\
&+ \frac{d \beta \eta^{2}}{2} \sum_{i=1}^{N} \theta_{i}^{2} \! \sigma_{i}^{2}(t \!-\! 1) \!+\! \frac{\beta \eta}{2 K} \sum_{i=1}^{N} \frac{\theta_{i}^{2}}{x_{i}^{t}} (D^{2} \!+\! d \sigma_{i}^{2}(t \!-\! 1)). 
\end{align}

Through iteratively aggregating Eq.~(\ref{convergence_rate_former}) with global iteration $t \!\in\! \{0, 1, \ldots, T\}$, we can obtain the upper bound of convergence rate as proposed in Eq.~(\ref{convergence_rate}). Then, taking expectation on both sides of Eq.~(\ref{consecutive_upper_bound}) and rearranging it, we have the following derivations regarding iterations $t$ and $t+1$:
\begin{align} \label{convergence_error_former}
&\!\!\!\!\! \text{ \small$\mathbb{E}[\|\nabla \! F(\boldsymbol{w(t)})\|_{2}^{2}] \leq \frac{2}{\eta \kappa}(\mathbb{E}[F(\boldsymbol{w(t)})] \!-\! \mathbb{E}[F(\boldsymbol{w(t \!+\! 1)})]) $} \nonumber \\
&\!\!\!\!\! + \text{\small $\frac{\beta \eta}{2 K} \sum_{i=1}^{N} \frac{\theta_{i}^{2}}{x_{i}^{t}} (D^{2} \!+\! d \sigma_{i}^{2}(t)) \!+\! \frac{d \beta \eta^{2}}{2} \sum_{i=1}^{N} \theta_{i}^{2} \! \sigma_{i}^{2}(t) \!+\! \frac{M \beta \eta^{2} N}{2}$}.
\end{align}

Averaging Eq.~(\ref{convergence_error_former}) under the global iteration $t \in \{0,1,\ldots, \\ T-1\}$, we obtain the convergence error of the FL model as shown in Eq.~(\ref{convergence_error}). Hence, the theorem holds. \end{proof}

\end{document}